\documentclass[10pt]{article}
\usepackage{paper_style}
\usepackage{physics}
\usepackage{xcolor}
\usepackage{dsfont}
\usepackage{placeins}
\usepackage[title]{appendix}
\usepackage{cite}
\usepackage{authblk}

\title{Rate-Limited Quantum-to-Classical Optimal Transport in Finite and Continuous-Variable Quantum Systems}
\author[1]{Hafez M. Garmaroudi}
\author[2]{S. Sandeep Pradhan}
\author[3]{Jun Chen}
\affil[1,3]{\small Department of Electrical and Computer Engineering, McMaster University\\
                    1280 Main St W, Hamilton, ON L8S 4L8, Canada. 
                    Email: \{mousas15, chenjun\}@mcmaster.ca}
\affil[2]{Department of Electrical Engineering and Computer Science, University of Michigan\\ 1301 Beal Avenue, Ann Arbor, MI 48109, USA.
                     Email: pradhanv@umich.edu }

\date{\today}

\begin{document}
\maketitle

% %%%%%%%%%%%%%%%%%%%%%%%%%%%%%%%%%%%%%%%%%%%%%%%%%%%%%%%%%%
% ABSTRACT
% %%%%%%%%%%%%%%%%%%%%%%%%%%%%%%%%%%%%%%%%%%%%%%%%%%%%%%%%%%

{\setstretch{.8}
% %%%%%%%%%%%%%%%%%%
\begin{abstract}

% {\color{red}
% We establish a coding theorem for rate-limited quantum-classical optimal transport systems with limited classical common randomness.
% The coding theorem, referred to as \textit{output-constrained rate-distortion theorem}, characterizes the rate region of measurement protocols on a product quantum source state for faithful construction of a given classical destination distribution while maintaining the source-destination distortion below a prescribed threshold with respect to a general distortion observable.
% This theorem provides a solution to the problem of rate-limited optimal transport, which aims to find the optimal cost of transforming a source quantum state to a destination distribution via a measurement channel with a limited classical communication rate. The coding theorem is further extended to cover Bosonic continuous-variable quantum systems.  The analytical evaluation is provided for the case of a  qubit measurement system with unlimited common randomness, as well as the case of Gaussian quantum systems.}

% Due to the existence of information bottleneck in systems, an optimal transportation mapping 

We consider the rate-limited quantum-to-classical optimal transport in terms of output-constrained rate-distortion coding for both finite-dimensional and continuous-variable quantum-to-classical systems with limited classical common randomness. The main coding theorem provides a single-letter characterization of the achievable rate region of a lossy quantum measurement source coding for an exact construction of the destination distribution (or the equivalent quantum state) while maintaining a threshold of distortion from the source state according to a generally defined distortion observable. The constraint on the output space fixes the output distribution to an IID predefined probability mass function. Therefore, this problem can also be viewed as information-constrained optimal transport which finds the optimal cost of transporting the source quantum state to the destination classical distribution via a quantum measurement with limited communication rate and common randomness.

% In an optimal transport problem, the optimal mapping might not always be feasible due to some practical limitations. For example, the existence of an information bottleneck between source and destination forces one to resort to some suboptimal mappings. Investigating this type of constrained optimal transport problem is clearly of both theoretical significance and practical interest. 

% In this work, we consider the problem of \textit{rate-limited optimal transport} in the context of quantum-to-classical systems by establishing an \textit{output-constrained rate-distortion coding theorem}. 
% This theorem develops a collective $n$-fold measurement followed by a noiseless communication channel and finds the least required transmission rate $R$ and common randomness $R_c$ to transport sufficiently large block of $n$ product states of a memoryless quantum source, to samples forming a perfectly IID classical destination distribution, while maintaining the distortion between them. 
% The coding theorem provides operational meanings to the problem of \textit{Rate-Limited Optimal Transport}, which finds the optimal transportation from source to destination subject to the rate constraints. We further provide an analytical evaluation of the quantum-to-classical rate-limited optimal transportation cost for the case of a qubit-Bernoulli system with unlimited common randomness. 
% % The evaluation results in a transcendental system of equations whose solution provides the rate-distortion curve of the transportation protocol.

We develop a coding framework for continuous-variable quantum systems by employing a clipping projection and a dequantization block and using our finite-dimensional coding theorem. 
Moreover, for the Gaussian quantum systems, we derive an analytical solution for rate-limited Wasserstein distance of order 2, along with a Gaussian optimality theorem, showing that Gaussian measurement optimizes the rate in a system with Gaussian quantum source and Gaussian destination distribution.
The results further show that in contrast to the classical Wasserstein distance of Gaussian distributions, which corresponds to an infinite transmission rate, in the Quantum Gaussian measurement system, the optimal transport is achieved with a finite transmission rate due to the inherent noise of the quantum measurement imposed by Heisenberg's uncertainty principle. 
% This maximum transmission rate aligns with the accessible information of the Gaussian ensemble provided by \cite{holevo2020-gaussian-maximizers-observables}.

%% <-- Keywords HERE!
\noindent
% \textit{\textbf{JEL Classification: }%
% Q12; C22; D81.} %% <-- JEL code HERE!

\end{abstract}
}

% %%%%%%%%%%%%%%%%%%%%%%%%%%%%%%%%%%%%%%%%%%%%%%%%%%%%%%%%%%
% BODY OF THE DOCUMENT
% %%%%%%%%%%%%%%%%%%%%%%%%%%%%%%%%%%%%%%%%%%%%%%%%%%%%%%%%%%

\section{Introduction}
\label{sec:introduction}
% Optimal Transport
\subsection{Overview of Recent Works}
% \subsubsection{The problem of Optimal Transport}
The goal of optimal transport is to map a source probability measure into a destination one with the minimum possible cost \cite{talagrand1996transportation,villani2009optimal}. 
 Let $X$ be a random variable in the source probability space $(\mathcal{X}, \~F_\~X, P_X)$, where $\mathcal{X}$ is the support, $\~F_\~X$ is the event space defined by the $\sigma$-algebra of  sets on $\~X$, and $P_X$ is the probability measure. Let $Y$ be a random variable with the target probability space $(\mathcal{Y}, \~F_\~Y, P_Y)$. The optimal transport problem aims at finding an optimal mapping $f:\~X \to\~Y$ that minimizes the expectation of the transportation cost $c(x,y)$ \cite{monge1781memoire}. 
However, as such deterministic mappings do not exist in many cases, one has to resort to stochastic channels to transform the source distribution to the target distribution. Thus the problem boils down to finding the optimal \textit{coupling}   $\pi^* \in \~P(\~X \times \~Y)$ of marginal distributions  $P_X$ and $P_Y$ that minimizes the transportation cost \cite{kantorovich1942translocation}:
\begin{align*}
    \pi^* = \argmin_{\pi \in \~P(\~X \times \~Y)} &\int_{\mathcal{X}\times \mathcal{Y}} c(x,y) \pi(dx,dy),
    \end{align*}
subject to
% \begin{align*}
    $\int_{\mathcal{X} \times B_\~Y}\;  \pi(dx,dy) = P_Y(B_\~Y)$, and  $\int_{B_\~X \times \mathcal{Y}}  \pi(dx,dy) = P_X(B_\~X)$,
% \end{align*}
for any $B_\~X \in \~F_\~X$ and $B_\~Y \in \~F_\~Y$.
For a metric space $(\~X\times \~Y, d)$, 
and the cost given by $d^p$, for some $p\geq 1$,  
Wasserstein $p$-distance between the two probability measures $P_X$ and $P_Y$ is defined as the $p$-th power of the optimal transportation cost \cite{villani2009optimal,vaserstein1969markov, monge1781memoire,kantorovich1942translocation}.
This problem has been studied extensively in the literature 
with applications in many areas such as 
information theory, machine learning and statistical inference \cite{marton1986simpleproofofblowup,talagrand1996transportation}. 

The lossy source-coding problem in information theory aims to determine the minimum required rate for compressing a given source in asymptotically large blocks
so that it can be reconstructed to meet a prescribed distortion constraint.
The fundamental trade-off between the asymptotic compression rate and the reconstruction distortion is given by the \textit{rate-distortion function} \cite{cover_elements_of_information_theory_2005}.  A single-letter characterization of the optimal rate-distortion trade-off was provided by Shannon \cite{shannon1959coding}. 
The concept of simulating a channel was considered in the classical setting by Cuff with the notion of coordination capacity \cite{cuff2009dissertation, cuff2010coordination},  and the problem of  \textit{distributed channel synthesis} \cite{cuff2013distributed}.
Cuff provided an exact characterization of the associated performance limits. 
A closely related problem, known as Output-Constrained (OC)  lossy source coding  was considered in \cite{tamas_output_constrained_2015,wagner_2022_ratedistortionperecptiontradeoff},
and the corresponding performance limits were characterized.
This problem has recently found applications in multiple areas \cite{blau2018_distortion_perception,blau_2019-rethinking_ratedistortion_perception, zhang_junchen2021universal_ratedistortionperception, junchen_2022_onratedistortionperception,liu_junchen2022crossdomain_lossycompression}. 
In contrast to distributed channel synthesis which attempts to simulate a fixed channel, in OC lossy source coding,
only the output distribution of the reconstruction is constrained. 
Moreover, in the latter,
the n-letter output distribution has to 
satisfy an exact product form of the desired single-letter distribution, whereas, in the former, it is only constrained to be close to this product form in variational distance. 
This renders the problem intimately connected to optimal transport.

% \subsubsection{Rate-Limited Optimal Transport}
In \cite{Ozgur_2020_information_constrained_optimal_transport}, the authors formally introduced the problem of Information-Constrained Optimal Transport by imposing an additional constraint on coupling $\pi$ in the form of a threshold on the mutual information between $X$ and $Y$, and established an upper bound on the information-constrained Wasserstein distance by generalizing Talagrand's transportation cost inequality. It is worth noting that the information-cost function in \cite{Ozgur_2020_information_constrained_optimal_transport} is equivalent to the rate-distortion function of OC lossy source coding with unlimited common randomness  \cite{tamas_output_constrained_2015}. The protocol associated with this source coding problem provides an operational significance to the information constraint in \cite{Ozgur_2020_information_constrained_optimal_transport}, which can be interpreted as the optimal transport problem when the communication rate between the source and the destination is limited.
We refer to this problem as \textit{Rate-Limited Optimal Transport} (RLOT).

The quantum version of the optimal transport problem has also been investigated in recent years \cite{datta2020relativeentropy_optimaltransport,ikeda2020foundationofquantumoptimaltransport, depalma_quantum_optimal_transport_2021,cole2021quantumoptimaltransport,depalma_2021_quantumwasserstein_order1}. 
In \cite{depalma_quantum_optimal_transport_2021}, the authors proposed a generalization of the quantum Wasserstein distance of order 2 and proved that it satisfies the triangle inequality. They further showed that the associated quantum optimal transport schemes are in one-to-one correspondence with the quantum channels, and in the case of quantum thermal states, the optimal transport schemes can be realized by quantum Gaussian attenuators/amplifiers. In  \cite{depalma_2021_quantumwasserstein_order1},  the quantum Wasserstein distance of order 1 was introduced, and several quantum measure concentration inequalities were established.

% \subsubsection{The problem of Lossy Source Coding}
% Lossy Source Coding

The topic of rate-distortion theory has also drawn significant attention in the field of quantum information theory. In an early attempt, Barnum \cite{barnum_quantum_2000} conjectured a lower bound on the rate-distortion function for a quantum channel with entanglement fidelity as the distortion measure, based on the coherent information quantity. His lower bound was later proved to be not tight in \cite{datta_wilde_2012_quantum_ratedistortion_reverseshanon}. The authors established the quantum rate-distortion theorems for both entanglement-assisted and unassisted systems in terms of quantum information quantities such as entanglement of purification and quantum mutual information. The key analysis in \cite{datta_wilde_2012_quantum_ratedistortion_reverseshanon} rely on
the reverse Shannon theorem \cite{winter_devetak_2014_reverse_shannon_theorem}, which addresses the problem of simulating a noisy channel with the help of a noiseless channel, or more generally, simulating one noisy channel with another noisy channel.

In \cite{winter2004extrinsic}, Winter introduced the notion of information in quantum measurements and established the \textit{measurement compression theorem}, which delineates the required classical rate and common randomness to faithfully simulate a  measurement $\Lambda$ for an input state $\rho$. 
In \cite{wilde2012information_theoretic_costs}, variants of this measurement compression theorem were studied for the case of non-feedback simulation and the case with the presence of quantum side information. 
Further, in \cite{Datta_2013_quantumtoclassical}, Datta et. al. invoked this measurement compression theorem to give a proof of the quantum-to-classical (QC) rate-distortion theorem. 
This idea of measurement simulation was further extended to distributed measurement simulation for composite quantum states in \cite{pradhan_atif2021distributed,pradhan_atif2021faithful}, where the required classical rates and common randomness to faithfully simulate a bipartite state $\rho_{AB}$ using distributed measurements are characterized.

Most of these works provide their coding theorems only for finite-dimensional quantum systems. The problems of measurement compression, the quantum rate distortion, and the reverse Shannon theorem have not been addressed in the continuous variable quantum systems. 
The key analytical tools that form the basis of these works such as operator Chernoff bounds \cite{wilde2013quantum_information_theory_book} and one-shot quantum covering lemma \cite{tomamichel2015quantum} have not been studied in the continuous-variable quantum systems. However, the problems of channel capacity of quantum Gaussian states, and accessible information of Gaussian quantum states have been studied extensively  \cite{holevo_hirota1999capacityofquantumgaussianchannels,holevo2001evaluating-capacities-bosonic, holevo2020-gaussian-maximizers-observables}.

% ; this distributed measurement compression theorem was then leveraged to establish inner and outer bounds of the rate region for the distributed QC rate-distortion problem.

% However, the RLOT problem has never been studied in the subject of quantum channels or quantum measurement systems. 

\subsection{Contributions}
The present paper introduces the QC \textit{Rate-Limited Optimal Transport}  problem involving the measurement of a quantum source to produce a classical outcome with the desired distribution. 
The cost associated with this transformation is measured by a distortion observable. The system avails limited-rate classical communication and common randomness resources.
We establish a single-letter characterization of the achievable rate region of the measurement protocols acting on the quantum source state for the construction of a  prescribed destination classical distribution with a specific tolerance for distortion (see Theorem \ref{th:maintheorem}). We further consider the QC RLOT problem for the Continuous-Variable (CV) quantum systems
and characterize the associated performance limits with the development of a novel continuous measurement coding protocol (see Theorem \ref{th:continuous-theorem}). 
Our work enables the generalization of quantum optimal transport to the Rate-Limited (RL) setting as well as the generalization of classical information-constrained optimal transport to the quantum setting. 
Moreover, the analysis of the CV quantum systems is also one of the key contributions of this paper (Section \ref{sec:continuous-systems}).

In particular, we provide a detailed analysis of quantum-to-classical RLOT for the case of qubit source state and entanglement fidelity distortion measure (see \cite{barnum_quantum_2000,datta_wilde_2012_quantum_ratedistortion_reverseshanon}); the minimum transportation cost is explicitly characterized and is shown to be achievable with a finite transmission rate (see Theorem \ref{th:qubit_ratedist_theorem} and \ref{th:qubit_optimal_transport}). 
We further provide an evaluation of the optimal performance limit for the Gaussian QC systems with unlimited common randomness using a quadratic distortion observable constructed from canonical quadrature operators. 
First, we develop a Gaussian measurement optimality theorem (see Theorem \ref{th:gaussian_optimality_theorem}) which shows for a continuous QC system with a Gaussian quantum source and Gaussian destination distribution, the RLOT is a Gaussian measurement. Moreover, a detailed analytical formulation provides the parameters of the optimal Gaussian measurement. In this special case with unlimited common randomness, the RL optimal transportation cost corresponds to the RL  Wasserstein 2-distance between the quantum source state and the classical outcome distribution. In stark contrast to the classical Gaussian systems,  in the Gaussian QC system, the QC Wasserstein distance is attainable at a finite communication rate. This is a direct consequence of the Heisenberg uncertainty principle and the fact that the measurement noise cannot be made smaller than a threshold \cite{Holevo_2019_quantum_systems_book, serafini_quantum_2017}.

We would also like to mention some key differences between our source coding theorem and other works. Specifically, in comparison to \cite{Datta_2013_quantumtoclassical}, our work on finite-dimensional systems (Section \ref{sec:finite-systems}),  
has the additional constraint that the output must follow a predetermined distribution in the exact IID format. This provides a multi-letter protocol that governs the optimal transportation of a quantum source state to a target distribution, through a RL classical channel and with limited common randomness. In contrast to the conventional rate-distortion theorem for which the common randomness provides no performance improvements, in this problem, the common randomness can help reduce the communication rate by providing the extra randomness required to ensure the output has the desired IID distribution.

% The proof of the discrete coding theorem builds on analytical tools such as Winter's measurement compression theorem \cite{winter2004extrinsic}, the non-feedback measurement simulation \cite{wilde2012information_theoretic_costs} and the batch post-processing of \cite{wagner_2022_ratedistortionperecptiontradeoff}.

% In contrast to their work where a QC channel is being faithfully simulated in a nearly perfect sense, in this work, we ensure the output is following the desired distribution in the perfect IID format while maintaining the distortion threshold. 

\subsection{Prospective Applications}
The QC RLOT problem, besides its theoretical significance, has possible applications in some practical quantum systems. 
In this work, we study the systems in which, Alice has many copies of a quantum source and has a description of the source in terms of its density operator. Alice has the freedom to perform any sort of measurement she wants on the copies of quantum source states. The goal is to have a compressed classical approximation of the quantum source (with respect to some distortion measure) while requiring this classical approximation to have a specific IID distribution. 

To understand the significance of this output constraint, we first look into the benefits of these output-preserved compressions in the classical systems described in \cite{blau_2019-rethinking_ratedistortion_perception, tamas_output_constrained_2015, wagner_2022_ratedistortionperecptiontradeoff,li_klejsa2011distribution}. 
Specifically, the following issues occur in conventional rate-distortion codings and quantization systems:
    (I) Discontinuities emerge within the range of outcome values, presenting as discrete steps in the outcomes.
    (II) In specific setups, the reconstruction assigns null values to some parts of the signal. For the example of Gaussian systems with mean-squared error distortion measure, the problem reduces to reverse waterfilling which produces null values at the tails of the Gaussian spectrum.
    (III) In the extreme case of zero communications, the output would generate constant null values even if it has some statistical information about the source.
To combat these issues, in the classical systems, randomized (dithered) uniform quantizers were used to add some random noise to the signal before the quantization which resulted in better reconstructions that avoided nullified tails or discrete steps \cite{tamas_output_constrained_2015,li_klejsa2011distribution}. These dithered forms of quantization were studied with their benefits in both uniform quantizers and entropy-coded quantizers in \cite{roberts1962picture,schuchman1964dither,gray1993dithered,zamir1992universal,zamir1996information}.

Using a similar idea of dithering, \cite{li_klejsa2011distribution, wagner_2022_ratedistortionperecptiontradeoff} provided the distribution-preserving quantizers which ensure that the reconstructed signal has the exact distribution as the source. 
% This output distribution-preservation constraint has the artificial noise inherent in its randomized decoder and common randomness to ensure that the output distribution is achieved exactly as described in \cite{tamas_output_constrained_2015}. 
This technique has been shown to be very effective in maintaining the perceptional quality of the voice and image codings. Especially, with the compressors using generative adversarial networks (GAN)  which have discriminator networks designed to perceptually distinguish between the source and reconstruction \cite{tschannen2018deepgan,rippel2017realtime_gan,agustsson2019generativeadvnet,wagner_2022_ratedistortionperecptiontradeoff}.

Furthermore, the Wasserstein distance satisfactorily addressed the problem of vanishing gradient in the training that afflicted classical GANs trained with other distance measures such as total variation or Jensen-Shannon divergence \cite{arjovsky2017wassersteingan}. 
 In quantum GAN systems, it is suggested in several works \cite{chakrabarti2019quantumwassersteingans,becker2021quantumstatisticallearning,depalma_2021_quantumwasserstein_order1} that the quantum Wasserstein may address similar issues.
 Hence, 
 % besides the perspective of rate-distortion, 
 the choice of Wasserstein distance as the loss function for these GAN models relates them to the rate-limited optimal transport problem, which was alluded to earlier in the introduction.\footnote{In other words, the way the lossless source coding problem gives an operational interpretation of the concept of entropy, the problem addressed in this paper gives a similar interpretation of the rate-limited Wasserstein distance.}
% As was mentioned earlier these two problems are practically the same.

In our QC Gaussian measurement OC rate-distortion problem, the Gaussian optimality Theorem \ref{th:gaussian_optimality_theorem} basically shows that the OC rate-distortion function can be modeled with a Gaussian measurement with a specific inherent noise independent of the source. This is equivalent to a similar observation in the classical conventional systems that the dithered quantizer can be modeled using additive Gaussian noise \cite{zamir1996information}.
In the light of these results, the QC rate-limited optimal transport problem can find application in analyzing the recently emerging quantum-classical GAN systems where quantum sources are used to generate classical data (image, voice, etc) as described in \cite{huang2021experimentalquantumgan,tsang2023hybrid_qc_gan,zoufal2019quantum_gan_distribution,romero2021variationalquantumgen_continuous,situ2020quantum_gan_discrete}.

% Concentration inequalities
The definition of the QC Wasserstein distance and rate-limited QC Wasserstein distance provided in Definition \ref{def:ratelim_wasserstein} may lead to new measure concentration inequalities in the QC framework. For example, in the classical systems, an information-constrained OT inequality is provided in \cite{Ozgur_2020_information_constrained_optimal_transport}, which both sharpens Talagrand's inequality and extends that inequality to the rate-limited setting.
Then following Marton's approach (blowing-up lemma) \cite{marton1986simpleproofofblowup}, the authors obtained new measure concentration inequalities. 
% of using Talagrand's inequality to prove the blowing-up lemma for the Gaussian spaces, the strengthened inequality provides a blowing-up lemma 
% in an n-dimensional sphere $\mathbb{S}^n: \mathbb{R}^{n-1}$ with uniform measure $\mu$. The lemma shows that for any set  $A \subset \mathbb{S}^n: \mathbb{R}^{n-1}$ defined on the sphere, the required blow-up to increase the probability to 1 is the complementary angle that turns the set $A$'s equivalent sphere cap $\mathrm{Cap}(A) = \mathrm{Cap}(z_0^n,\theta)$ into a half-sphere $\mathrm{Cap}(z_0^n,\theta + \omega)$, with $\omega > \pi/2 - \theta $. 
% Similarly,  using Marton's approach on the corresponding information-constrained Wasserstein distance provides a generalization of the measure concentration on the sphere. This generalization, instead of only providing the blow-up neighbor $\omega > \pi/2 - \theta $ so that $\mu(A_\omega) \to 1$, shows that for any value of $\omega \in (\pi/2 - \theta, \pi/2]$, the blow-up $A_w$ is larger than the intersection between two orthogonal caps with asymptotic probability 1. This generalized blow-up lemma provided by the rate-limited Wasserstein distance recovers as a special case, the blow-up lemma provided by the conventional Wasserstein distance \cite{Ozgur_2020_information_constrained_optimal_transport, wu_ozgur2018capacityofrelaychannel}. 
Similarly, in the quantum setting, 
a quantum version of the Gaussian concentration inequality
was provided in
\cite{depalma_2021_quantumwasserstein_order1} 
based on the quantum Wasserstein distance of order 1 introduced therein. Therefore, we expect similar development of the optimal transport inequalities and concentration inequalities based on the introduced rate-limited QC Wasserstein distance.
These inequalities in the classical setting were further used by \cite{Ozgur_2020_information_constrained_optimal_transport,wu_ozgur2018capacityofrelaychannel} to prove a strong data processing inequality which further provides an upper bound on the capacity of primitive relay channels. These results are also expected to be extended to the quantum-classical systems which can have similar applications in analyzing the quantum relay channels, for example, see \cite{pereg_boche2021quantumbroadcastdecoder}. 

% {\color{blue}
% Some other applications of the RLOT problem are provided in \cite{Ozgur_2020_information_constrained_optimal_transport}. First, a bound was provided for the value of the RLOT. Then based on this bound, a strong data processing inequality is provided. Using this SDPI, an old open problem was solved which was the problem of finding the capacity of a primitive relay channel.

% It was noted in \cite{Ozgur_2020_information_constrained_optimal_transport} that rate-limited classical optimal transport inequalities find applications in the subject of concentration of measures and solving some problems in network information theory such as the relay channel. We expect similar applications of the rate-limited QC optimal transport to shed light on information processing problems in the quantum-classical framework such as QC Talagrand's inequality, and strong data processing inequality. 
% }

The contents of this paper are organized as follows. Finite dimensional QC systems are addressed in Section \ref{sec:finite-systems} with the statement of the coding theorem, the proof of the achievability and the proof of converse part, are given in Sections \ref{sec:theorem_statement}, \ref{sec:achievability} and \ref{sec:converse}, respectively.
% Also, the proof of cardinality bound is provided in Appendix \ref{appE}.
We develop the coding theorem of CV quantum systems in Section \ref{sec:continuous-systems}. Specifically, we introduce the continuity theorems that are needed for generalizing the definitions of measurement systems to continuous Hilbert spaces in Section \ref{subsec:generalized-definitions}, state the coding theorem for continuous Hilbert spaces in Section \ref{subsec:theorem-continuous}, and prove the achievability part in Section \ref{subsec:achievability-continuous}. 
In Section  \ref{sec:evaluations_qubit}, we consider the case of qubit measurement systems with unlimited common randomness, for which a detailed analysis of RLOT is provided. 
The Gaussian measurement systems are analyzed in Section \ref{sec:evaluation_gaussian} which introduces the rate-limited QC 2nd-order Wasserstein distance. 
Finally, proofs of the important theorems and lemmas are provided in the appendices.

% {\color{blue}
% In addition, by looking at the evaluation in Section \ref{sec:evaluation_gaussian}, as a rate-distortion coding problem, the provided Gaussian optimality Theorem \ref{th:gaussian_optimality_theorem} ensures that post-measured states are also Gaussian, which implies that the reconstructed average density operator on Bob's side is also of the proper Gaussian form, which is not generally true for an unconstrained QC rate-distortion in \cite{Datta_2013_quantumtoclassical}. This is an important feature for the systems which only accept Gaussian states. Further possible applications may also rise in the CV quantum key distributions with Gaussian modulation \cite{diamanti2015CVQKD,laudenbach2018CVQKD} and hybrid quantum-classical machine learning.
% }

\subsection{Notation}
Let $\~H_A$ denote the separable Hilbert space of system $A$. 
% For a finite-dimensional space, $\dim{\~H_A} < \infty$ denotes the dimensionality of the system.
Further, Let $\~L(\~H)$ be the set of linear operators on $ \~H$, and $\~G(\~H) \subset \~L(\~H)$ be the set of all density operators on $\~H$.  
A density operator $\rho^A \in \~G(\~H_A)$ in Hilbert space $A$ is denoted by the superscript $A$. Moreover, a composite density operator $\tau^{AB} \in \~G(\~H_A \otimes \~H_B)$ is defined over the tensor product Hilbert space of $\~H_A \otimes \~H_B$. 
% The von-Neumann entropy of a quantum state $\rho^A \in \~H_A$ is denoted by $H(\rho) = H(A) = \Tr{\rho \log \rho}$.
The trace norm is defined for an operator $P \in \~L(\~H)$ by $\norm{P}_1 = \mathrm{Tr}\big[\sqrt{P^\dagger P}\big]$.
Also, $\~{QN}(m_X, \Sigma_X)$ denotes the quantum Gaussian state with mean vector $m_X$ and covariance matrix $\Sigma_X$.
% For the classical systems, let $\~X$ denote the  alphabet. In the finite systems, $\~X $ is a finite set with $\abs{\~X}$ denoting the cardinality of the set.
% , whereas in the continuous systems, $\~X =  \mathbb{R}^m$ is a linear space. 
A sequence of $n$ samples taking values in a set $\mathcal{X}$ is expressed by the $n$ superscript and each sample by a subscript, $v^n = \{v_1, v_2, \cdots , v_n\}$. The tensor product of $n$ Hilbert spaces  is denoted as $\~H_{A_1} \otimes \~H_{A_2} \otimes \cdots \otimes \~H_{A_n}$. The composite density operator of the state of $n$ systems is denoted by $\rho^{A^n}$, where the local state of the $i$-th system is given by
$\rho^{A_i}$.
The term $d_{TV}(P, Q)$ is the total variation distance of between distributions $P$ and $Q$.
% For the matrix calculus, the cone of positive semi-definite and positive definite $n \times n$ matrices are denoted by $S_{+}^n$ and $S_{++}^n$ respectively.
 The trace of a matrix object is denoted by $\sp{.}$, as opposed to  $\Tr{.}$, the trace on a Hilbert space.

\section{Finite Dimensional Quantum Systems} 
\label{sec:finite-systems}
The system is comprised of an n-letter memoryless source with its product state $\rho^{\otimes n}$ as the input of an encoder on Alice's side, where 
$\rho$ is a density operator defined on a Hilbert space $\mathcal{H}_A$. On Bob's side, 
we have the reconstruction Hilbert space 
$\mathcal{H}_X$, representing the classical outcomes as quantum registered classical states, with an orthonormal basis indexed by a finite set $\mathcal{X}$. 
%  The encoder is a set of $|\~M|$ collective measurement POVMs $\Upsilon^{(m)} \equiv \{\Upsilon^{(m)}_l \}_{l=1}^{|\~L|}$ defined on the n-fold Hilbert space , selected by the random index $M$ from the set of shared random data $\~M$. Bob receives the outcome $L$ of the measurement through a classical channel and applies a randomized decoder to its input pair $(L,M)$ to obtain the final sequence $X^n$ stored in a quantum register. 
% We further assume that the classical variables $(L, M)$ have the rates $(R, R_c)$ according to $|\~L| = 2^{nR} $ and $|\~M| = 2^{nR_{c}} $. 
We also let the quantum state $R$ denote the reference of the source with the associated Hilbert space $\mathcal{H}_R$ with $\mbox{dim}(\mathcal{H}_R)=\mbox{dim}(\mathcal{H}_A)$. The composite state $\ket{\Psi_\rho}^{RA}$ of the source $\rho^A$ and the reference is obtained by the canonical purification \cite{wilde2013quantum_information_theory_book,wilde2012information_theoretic_costs,winter2004extrinsic}. The corresponding density operator of this purified state is denoted by $\Psi_\rho^{RA}$.

\subsection{System Model and Problem Formulation}
\subsubsection{Distortion Measure} \label{subsec:distortionmeasure}
The distortion measure between two systems $R$ and $X$ is defined in the general form using a distortion observable $\Delta_{RX}>0$
defined on $\mathcal{H}_R \otimes \mathcal{H}_X$ for the single-letter composite state $\tau^{RX}$, as described in \cite{Datta_2013_quantumtoclassical}: 
\begin{align}
    d(\Delta_{RX}, \tau^{RX}) := \Tr{\Delta_{RX} \tau^{RX}}.
\end{align}
Then, having an n-letter composite state $\tau^{R^n X^n}$, and the distortion observable for each $i$-th system defined as $\Delta_{R_i X_i}$, the average n-letter distortion is defined as
\begin{align}   \label{eq:average-singleletter-representation-distortion}
     d_n(\Delta^{(n)}, \tau^{R^n X^n}) &:= \Tr{\Delta^{(n)} \tau^{R^n X^n}}
    = \frac{1}{n} \sum_{i=1}^n \Tr{\Delta_{RX}^{R_i X_i}\tau^{R_i X_i}},
\end{align}
where $\tau^{R_i X_i} = \Tr_{[n] \setminus i} \{\tau^{R^n X^n}\}$ is the localized $i$-th composite state, and $\Delta^{(n)}$ is the average n-letter distortion observable defined as $\Delta^{(n)} := \frac{1}{n} \sum_{i=1}^n \Delta_{RX}^{R_i X_i} \otimes (\mathrm{id}^{RX})^{\otimes [n] \setminus i}$.
% \begin{align}
%     \Delta^{(n)} := \frac{1}{n} \sum_{i=1}^n \Delta_{RX}^{R_i X_i} \otimes (\mathrm{id}^{RX})^{\otimes [n] \setminus i}.
% \end{align}

In the case of a discrete QC system, the composite state has the form $\tau^{RX} = \sum_x P_X(x) \hat \rho_x^R \otimes \ketbra{x}{x}^X$, where $\hat \rho_x$ is the post-measurement reference (PMR) state and $P_X$ is the PMF of outcomes. We further decompose the distortion observable as $\Delta_{RX} = \sum_{t=1}^T \Delta_{R}^t \otimes \Delta_{X}^t$ using the tensor product decomposition. Thus, we get
\begin{align}
    % line 1
    \Tr{\tau^{RX} \Delta_{RX}} &= \Tr_R \left[  \Tr_X \left[  \left(\sum_x P_X(x) \hat \rho_x^R \otimes \ketbra{x}{x}^X\right) \left(\sum_{t=1}^T \Delta_{R}^t \otimes \Delta_{X}^t\right)  \right]  \right] \nonumber\\
    % line 4
    &= \sum_x P_X(x) \Tr_R \left\{ \hat \rho_x^R \Delta_R(x) \right\} = \Exx{X}{\Tr_R \left\{   \hat \rho_X^R \Delta_R(X)   \right\} },
\end{align}
where $\{\Delta_R(x): x\in \~X\}$ is a mapping $\Delta_R: \~X \to \~H_R$ of the form, $\Delta_R(x) := \sum_{t=1}^{T} \Delta_{R}^t \braket{x|\Delta_{X}^t}{x}$.
% \begin{align} % \label{eq:Delta-observable-decomposition}
%     \Delta_R(x) := \sum_{t=1}^{T} \Delta_{R}^t \braket{x|\Delta_{X}^t}{x}.
% \end{align}

\ifmycomments
The overall composite state of the post-measurement system with the measurement $\~M^{A \to \~X}: \{\Lambda_x\}_{x\in \~X}$, is given by
\begin{align}
    \tau_{RX} :=\sum_x \Tr_A\Big\{(\text{id}\otimes \Lambda_x)\Psi^{RA}\Big\} \otimes \ketbra{x}{x}
\end{align}
\fi

\subsubsection{Source Coding Scheme}
The system is comprised of an n-letter source coding scheme defined below.
\begin{definition} (\textbf{Discrete Source Coding Scheme})\label{def:source-coding}
An $(n,R, R_c)$ source-coding scheme for this QC system is comprised of an encoder $\~E_n$ on Alice's side and a decoder $\~D_n$ on Bob's side, with the following elements. The encoder is a set of $|\~M| = 2^{nR_c}$ collective n-letter measurement POVMs $\Upsilon^{(m)}\equiv \{\Upsilon^{(m)}_l\}_{l\in \~L}$, each comprised of $|\~L| = 2^{nR}$ POVM operators corresponding to $|\~L|$ outcomes and the common randomness value $m$, which determines the specific POVM that will be applied to the source state. 
Bob receives the outcome $L$ of the measurement through a classical channel and applies a randomized decoder to this input pair $(L,M)$ to obtain the final sequence $X^n$ stored in a quantum register. 
% We further assume that the classical variables $(L, M)$ have the rates $(R, R_c)$ according to $|\~L| = 2^{nR} $ and $|\~M| = 2^{nR_{c}} $. 
% Also, a codebook of size $|\~L|\times |\~M|$ is defined with codewords $w^n(l,m)$ indexed with the parameters $l$ and $m$, which will be explained later. The decoder is a classical i.i.d. post-processing channel $P_{X|W}$ from the codeword symbols to output symbols. 
Thus, the composite state of the reference and output induced by this coding scheme is
\begin{align}
    \tau_{\text{ind}}^{R^n X^n}
    &=\sum_{x^n} \sum_{m,l} \frac{1}{|\~M|} \Tr_{A^n}\left\{(\mathrm{id}^R \otimes \Upsilon^{(m)}_l)[\Psi_\rho^{RA}]^{\otimes n})\right\} \otimes \mathcal{D}_n(x^n|l,m) \ket{x^n}\bra{x^n}.
\end{align}
\end{definition}
% The encoder can alternatively be defined with the $\Gamma^{(m)}_{w^n}$ POVMs, which directly return the intermediate state $w^n$ codewords, as will be defined in \eqref{eq:Gamma_POVM}. 
We define the average n-letter distortion for the source coding system with encoder-decoder pair $(\~E_n,\~D_n)$,  distortion observable $\Delta_{R X}$ and source product state $\rho^{\otimes n}$ as
\begin{align}   \label{eq:average-distortion-definition}
     d_n(\rho^{\otimes n}, \~D_n \circ \~E_n) &= \Tr{\Delta^{(n)} (\mathrm{id}^{R^n} \otimes \~D_n \circ \~E_n) (\Psi_\rho^{R A})^{\otimes n}}.
\end{align}
% where $\Psi_\rho^{RA}$ is a purification of the source state $\rho^A$ \cite{wilde2013quantum_information_theory_book,wilde2012information_theoretic_costs}. 
The goal is to prepare the destination quantum ensemble on Bob's side while maintaining the distortion limit from the input reference state.

\subsection{Main Results: Achievable Rate Region for Discrete States}\label{sec:theorem_statement}

We use the following definition of achievability for the discrete quantum systems.
\begin{definition} (\textbf{Achievable pair})\label{def:achievable_pair}
A desired PMF $Q_X$ on the output space $\~X$ and a maximum tolerable distortion level $D$ are given. A rate pair $(R,R_c)$ is said to be achievable if 
for any $\epsilon>0$, and 
 all sufficiently large $n$, there exists an  $(n, R, R_c)$ coding scheme  comprising of a measurement encoder $\~E_n$ and the decoder $\~D_n$ that satisfy:
\begin{align}
    X^n \sim  Q_X^n,   \quad                      
    % \label{eq:definition_achievability_outputdist},
     d_n(\rho^{\otimes n}, \~D_n \circ \~E_n) \leq D+\epsilon.
     \label{eq:definition_achievability_distortion}
\end{align}
\end{definition}
The expression \eqref{eq:definition_achievability_distortion}
indicates that the output sequence must be IID with fixed distribution $Q_X$ and that the n-letter distortion between the input and output state must be asymptotically less than a threshold $D$. We further define the achievable rate region as follows.
\begin{definition}(\textbf{Achievable Rate Region})\label{def:achievable_rateregion}
    Having the desired output PMF $Q_X$, the input state $\rho$ and a distortion threshold $D$, the achievable rate region $\~R(D,\rho || Q_X)$ is defined as the closure of the set of all achievable rate pairs with respect to the given $\rho$, $Q_X$ and $D$. The output-constrained rate-distortion function is defined for any specific value of $R_c$ as $R(D;R_c,\rho||Q_X) \equiv \inf \left\{R: (R,R_c)\in \~R(D,
     \rho||Q_X) \right\}$.
 % \begin{align}
 %     R(D;R_c,\rho||Q_X) \equiv \inf \left\{R: (R,R_c)\in \~R(D,
 %     \rho||Q_X) \right\}.
 % \end{align}
 The inverse of the above function which for any fixed $R_c$, is a mapping from the communication rates to their corresponding minimum transportation cost, is called the RLOT Cost function and expressed by $D(R; R_c, \rho||Q_X)$.  
\end{definition}

Based on the above definitions, we establish the main theorem 
which provides the single-letter characterization of the achievable rate region 
as follows:
\begin{definition}\label{def:maintheorem-definition}
    Given the distortion threshold $D$, the output PMF $Q_X$ and having a product input state $\rho^{\otimes n}$, define $\~R_I(D,\rho\| Q_X)$ as the closure of the set of all rate pairs $(R,R_c)$, for which 
    there exists an intermediate state $W$ with a corresponding measurement POVM $\Lambda \equiv \{\Lambda_w\}_{w\in \~W}$ and randomized post-processing transformation $P_{X|W}$  satisfying
    \begin{align} \label{eq:info_inequality1}
        R &\geq I(R;W)_\tau, \\
        R+R_c &\geq I(W;X)_\tau,\label{eq:info_inequality2}
    \end{align}
    where $W$, with a Hilbert space $\mathcal{H}_W$  along with an orthonormal basis indexed by a finite set $\mathcal{W}$,  constructs a quantum Markov chain $R-W-X$ with the overall post-measured composite state
  \begin{align*}
      \tau^{RWX} = \sum_{w,x} P_{X|W}(x|w) \left(\sqrt{\rho} \Lambda_w \sqrt{\rho} \right)^R \otimes \ketbra{w}{w}^W \otimes \ketbra{x}{x}^X,  
  \end{align*}
  from the set 
		\begin{align} \label{eq:set_feasible_theorem}
		    \~S(D) = \left\{ \tau^{RWX} \middle\vert 
		    \begin{array}{l}
                    \sum_w P_{X|W}(x|w) \Tr{ \Lambda_w \rho}  =  Q_X(x)  \quad \text{for }x\in \~X\\
                    \Tr{\Delta_{RX} \tau^{RX}} \leq D\\
                     |\~W| \leq  (\dim \~H_A)^2 +|\~X|+ 1  
            \end{array}
            \right\}.
		\end{align}
\end{definition}
\begin{theorem}\label{th:maintheorem}
 A single-letter characterization of the achievable rate region is given by
$\~R(D,\rho\| Q_X)=\~R_I(D,\rho\| Q_X)$.
\end{theorem}

\begin{definition}
  Given an input state $\rho$ and an output classical distribution $Q_X$, and a distortion observable $\Delta_{RX}$, we separately define the information-constrained optimal transport by, 
  \begin{align*}
      D_\text{Info}(R; \rho || Q_X) :=\min_{\Lambda_X} &\Tr{\Delta_{RX} \tau^{RX}},\\
      \text{subject to: }& I(R;X)_\tau \leq R,\quad \Trxx{R}{\tau^{RX}} = \rho,\quad  X \sim Q_X.
  \end{align*}
\end{definition}

\begin{remark}
    A consequence of the above theorem is that in the case of unlimited common randomness, the RLOT cost function is equal to the information-constrained optimal transport
    % \begin{align}
        $D(R;\infty, \rho || Q_X)  = D_\text{Info}(R; \rho || Q_X)$.
    % \end{align}
\end{remark}

\begin{remark}
   With tensor-product source state $\rho^{\otimes n}$ and IID output distribution $Q_X^n$,  and the cost function
    \eqref{eq:average-singleletter-representation-distortion} of any quantum-classical state $\tau^{R^n X^n}$, 
        the information-constrained optimal transportation cost is tensorizable in the following sense,
    % \begin{align*}
        $D_\text{Info}(nR; \rho^{\otimes n} || Q^n_X) =  D_\text{Info}(R; \rho || Q_X)$. This result is also applicable to the CV quantum systems considered in Section \ref{sec:continuous-systems}.
    % \end{align*}
    % This tensorization is also applicable to the rate-limited Wasserstein distance defined in Definition \ref{def:ratelim_wasserstein}.
\end{remark}
\subsection{Proof of Achievability for Finite Systems} \label{sec:achievability}
In this section, we prove the achievability of  Theorem \ref{th:maintheorem}, given that the pair $(Q_X, D)$ is provided by the setting of the theorem.

\subsubsection{Codebook and Encoder Construction}
In this section, by following the random codebook construction of \cite{winter2004extrinsic}, we generate a codebook in the intermediate space $\~W$ from the probability $P_W(w ) = \Tr{\Lambda_w \rho}$, which is derived from applying the measurement POVM $\Lambda \equiv \{\Lambda_w\}_{w \in \~W}$ to source state $\rho$.
Then a sequence of $n$ independent outcomes $W$ has the IID distribution  $P^n_W(w^n) $. The pruned distribution  is then defined by only selecting $w^n$ from the typical set of $W$,
 \begin{align*}
     P_{W'^n}(w^n) = \begin{cases}
                        P_{W}^n(w^n) /P_\~T    &\text{if } w^n \in \~T_W^{n,\delta}\\
                        0 &\text{otherwise}
                    \end{cases},
 \end{align*}
 where $P_\~T := \Pr{W^n \in \~T_W^{n,\delta}} \geq 1 - \epsilon$ and $\epsilon$ is a presumed fixed parameter.
  Consequently, a total of $|\~M|\times |\~L|$ random codewords $w^n$ are generated from the pruned distribution $P_{W'^n}$ and indexed with $(m,l)$ pair, comprising a random codebook. We then repeat this process to generate $|\~K|$ codebook realizations. The codewords in each codebook are indexed as $W^n(m,k,l)$. 
 The random variable $K$ is introduced as additional randomness for analytical purposes which will be de-randomized at the end. 
 
 Also, for each $ w^n$ sequence, the following set of typically projected PMR operators is defined
 \begin{align*}
    \hat\rho'_{w^n} := \Pi_{\rho, \delta}^n \;\Pi_{\hat\rho_{w^n},\delta}^n \;\hat\rho_{w^n} \;\Pi_{\hat\rho_{w^n},\delta}^n \;\Pi_{\rho, \delta}^n,
\end{align*}
where $\Pi_{\rho, \delta}^n$ and $\Pi_{\hat\rho_{w^n},\delta}^n$ are the typical set and conditional typical set projectors respectively \cite{wilde2012information_theoretic_costs}, and where
 $  \hat \rho_w := \frac{1}{P_W(w)}\sqrt{\rho} \Lambda_w \sqrt{\rho}$, 
are the conditional PMR states given the outcome $w$.
% The overall state can also be formulated in terms of the PMR states  as
% \begin{align*}
%     \tau^{RWX} = \sum_{w,x} \hat \rho_w \otimes P_{W}(w)\ketbra{w}{w}^W \otimes P_{X|W}(x|w) \ketbra{x}{x}^X,  
% \end{align*}
% where
% %\begin{align} \label{eq:postmeasurementstate}
%  $   \hat \rho_w := \frac{1}{P_W(w)}\sqrt{\rho} \Lambda_w \sqrt{\rho}$, and 
%     $P_W(w) = \Tr{\Lambda_w \rho}$,
% %\end{align}
% are the conditional PMR state given the outcome $w$, and the probability of the outcome $w$, respectively. 
We are also interested in the expectation of the above operators over the typical sequence $W^n \in \~T_W^{n,\delta}$, which is
\begin{align*}
    \hat \rho'^n &:= \Exx{W'^n}{\hat \rho'^n_{W'^n}} = \sum_{w^n \in \~T_W^{n,\delta}} P_{W'^n}(w^n) \hat \rho'^n_{w^n}.
\end{align*}

Further define  a cut-off projector $\Pi$, which projects to the subspace spanned by the eigenstates of $\hat\rho'^n$ with eigenvalues larger than $\epsilon \alpha$, where $\alpha:= 2^{-n[H(R) + \delta]}$. Then the cut-off version of the operators and the expected cut-off operator are given by
\begin{align}
    \hat\rho''_{w^n} := \Pi \;\hat\rho'_{w^n} \;
    \Pi,  \quad
    \hat \rho''^n := \Pi\;    \hat \rho'^n    \;\Pi. \label{eq:cutoff_expected_operator}
\end{align}
\ifmycomments
{\color{blue}
The reason why the $\Pi$ operator is defined is that unlike the classical typical set for which the probability of any sequence in the typical set is lowerbounded by $2^{-n [H(X)+ \delta]}$, in the quantum case we have
\begin{align}
 \Pi_{\rho,\delta}^n \rho^{\otimes n}  \Pi_{\rho,\delta}^n \geq 2^{-n[H(X)+ \delta]}     \Pi_{\rho,\delta}^n.
\end{align}
This does not imply a lower-bound on the probability of each individual typical sequence, therefore, there still exists a need for a projector to cut off the low-probability typical sequences, as would be needed later in the Chernoff bound conditions \cite{wilde2012information_theoretic_costs}. 
}
\fi
 
%\subsubsection{Encoder POVM} 
 Consequently, similar to \cite{wilde2012information_theoretic_costs}, for each $(k,m)$ we define the POVM operators
\begin{align}
    \Upsilon^{(k,m)}_l := \frac{1-\epsilon}{1+\eta} \frac{1}{|\~L|} (\rho^{\otimes n})^{-1/2}    \hat\rho''_{W^n(m,k,l)}     (\rho^{\otimes n})^{-1/2},
\end{align}
with  $\eta \in (0,1)$ a parameter which can be determined later. One can alternatively define the POVM operators such that it directly outputs $W^n$:
\begin{align}
    \Gamma_{w^n}^{(m,k)} &:= \sum_{l} \mathds{1}\{W^n(m,k,l) = w^n\} \Upsilon_l^{(k,m)}
     = \gamma_{w^n}^{(m,k)} \;(\rho^{\otimes n})^{-1/2} \;{\hat\rho''_{w^n}}\; (\rho^{\otimes n})^{-1/2}, \label{eq:Gamma_POVM}
\end{align}
where
\begin{align} \label{eq:gamma_coefficient}
   \gamma_{w^n}^{(k,m)} := \frac{1}{|\~L|}\sum_{l=1}^{|\~L|} \frac{1-\epsilon}{1+\eta} \mathds{1}\left\{  W^n(m,k,l) = w^n  \right\}.
\end{align}
 % Then the set of measurement operators for each $m\in \~M$ and $k\in\~K$ is given by
 % \begin{align*}
 %     \tilde M_\Gamma^{(m,k),n} = \{\Gamma_{w^n}^{(m,k)}: w^n \in \~T^{n,\delta}_W\}.
 % \end{align*}
  By applying the operator Chernoff bound according to \cite{wilde2012information_theoretic_costs}, we claim that for $|\~L|  \geq 2^{n(I(R;W)+ 3\delta)}$,  the intersection of the following events $E_{m,k}$ happen with probability close to 1 for all  $\eta\in (0,1)$:
\begin{align} \label{eq:Em-event-chernoff}
    E_{m,k}: \quad \frac{1}{|\~L|} \sum_l \hat\rho''_{W^n(m,k,l)} \in [(1 \pm \eta) \hat\rho''^n] \quad \forall m \in \~M, \; k \in \~K.
\end{align}
%%%%%%%%%%%%%%%%%%%%%%%%%%%%%%%%%%%%%%%%%%%%%
\ifmycomments
{\color{blue}
Then by using the  following expressions we have for each $ m \in \~M$ \cite{wilde2012information_theoretic_costs}:
\begin{align}
    % line 1
    \sqrt{(\rho^{\otimes n})}\sum_{w^n\in \~T^{n,\delta}_W} \Gamma_{w^n}^{(m)}\sqrt{(\rho^{\otimes n})} &= \sum_{w^n\in \~T^{n,\delta}_W}\gamma_{w^n}^{(m)} \;{\hat\rho''_{w^n}}\\
    % line 2
    &= \frac{1-\epsilon}{1+\eta}    \frac{1}{|\~L|}    \sum_{w^n\in \~T^{n,\delta}_W}\sum_{l=1}^{|\~L|}  \mathds{1}\left\{  W^n(m,l) = w^n  \right\} \;{\hat\rho''_{w^n}}\\
    % line 3
    &= \frac{1-\epsilon}{1+\eta}    \frac{1}{|\~L|}    \sum_l \hat\rho''_{w^n(m,k,l)}\\
    %line 4
    &\leq  \frac{1-\epsilon}{1+\eta} (1 + \eta) \hat\rho''^n =  (1-\epsilon)  \hat\rho''^n
\end{align}
where the inequality appeals to Chernoff bound \eqref{eq:Em-event-chernoff}. Then substituting the expression of the expected cut-off state, we continue:
\begin{align}
    (1-\epsilon)  \hat\rho''^n &= (1-\epsilon) \sum_{w^n \in \~T_W^{n,\delta}} P_{W'^n}(w^n) \Pi\;    \Pi_{\rho, \delta}^n \;\Pi_{\hat\rho_{w^n},\delta}^n \;\hat\rho_{w^n} \;\Pi_{\hat\rho_{w^n},\delta}^n \;\Pi_{\rho, \delta}^n   \;\Pi\\
    % line 2
    &\leq  (1-\epsilon) \sum_{w^n \in \~T_W^{n,\delta}} P_{W'^n}(w^n) \Pi\;    \Pi_{\rho, \delta}^n \;\hat\rho_{w^n} \;\Pi_{\rho, \delta}^n   \;\Pi\\
    % line 3
    &=  (1-\epsilon) \sum_{w^n \in \~T_W^{n,\delta}} \frac{P_{W^n}(w^n)}{1-\epsilon} \Pi\;    \Pi_{\rho, \delta}^n \;\hat\rho_{w^n} \;\Pi_{\rho, \delta}^n   \;\Pi\\
    % line 4
    &\leq  \sum_{w^n \in \~W^{n}} P_{W^n}(w^n) \Pi\;    \Pi_{\rho, \delta}^n \;\hat\rho_{w^n} \;\Pi_{\rho, \delta}^n   \;\Pi\\
    % line 5
    &=   \Pi\;    \Pi_{\rho, \delta}^n \rho^{\otimes n} \;\Pi_{\rho, \delta}^n   \;\Pi\\
    &\leq \rho^{\otimes n} =\omega.
\end{align}
The first inequality follows from the fact that projection reduces the space, i.e.,
\begin{align}
\Pi_{\hat\rho_{w^n},\delta}^n \;\hat\rho_{w^n} \;\Pi_{\hat\rho_{w^n},\delta}^n \leq \hat \rho_{w^n}.
\end{align} 
}
\fi
%%%%%%%%%%%%%%%%%%%%%%%%%%%%%%%%%%%%%%%%%%%%%
Then, for each $m\in \~M$ and $k \in \~K$, the set of measurement POVMs $\{\Gamma_{w^n}^{(m,k)}\}_{w^n \in \~T^{n,\delta}_W}$ forms a sub-POVM, i.e.,
\begin{align*}
    \sum_{w^n\in \~T^{n,\delta}_W} \Gamma_{w^n}^{(m,k)} \leq I \quad \forall m\in \~M, \; k \in \~K.
\end{align*}
We further complete this sub-POVM by appending an extra operator 
\begin{align*}
 \Gamma_{w_0^n}^{(m,k)}:= I - \sum_{w^n} \Gamma_{w^n}^{(m,k)} .
\end{align*}
The new set of operators $\left\{\{\Gamma_{w^n}^{(m,k)}\}_{w^n \in \~T^{n,\delta}_W}\ , \ \Gamma_{w_0^n}^{(m,k)} \right\}$  is a valid POVM. The intermediate POVM is established, by picking one of these  $ |\~M|$ POVMs according to the uniformly distributed common randomness,
 \begin{align*}
     {\Tilde\Lambda}_{w^n}^{{(k)}^A} := \frac{1}{|\~M|}\sum_{m=1}^{|\~M|} \Gamma_{w^n}^{(m,k)}, \qquad \forall k\in\~K.
 \end{align*}
 
The decoder applies the  $P_{X|W}$ classical memoryless channel to each element of $w^n$ sequence. By combining the encoder and decoder, the overall encoder-decoder POVM is obtained by
 \begin{align}    \label{eq:fullencoderdecoder}
     \tilde{\tilde{\Lambda}}^{(k)}_{x^n} \equiv \sum_{w^n\in \~W^n} P^n_{X|W}(x^n|w^n) {\tilde\Lambda}_{w^n}^{{(k)}^A}, \quad \forall x^n \in \~X^n,
 \end{align}
for all $k \in \mathcal{K}$. It should be noted that the above memoryless decoder is not the final form. We will later modify this decoder to yield a non-product batch decoder in section \ref{sec:exactly_iid}. This is particularly important to ensure a perfect IID output distribution.
  Using the above POVMs one can write the induced composite state of the reference and output for each random codebook realization $k\in \~K$ as
 \begin{align}\label{eq:nletterinducedstate}
     \tau_{\text{ind},k}^{R^n X^n} = \sum_{x^n} \Tr_{A^n}\left\{     ({(\text{id}^R)}^{\otimes n} \otimes \tilde{\tilde{\Lambda}}^{(k)}_{x^n})(\Psi_\rho^{R A})^{\otimes n}      \right\} \otimes \ketbra{x^n}{x^n}.
 \end{align}
% where $\Psi_{RA}$ is a purification of source state $\rho$.

\subsubsection{Proof of Near IID Output Distribution} \label{subsec:proof_of_near_iid}
It turns out that the proof of near IID output distribution does not depend on the codebook index $k \in \~K$. Therefore, we hereby remove the index $k$ from all expressions of this subsection, which means the following formulations apply to any fixed $k\in \~K$.
% \subsubsection{Forming the Induced Output State}
By tracing over the reference state in \eqref{eq:nletterinducedstate} we write the output state
\begin{align}\label{eq:sigma_induced_outputstate}
     \sigma_{\text{ind}}^{X^n} = \sum_{x^n} \Tr\left\{     ({(\text{id}^R)}^{\otimes n}\otimes \tilde{\tilde{\Lambda}}_{x^n})(\Psi_\rho^{R A})^{\otimes n}      \right\}  \ketbra{x^n}{x^n}.
 \end{align}
 Also, from  the conditions of the $\~S(D) $ feasible set \eqref{eq:set_feasible_theorem}, the output desired tensor state has the form,
 \begin{align}\label{eq:sigma_desired_outputstate}
     \big(\sigma_\text{des}^X\big)^{\otimes n} =  \sum_{x^n} Q^n_X(x^n) \ketbra{x^n}{x^n} =  \sum_{x^n} \left[\prod_{i=1}^n \sum_{w \in \~W} P_{X|W}(x_i|w_i) \Tr{ \Lambda_{w_i} \rho}\right]\ketbra{x^n}{x^n} .
 \end{align}
 
% \subsubsection{Trace Distance Between Output States}
 Consequently, the trace distance between the induced output state and the desired product state is,
 \begin{align*}    
 &\norm{\big(\sigma_\text{des}^X\big)^{\otimes n} - \sigma_{\text{ind}}^{X^n}  }_1 \\
    &= \sum_{x^n} \abs{
            \sum_{w^n} P^n_{X|W}(x^n|w^n) P_W^n(w^n)  -  
            \frac{1}{|\~M|} \sum_{w^n \in \~W^n}   P^n_{X|W}(x^n|w^n) 
            \Tr\Bigg\{    
                \bigg(\sum_{m=1}^{|\~M|} \Gamma_{w^n}^{(m)}\bigg)  (\rho^{\otimes n})
            \Bigg\}
        } \leq S_1 + S_2,
 \end{align*}
 where we split and bound the above term by separating the extra operator from the rest of the POVM:
 \begin{align}
     S_1     &\triangleq \sum_{x^n} \abs{
            \sum_{w^n} P^n_{X|W}(x^n|w^n) P_W^n(w^n)  -  
            \sum_{w^n\neq w_0^n}   P^n_{X|W}(x^n|w^n) \Tr{    \left(\frac{1}{|\~M|}\sum_{m=1} \Gamma_{w^n}^{(m)}\right)  (\rho^{\otimes n})}
        }, \label{eq:S1_def}\\
     S_2  &\triangleq \sum_{x^n} \abs{
            P^n_{X|W}(x^n|w_0^n)  \Tr \Bigg\{
                \frac{1}{|\~M|}\sum_{m=1} \Bigg(
                    I - \sum_{w^n \neq w_0^n} \Gamma_{w^n}^{(m)}
                \Bigg) 
                (\rho^{\otimes n})
            \Bigg\}
        }.\label{eq:S2_def}
 \end{align}
 We further simplify  $S_1$ by substituting \eqref{eq:Gamma_POVM} into \eqref{eq:S1_def} and bound it again by $S_1 \leq S_{11}+ S_{12}$ by adding and subtracting a proper term and using triangle inequality:
\begin{align}
    S_{11}  &\triangleq \sum_{x^n} \abs{
            \sum_{w^n} P^n_{X|W}(x^n|w^n) P_W^n(w^n)  -  
            \frac{1}{|\~M||\~L|} \frac{1-\epsilon}{1+\eta} \sum_{m,l} P^n_{X|W}(x^n|W^n(l,m)) },\\
    S_{12}  &\triangleq \frac{1}{|\~M||\~L|} \frac{1-\epsilon}{1+\eta}\sum_{x^n} \abs{
             \sum_{m,l} P^n_{X|W}(x^n|W^n(l,m))  \left(1 -  
              \Tr{\hat \rho''_{W^n(l,m)}}\right)
        } \nonumber\\
        % &=\frac{1}{|\~M||\~L|} \frac{1-\epsilon}{1+\eta}\sum_{x^n}
        %      \sum_{m,l} P^n_{X|W}(x^n|W^n(l,m))  \left(1 -  
        %       \Tr{\hat \rho''_{W^n(l,m)}}\right) \nonumber \\
        &=\frac{1}{|\~M||\~L|} \frac{1-\epsilon}{1+\eta}
             \sum_{m,l}  \left(1 -  
              \Tr{\hat \rho''_{W^n(l,m)}}\right).           
\end{align}
For $S_{11}$, using the classical soft-covering lemma   \cite[Lemma 2]{cuff2013distributed} with the condition that $R+R_c > I(X;W)$, one can provide a decaying upper bound for its expectation as
\begin{align}\label{eq:softcoveringS11}
 \Ex{S_{11}}\leq \frac{3}{2}\exp{-t n},   
\end{align}
 for some $t>0$. Also, by taking the expectation of $S_{12}$  we have
\begin{align} \label{eq:ES12}
    \Ex{S_{12}} = \frac{1-\epsilon}{1+\eta} (1 - \Tr{ \hat \rho''^n})
              &\leq \frac{1-\epsilon}{1+\eta} (2\epsilon + 2\sqrt{\epsilon} ) \triangleq \epsilon_2,  
\end{align}
where the equality follows from \eqref{eq:cutoff_expected_operator} and the inequality appeals to the properties of the typical set and the Gentle Measurement Lemma \cite{wilde2012information_theoretic_costs,winter_1999_codingtheorem_quantumchannels,hiroshi_2007_makinggoodcodes}.
\ifmycomments
{\color{blue}
The second equality is from \eqref{eq:cutoff_expected_operator} and the inequality is from the following lower bound on the trace:
\begin{align}
    \Tr{\hat\rho''^n} &\geq  \Tr{\hat\rho'^n} - \norm{\hat\rho''^n - \hat\rho'^n}_1\\
    &=  \Tr{\hat\rho'^n} - \sum_{i }  \Big(\theta_{\hat \rho'^n,i} - \theta_{\hat \rho''^n,i}\Big) \\
     &\geq  \Tr{\hat\rho'^n} - \alpha\epsilon .\text{rank}(\hat\rho'^n) \\
     &\geq \Tr{\hat\rho'^n} - \alpha \epsilon. \,\alpha^{-1} \\
     &\geq 1 - \epsilon - 2\sqrt{\epsilon} - \epsilon = 1 - 2\epsilon - 2\sqrt{\epsilon}
\end{align}
where  $\theta_{\hat \rho'^n, i}$ and $\theta_{\hat \rho''^n, i}$  are the 
 $i$-th eigenvalues of the operators $\hat \rho'^n$ and $\hat \rho''^n$ respectively. 
 Also, the last inequality follows from the following bound on the trace of $\hat \rho'^n$ expected operator  \cite{wilde2012information_theoretic_costs}:
 \begin{align}
     \Tr{\hat \rho'^n} &= \sum_{w^n \in \~T^{n,\delta}_W} P_{W'^n}(w^n)\Tr{\hat \rho'^n_{w^n}}\\
     &= \sum_{w^n \in \~T^{n,\delta}_W} P_{W'^n}(w^n)\Tr{\Pi_{\rho, \delta}^n \;\Pi_{\hat\rho_{w^n},\delta}^n \;\hat\rho_{w^n} \;\Pi_{\hat\rho_{w^n},\delta}^n}\\
     &\geq \sum_{w^n \in \~T^{n,\delta}_W} P_{W'^n}(w^n)\bigg(\Tr{\Pi_{\rho,\delta}^n \hat \rho_{w^n} } - \norm{\hat \rho_{w^n} - \Pi_{\hat\rho_{w^n},\delta}^n \;\hat\rho_{w^n} \;\Pi_{\hat\rho_{w^n},\delta}^n}_1\bigg)\\
     &\geq 1 - \epsilon - 2 \sqrt{\epsilon} 
 \end{align}
 where the last inequality follows from the properties of the typical set and the Gentle Measurement Lemma \cite{wilde2012information_theoretic_costs,winter_1999_codingtheorem_quantumchannels,hiroshi_2007_makinggoodcodes}.
 }
 \fi
Next, we bound and simplify the expectation of $S_2$ by substituting \eqref{eq:Gamma_POVM} into \eqref{eq:S2_def}:
\begin{align}
        %line 1
       \Ex{S_2} &\leq  
       \Ex{
            \frac{1}{|\~M|} \sum_{m} \sum_{x^n } P^n_{X|W}(x^n|w_0^n) 
            \abs{ 
                \Tr\Bigg\{
                    (\rho^{\otimes n}) - \sum_{w^n \neq w_0^n}\gamma_{w^n}^{(m)} \hat\rho''_{w^n} 
                \Bigg\} 
            }
        } \nonumber \\
        %line 2
        % &=  \frac{1}{|\~M|} \sum_{m}  \Ex{\abs{ 1- \Tr{\sum_{w^n \neq w_0^n}\gamma_{w^n}^{(m)} \hat\rho''_{w^n} } }} \nonumber\\
        %line 3
        &\stackrel{a}{=}  1 - \frac{1}{|\~M|} \sum_{m}   \Tr \Bigg\{\sum_{w^n \neq w_0^n}\Ex{\gamma_{w^n}^{(m)}} \hat\rho''_{w^n} \Bigg\} 
        %line 4
        % &=  1 - \frac{1-\epsilon}{1+\eta} \Tr{\sum_{w^n \in \~T_\delta^n(W)}\frac{P_W^n(w^n)}{1-\epsilon} \hat\rho''_{w^n} } \\
        % %line 5
        % &=  1 - \frac{1-\epsilon}{1+\eta} \Tr{\hat\rho''^n } \\
        % %line 6
        \stackrel{b}{\leq} 1- \frac{1-\epsilon}{1+\eta} (1 - 2\epsilon - 2\sqrt{\epsilon} )
        = \frac{\eta + \epsilon}{1 + \eta} + \epsilon_2 \triangleq \epsilon_3 ,         \label{eq:expected_inequality_S2}
\end{align}
where in (a) we remove the absolute sign because the trace is always less than or equal to one, and (b) uses the result from \cite{wilde2012information_theoretic_costs}. Hence, combining \eqref{eq:expected_inequality_S2}, \eqref{eq:ES12} and \eqref{eq:softcoveringS11} we show that the expected distance between the output state induced by the random codebook and the product single-letter state is arbitrarily small for sufficiently large n:
\begin{align} \label{eq:outputstate-expected-inequality}
    \Ex{\norm{(\sigma_{\text{des}}^X)^{\otimes n} - \sigma_{\text{ind}}^{X^n}}_1} \leq \epsilon_2 + \epsilon_3 + \frac{3}{2}\exp{-tn} \triangleq \epsilon_{os}.
\end{align}
\ifmycomments
\myqu{right?}\myans{Yes, the answer is provided in video meeting 10.} And the last line is similar to \eqref{eq:ES12}. 

\mynote{we did not need to use the gentle measurement for the ensembles in the above steps. The lower bound on the trace of the ensemble state which was defined in Eq. (28) of Wilde was enough.}
\fi

\subsubsection{Proof of Distortion Constraint}
% In this section, we prove that a coding scheme with the aforementioned design exists that satisfies the average distortion constraint in \eqref{eq:definition_achievability_distortion}. 
The average distortion for a codebook $k\in \~K$ is given by
\begin{align}
    %line 1
   {d_n}^{\{k\}} (\rho^{\otimes n}&, \~D_n \circ \~E_n) = \Tr{\Delta^{(n)} \tau_{\text{ind},k}^{R^n X^n}} \nonumber\\
   % line 2
   &= \Tr{\Delta^{(n)} \bigg(     \sum_{x^n} \Tr_{A^n}\left\{     (\text{id}^{\otimes n}_R\otimes \tilde{\tilde{\Lambda}}^{(k)}_{x^n})(\Psi_\rho^{R A})^{\otimes n}      \right\} \otimes \ketbra{x^n}{x^n}      \bigg)} \nonumber\\
    % line 6
   &= \Tr \bigg\{\Delta^{(n)} \bigg(       \sum_{m,l} \frac{1}{|\~M|} \Tr_{A^n}\left\{     (\text{id}^{\otimes n}_R\otimes \Upsilon^{(k,m)}_{l})(\Psi_\rho^{R A})^{\otimes n}      \right\} \otimes  \sigma_{W^n(m,k,l)} \bigg)\bigg\}.
\end{align}
where 
    $\sigma_{w^n} = \sum_{x^n} P_{X|W}^n (x^n| w^n) \ketbra{x^n}{x^n}$
is the classical decoder channel.
Recall from Section \ref{subsec:proof_of_near_iid} that in order to have a faithful near IID output state, we need to satisfy the conditions of soft-covering lemma $|\~M||\~L| > 2^{n I(X;W)}$, which is needed for \eqref{eq:softcoveringS11}. On the other hand, according to the non-feedback measurement compression theorem \cite{wilde2012information_theoretic_costs}, we need a sum rate of at least $I(XR;W)$ to have a faithful measurement simulation. Thus, by setting $|\~K|> 2^{n(I(X R;W) - I(X;W))}$, we define an inter-codebook average state
\begin{align}
    \tau_{\text{avg}}^n \equiv \sum_{k,m,l} \frac{1}{|\~K||\~M|} \Tr_{A^n}\left\{\left(\text{id}_R \otimes \Upsilon_l^{(k,m)}\right)(\Psi^{RA}_\rho)^{\otimes n}\right\} 
            \otimes \sigma_{W^n(m,k,l)}.
\end{align}
Consequently, according to non-feedback measurement compression theorem \cite{wilde2012information_theoretic_costs}, this  inter-codebook average state is a faithful simulation of the ideal product measurement system; i.e., for any $\epsilon_{mc}>0$ and for all sufficiently large $n$, 
\begin{align} \label{eq:faithful_simulation_ineq}
    \Exx{c}{\norm{\tau^{\otimes n} - \tau_{\text{avg}}^n}_1} \leq \epsilon_{mc},
\end{align}
where the expectation is over all codebook realizations.
%%%%%%%%%%%%%%%%%%%%%%%%%%%%%%%%%%%%%%%%%%%%%%%%%%%%%%%
\ifmycomments
\myqu{In general case, the term $\epsilon_{mc}$ is a fixed parameter which is not necessarily related to $\epsilon$. But by looking carefully through the proof of measurement compression theorem, we realize that the $\epsilon_{mc}$ is defined the same way as $\epsilon$ is defined in our system which is the probability of not being in the typical set. So I believe in our specific case $\epsilon_{mc} = \epsilon$. Prove me wrong...}

{ \color{blue} This part was omitted. At first, I followed the notation of Wilde in which he bounds the probability. But in Proof of Prof. Pradhan they used expectation so I changed to that. But here is the previous version for reference:

we can claim that the source coding scheme in the above expression faithfully simulates the product reference and output composite state \eqref{eq:nletterproductstate} with probability close to one:
\begin{align} \label{eq:faithful_joint_state}
    \Pr{\norm{\tau^{\otimes n} - \tau_{\text{avg}}^n}_1 \leq \epsilon} \geq 1 - \delta(\epsilon)
\end{align}
or in the corrected proof of professor pradhan, the expectation of the term was bounded as follows:
\begin{align}
    \Exx{c}{\norm{\tau^{\otimes n} - \tau_{\text{avg}}^n}_1} \leq \epsilon
\end{align}}
\mynote{ if we where to use the probability bound, we have to transform it into bound on expectation first. That is done by showing that the $l_1$ norm is bounded so that it will not go to infinity. And for that we can use the fact that $\tau$ are distributions which means they will be bounded.}
\fi
%%%%%%%%%%%%%%%%%%%%%%%%%%%%%%%%%%%%%%%%%%%%%%%%%%%%%%%
Then, we  bound the expected average distortion as follows:
\begin{align}
    %line 1
      &\Exx{K}{{d_n}^{\{K\}}(\rho^{\otimes n}, \~D_n \circ \~E_n)} = \frac{1}{|\~K|}\sum_k {d_n}^{\{k\}}(\rho^{\otimes n}, \~D_n \circ \~E_n) \nonumber\\
    % line 2
    &= \Tr\Bigg\{
            \Delta^{(n)} \Bigg(
                \sum_{k,m,l} \frac{1}{|\~K||\~M|}
                \Tr_{A^n}\left\{\left(\text{id}^{\otimes n}_R \otimes \Upsilon_l^{(k,m)}\right)(\Psi^{\otimes n}_\rho)^{R A}\right\} \otimes \sigma_{W^n(m,k,l)}
            \Bigg)
    \Bigg\}\nonumber\\
    %line 3
    & = d_{\max}\Tr{\frac{\Delta^{(n)}}{d_{\max}} (\tau_{\text{avg}}^{(n)} - \tau^{\otimes n})}  + \Tr{\Delta^{(n)}\tau^{\otimes n} }\nonumber\\
    % line 4
    & \leq d_{\max}\norm{ \tau_{\text{avg}}^{(n)} - \tau^{\otimes n}}_1  + \Tr{\Delta^{(n)}\tau^{\otimes n} } \nonumber
    %line 5
     \leq d_{\max}\norm{ \tau_{\text{avg}}^{(n)} - \tau^{\otimes n}}_1  + D, \label{eq:expected_average_distortion_inequality}
\end{align}
where $d_{\max}$ is the largest eigenvalue of the distortion observable. The first inequality holds by definition of the trace distance and the fact that $ 0 \leq \frac{\Delta^{(n)}}{d_{\max}} \leq I$. The second inequality holds because the average distortion of $n$ identical copies of the single-letter system is the same as single-letter distortion. Next, we take the expectation of  both sides with respect to all possible codebook realizations. Thus, for all sufficiently large $n$,
\begin{align*}
    \Exx{c}{\Exx{K}{{d_n}^{\{k\}}(\rho^{\otimes n}, \~D_n \circ \~E_n)}} &\leq d_{\max}\Exx{c}{\norm{ \tau_{\text{avg}}^{(n)} - \tau^{\otimes n}}_1}  + D
     \leq d_{\max}\epsilon_{mc} + D,
\end{align*}
where for the first inequality we take the expectation of \eqref{eq:expected_average_distortion_inequality} and the second inequality follows from \eqref{eq:faithful_simulation_ineq}. Further,  the LHS of above inequality can be rewritten as follows by changing the order of expectations,
\begin{align}
\begin{split}
    \Exx{c}{\Exx{K}{{d_n}^{\{K\}}(\rho^{\otimes n}, \~D_n \circ \~E_n)}} &= \Exx{K}{\Exx{c}{ {d_n}^{\{K\}}(\rho^{\otimes n}, \~D_n \circ \~E_n)}}
    = \Exx{c}{{d_n}^{\{k\}}(\rho^{\otimes n}, \~D_n \circ \~E_n)},
    \end{split}
\end{align}
where the second equality holds for any codebook $k\in\~K$ and follows because the expectation of the distance measure over all codebooks is independent of $K$.
Then it is proved that the expected average distortion for any codebook $k\in \~K$ is asymptotically bounded by $D$:
\begin{align}\label{eq:expected_distortion_bound}
    \Exx{c}{{d_n}^{\{k\}}(\rho^{\otimes n}, \~D_n \circ \~E_n)} \leq D + d_{\max}\epsilon_{mc}.
\end{align}

\subsubsection{Intersection of the Constraints}\label{sec:intersection-constraints}
 In this section, we show that the previous bounds on the expected codebook realizations have an intersection with nonzero probability. i.e., there exists a  codebook realization that can realize all events together. The following four cases are the required events in the achievability proof which were proved to hold for the expected codebook realizations. Here by applying the union bound, we ensure that there exists at least one codebook realization that satisfies all the constraints.
 
\ifmycomments
 \mynote{By definition, $P(x^n \notin \~T_\delta^n(X)) >1 - \epsilon$. So the way they did it in Cover's book is that they fix $\delta$ and then $\epsilon$ is also fixed. So for every $\delta$ there exists a small $\epsilon$.  The proof is from the weak law of large numbers. So $\epsilon $ is a fixed number when $n\to \infty$.
\\
 It is also possible to bound with an exponentially decaying function of $n$, instead of fixed epsilon. In that case, we use the Chernoff bound instead of the weak law of large numbers. But that is not needed in this case if we define the parameters correctly.}
\fi

\begin{enumerate}
     
    \item It is shown that the $\Gamma_{w^n}^{(m,k)}$ form valid sub-POVM for all $m\in \~M$ and $k \in \~K$. This is considered as event $E_1$. Using the Chernoff bound, \cite{wilde2012information_theoretic_costs} if $R > I(R;W)_\sigma$ then for some $c>0$, we have 
    \begin{align}
     \Pr{\neg E_1} \leq c \exp{ - 2^{n\delta} \epsilon^3  } .
    \end{align}
    \item Define $E_2$ as the event  $\{S_{11} \leq \exp{-\nu n}\}$ for some $\nu>0 $. Then by applying Markov inequality to expression \eqref{eq:softcoveringS11}, we find the bound
    \begin{align}
     \Pr{\neg E_2} = \Pr{ S_{11} \geq \exp{-\nu n}} \leq \frac{3}{2}\frac{\exp{-tn}}{\exp{-\nu n}}.
    \end{align}

    \item For the bounds on the expectations of $S_{12}$ and $S_{2}$, we let $E_{31}$ and $E_{32}$ to be the corresponding events for these random variables. Then applying the Markov inequality to these inequalities \eqref{eq:ES12}, \eqref{eq:expected_inequality_S2}, we have the following bound for a fixed value $\delta_3>0$:
  \begin{align}
      \Pr{\neg E_{31}\}=\Pr\{ S_2 \geq 2\delta_3} \leq \frac{\Ex{S_2}}{\delta_3} \leq \frac{\epsilon_3}{2\delta_3},\\
      \Pr{\neg E_{32}\}= \Pr\{S_{12} \geq 2\delta_3} \leq \frac{\Ex{S_{12}}}{\delta_3} \leq \frac{\epsilon_3}{2\delta_3}.
  \end{align}
    Note that we used the fact that $\epsilon_2 \leq \epsilon_3$.

    \item Define $E_4$ as the event when the average n-letter distortion constraint is satisfied. By applying Markov inequality to  \eqref{eq:expected_distortion_bound} we obtain for any fixed value $\delta_d>0$ that
    \begin{align}
       \Pr{ \neg E_4} = \Pr{  d(\rho^{\otimes n}, \~D_n \circ \~E_n) \geq D+  \delta_d } \leq \frac{\Exx{c}{ d(\rho^{\otimes n}, \~D_n \circ \~E_n) }}{D+  \delta_d} \leq \frac{D + d_\text{max} \epsilon_{mc}}{D +\delta_d }.
    \end{align}
\end{enumerate} 
    Then the probability of not being in the intersection is bounded by using the union bound
  \begin{align}
    \Pr{\neg{\bigcap_{i=1}^4 E_i}}  &\leq \sum_{i=1}^4 \Pr{\neg{E_i}}
    \leq 
    c \exp{ - 2^{n\delta} \epsilon^3  } 
    + \frac{3}{2}\frac{\exp{-kn}}{\exp{-\nu n}}
    +  \frac{\epsilon_3}{\delta_3}
    + \frac{D + d_\text{max} \epsilon_{mc}}{D +\delta_d }.
  \end{align}
With proper choice of $\delta$ and $\nu \in (0,t)$ ,  the first two terms of the RHS decay exponentially, while $\epsilon_{mc}, \delta_3, \delta_d$ are fixed. Then by a proper choice of the parameters $\epsilon_{mc}, \delta_3$ and $\delta_d$, we ensure that,
\begin{align}
\frac{\epsilon_3}{\delta_3} +  \frac{D+d_{\text{max}}\epsilon_{mc}}{D+\delta_d} <1.\label{eq:asymptoticintersectionprob}
\end{align}
This means there exists with nonzero probability, a valid quantum measurement coding scheme that satisfies all the above four conditions simultaneously.

\subsubsection{Exactly Satisfying IID Output Distribution} \label{sec:exactly_iid}

It remains to prove that the perfect IID output distribution can be achieved from the near-perfect one with an arbitrarily small increase in distortion level.
The desired perfect IID distribution and  the near-perfect induced output distribution for this source coding scheme are expressed by
\begin{align}
    P_{X^n}^{\text{des}}(x^n) &:= \braket{x^n}{(\sigma^{X}_\text{des})^{\otimes n}|x^n}= Q_X^n(x^n) 
    =\sum_{w^n} P_{X|W}^n (x^n|w^n) P_W^n(w^n),\\
    P_{X^n}^{\text{ind}}(x^n) &:= \braket{x^n}{\sigma^{X^n}_\text{ind}|x^n}  
    = \sum_{w^n \in W^n} P_{X|W}^n (x^n|w^n) \Tr{(\text{id}_R\otimes \tilde{\Lambda}^A_{w^n})(\Psi_{RA}^\rho)^{\otimes n}},
\end{align}
where, $(\sigma^{X}_\text{des})^{\otimes n}$ and $\sigma^{X^n}_\text{ind}$ are defined in \eqref{eq:sigma_desired_outputstate} and \eqref{eq:sigma_induced_outputstate}, respectively.
Using \cite[Theorem 1]{wagner_2022_ratedistortionperecptiontradeoff}, with a fixed measurement POVM and by just changing the IID post-processing channel $P_{X|W}$ to a batch decoder $\tilde P_{X^n|W^n}$ we show that one can satisfy the perfect IID condition from the near-perfect one. We define the batch decoder with the conditional probability of any event $A \subseteq \~X^n$ given a measurement outcome sequence $w^n$ as
\begin{align} \label{eq:newdecoderquantummeasurement}
    \tilde P_{\hat X^n|W^n}(A|w^n) = \sum_{x^n \in A \cap \~X^n_+} \theta_{x^n} P_{X|W}^n (x^n | w^n) + P_{X|W}^n (A \setminus \~X^n_+ | w^n)+ \phi_{w^n} Z(A),
\end{align}
where $\~X^n_+ := \left\{      x^n \in \~X^n \Big|  P_{X^n}^{\text{ind}}(x^n) >  P_{X^n}^{\text{des}}(x^n)    \right\}$,  and the given expressions are defined as follows:
\begin{align}
    \theta_{x^n} &:= \frac{ P_{X^n}^{\text{des}}(x^n)}{ P_{X^n}^{\text{ind}}(x^n)} \quad x^n \in \~X^n_+,  \qquad \qquad \quad \phi_{w^n} := \sum_{x^n \in \~X^n_+} ( 1 - \theta_{x^n}) P_{X|W}^n (x^n |w^n),\\
    % line 4
    Z(A) &:= \frac{ P_{X^n}^{\text{des}}(A\setminus \~X^n_+) -  P_{X^n}^{\text{ind}}(A \setminus \~X^n_+)}
    {d_{TV}( P_{X^n}^{\text{ind}},  P_{X^n}^{\text{des}})}.
\end{align}
The validity and admissibility of the new post-processing decoder can be verified with simple calculus, meaning that $ \tilde P_{\hat X^n|W^n}(\~X^n | w^n) = 1$ for all $w^n \in \~W^n$, and the new induced output distribution satisfies the desired IID condition:
\begin{align} 
    \tilde P_{\hat X^n}^{\text{ind}}(A) &:= \sum_{w^n} \tilde P_{\hat X^n|W^n}(A|w^n) \Tr{(\text{id}^R \otimes \tilde{\Lambda}^A_{w^n})(\Psi_{RA}^\rho)^{\otimes n}} 
    = P_{X^n}^{\text{des}}(A).
\end{align}
This way we characterize the $n$-letter decoder of the system for any event $A \in \~X^n$:
\begin{align}
    ~D_n(A | l,m) := \tilde P_{\hat X^n | W^n}(A | W^n(l,m)).
\end{align}
Also, the following set of equalities hold for any $w^n \in \~W^n$ by the definition of the batch decoder in \eqref{eq:newdecoderquantummeasurement},
\begin{align}
    &d_{TV} \Big(\tilde P_{\hat X^n|W^n}(.|w^n), P^n_{X|W}(.|w^n)\Big)\nonumber\\
    % &= \frac{1}{2} \sum_{x^n \in \~X^n} \abs{\tilde P_{\hat X^n|W^n}(x^n|w^n) -  P^n_{X|W}(x^n|w^n)}\\
    &= \frac{1}{2} \sum_{x^n \in \~X^n_+} \abs{\tilde P_{\hat X^n|W^n}(x^n|w^n) -  P^n_{X|W}(x^n|w^n)} + \frac{1}{2} \sum_{x^n \in \~X^n \setminus \~X^n_+} \abs{\tilde P_{\hat X^n|W^n}(x^n|w^n) -  P^n_{X|W}(x^n|w^n)}\nonumber\\
    &= \frac{1}{2} \sum_{x^n \in \~X^n_+} \abs{ (\theta_{x^n}-1) P^n_{X|W}(x^n|w^n) + \phi_{w^n} Z(x^n)}
    + \frac{1}{2}\sum_{x^n \in \~X^n \setminus \~X^n_+} \abs{\phi_{w^n} Z(x^n) }\nonumber\\
    &\stackrel{a}{=} \frac{1}{2} \sum_{x^n \in \~X^n_+} \abs{ (\theta_{x^n}-1) P^n_{X|W}(x^n|w^n) }
        + \frac{1}{2} \sum_{x^n \in \~X^n \setminus \~X^n_+} \abs{\phi_{w^n} \frac{ P_{X^n}^{\text{des}}(x^n) -  P_{X^n}^{\text{ind}}(x^n)}
        {d_{TV}( P_{X^n}^{\text{ind}},  P_{X^n}^{\text{des}})} } = \phi_{w^n},
\end{align}
where $(a)$ follows by the definition of $Z(A)$, and that $Z(x^n) = 0$ for all $x^n \in \~X^n_+$ and the last equality follows from the definition of total variation distance.
Thus, by definition of the total variation, there exists a coupling such that $\Pr(X^n \neq \hat X^n \Big| w^n) \leq \phi_{w^n}$ for all $w^n \in \~W^n$.
Then from the above inequality, using the argument in  \cite{wagner_2022_ratedistortionperecptiontradeoff}, the probability of outputs not being equal is bounded by 
\begin{align}
    \Pr(X^n \neq \hat X^n ) \leq d_{TV}( P_{X^n}^{\text{ind}},  P_{X^n}^{\text{des}})\leq \epsilon_{os},
\end{align}
and the second inequality appeals to \eqref{eq:outputstate-expected-inequality} and the union bound argument in section \ref{sec:intersection-constraints}.
\ifmycomments
{\color{blue}
\begin{align}
% line 1
    \Pr(X^n \neq \hat X^n ) &= \sum_{w^n}\Pr(X^n \neq \hat X^n |w^n ) P_{W^n}(w^n)\\
% line 2
    & \leq \sum_{w^n}\phi_{w^n}P_{W^n}(w^n)\\
% line 3
    &= \sum_{w^n} \sum_{x^n \in \~X^n_+} \left( 1 - \frac{P_{X^n}^{\text{des}}(x^n)
}{P_{X^n}^{\text{ind}}(x^n)
}\right) P_{X|W}^n (x^n |w^n)P_{W^n}(w^n)\\
% line 4
    &=  \sum_{x^n \in \~X^n_+} \left( 1 - \frac{P_{X^n}^{\text{des}}(x^n)
}{P_{X^n}^{\text{ind}}(x^n)
}\right) \sum_{w^n}P_{X|W}^n (x^n |w^n)P_{W^n}(w^n)\\
% line 5
    &=  \sum_{x^n \in \~X^n_+} \left( P_{X^n}^{\text{ind}}(x^n) - P_{X^n}^{\text{des}}(x^n)
\right) \\
% line 6
    &= d_{TV} \Big(  P_{X^n}^{\text{des}}(.), P_{X^n}^{\text{ind}}(.)  \Big)\\
    &\leq \delta
\end{align}
where the last inequality follows from the near-perfect faithfulness.
}
\fi
Next, we bound the n-letter distortion for the new decoder using the above bound. First, note that $\tau_{R_i \hat X_i}$ is the local $i$-th reference-output state of the system, given by
\begin{align}
    \tau^{R_i \hat X_i} &
    = \sum_{x^n}\Tr_{R^{n\setminus \{i\}} A^n}\left\{(\text{id}^{\otimes n} \otimes \hat \Lambda_{x^n})(\Psi_\rho^{R A})^{\otimes n}\right\}\otimes \ketbra{x_i}{x_i}  = \sum_{x_i} Q_X(x_i) \zeta_{x_i}^{R_i} \otimes \ketbra{x_i}{x_i},
\end{align}
where $\zeta_{x_i}^{R_i}$ is the PMR state of the $i$-th local state given the outcome $x_i$, given by
\begin{align*}
\zeta_{x_i}^{R_i} &= \frac{1}{Q_X(x_i)}\braket{\text{id}\otimes x_i|\tau_{R_i X_i}}{\text{id} \otimes x_i}
% incorrect
% &=  \frac{1}{P_X(x_i)} \sum_{j=1:n, j\neq i} \sum_{x_j \in\~X }\Tr_{R^{n\setminus \{j\}} A^n}\left\{(\text{id}^{\otimes n} \otimes \hat \Lambda_{x^n})(\Psi_\rho^{R A})^{\otimes n}\right\},
%correct
= \frac{1}{Q_X(x_i)} \Tr_{R^{n\setminus \{i\}} A^n}\left\{\left(\text{id}^{\otimes n} \otimes \left(\sum_{x^{[n] \setminus i}}\hat \Lambda_{x^n} \right)\right)(\Psi_\rho^{R A})^{\otimes n}\right\}
\end{align*}
where $\hat \Lambda \equiv \{\hat \Lambda_{x^n}\}_{x^n \in \~X^{\otimes n}}$ is the combined encoder-decoder collective POVM induced by the batch decoder. 
We expand the n-letter average distortion as

\begin{align}
% line 1
    \Tr{\Delta^{(n)} \tau^{R^n \hat X^n} } &=    
    % &=\frac{1}{n} \sum_{i=1}^n \Exx{\hat X_i}{\Tr_R \left\{ \zeta_{X_i}^{R_i} \Delta_{R_i}(X_i)\right\}}\\
% line 2
     \frac{1}{n} \sum_{i=1}^n \Exx{\hat X_i}{\Tr_R \left\{ \zeta_{\hat X_i}^{R_i} \Delta_{R_i}(\hat X_i)\right\}\mathds{1}_{\hat X_i = X_i}}  + \frac{1}{n} \sum_{i=1}^n \Exx{\hat X_i}{\Tr_R \left\{ \zeta_{\hat X_i}^{R_i} \Delta_{R_i}( \hat X_i)\right\}\mathds{1}_{\hat X_i \neq X_i}}\nonumber\\
% line 4
    &= \frac{1}{n} \sum_{i=1}^n \Tr{\Delta_{R_i X_i} \tau^{R_i  X_i}} 
    + \frac{1}{n} \sum_{i=1}^n \Exx{\hat X_i}{\Tr_R \left\{ \zeta_{\hat X_i}^{R_i} \Delta_{R_i}( \hat X_i)\right\}\mathds{1}_{\hat X_i \neq X_i}}\nonumber\\
    % line 7
    &\leq D+\epsilon +  \frac{1}{n} \sum_{i=1}^n\Exx{ \hat X_i}{d_\mathrm{max}(\hat X_i) \Trxx{R}{ \zeta_{ \hat X_i}^{R}} \mathds{1}_{\hat X_i \neq X_i}} 
    \leq D + \epsilon + \hat d_\mathrm{max} \epsilon_{os}:= D + \epsilon_4,
\end{align}
where $d_\mathrm{max}(x) = \max \mathrm{eig}(\Delta_R(x)) < \infty$ for $x\in \~X$, and $\hat d_\mathrm{max} = \max_{x\in\~X} d_\mathrm{max}(x)$. 
\ifmycomments
{\color{blue}
(This part uses the uniform integrability ...)
% line 7
\begin{align}
    &\leq D+\epsilon + \sup_{A}  \Exx{ X}{\Tr_R \left\{ \zeta_{ X}^{R} \Delta_{R}(  X)\right\}\mathds{1}_{A}}.
\end{align}
with $A$ being any event with probability $P(A) = P(\hat X_i \neq X_i) \leq \epsilon_{os}$.
By accepting the assumption that the system is uniformly integrable (which is always true for the systems with finite dimensions), it is derived from the above that $\Tr{\Delta^{(n)} \tau_{R^n \hat X^n}} \leq D + \epsilon_4$.
}
\fi

%%%%%%%%%%%%%%%%%%%%%%%%%%%%%%%%%%%%%%%%%%%%%%%%%%%%%%%%%%%
%%%%%%%%%%%%%%%      Section:  Converse  %%%%%%%%%%%%%%%%%%
%%%%%%%%%%%%%%%%%%%%%%%%%%%%%%%%%%%%%%%%%%%%%%%%%%%%%%%%%%%
\subsection{Proof of the Converse} \label{sec:converse}

Let us fix the desired output PMF $Q_X$ and the distortion level $D$.  Further, let $(R,R_c)$ be achievable with the definition \ref{def:achievable_pair}, which means $(R,R_c ) \in \mathrm{interior} \left[ \~R(D, \rho || Q_X)\right]$. This by definition, means that for any $\epsilon'>0$, and all sufficiently large $n$, there exists an $(n,R,R_c)$ coding scheme with $(\~E'_n, \~D'_n)$ that satisfies 
\begin{align}
    X^n \sim Q_X^n, \quad d_n(\rho^{\otimes n}, \~D'_n \circ \~E'_n) \leq D + \epsilon'.
\end{align}
%To establish the single-letter distortion constraint outlined in \eqref{eq:set_feasible_theorem}, 
% We next employ a continuity analysis using the approach of  \cite{tamas_output_constrained_2015}. 
Next, using the convexity lemma \cite[Lemma 1]{tamas_output_constrained_2015} and their argument of right continuity, it implies that the OC rate-distortion function $R(D; R_c, \rho||Q_X)$ is continuous on $D \in [0,\infty)$.
% \begin{lemma}  \label{lemma-convexity of I_Rc}
% $R(D; R_c, \rho||Q_X)$ is a convex function of $D$ on $0<D<\infty$.
% \end{lemma}
% The proof involves the time-sharing method similar to \cite[Lemma 1]{tamas_output_constrained_2015} as provided in \cite{hafez2023thesis}.

% \begin{proof}
%     See Appendix \ref{appA-convexity}.
% \end{proof}
Because the rate pair $(R,R_c)$ is in the interior of the rate region, for any fixed $R_c$ we have $R > R(D;R_c, \rho || Q_X)$. By this and the continuity of $R(D; R_c, \rho||Q_X)$, there exists an $\epsilon > 0$ for which we still have $R > R(D - \epsilon;R_c, \rho || Q_X)$. Therefore, by the definition of achievability, for that specific $\epsilon$, and sufficiently large value of $n$, there exists $(n,R,R_c)$ coding scheme with $(\~E_n, \~D_n)$ that satisfies
% This continuity implies that for any rate pair inside the interior of rate region, $(R,R_c)\in  \~R(D; \rho ||Q_X)$; i.e. $R \geq R(D; R_c, \rho||Q_X)$, there exists an $\epsilon>0$ that still satisfies $R \geq  R(D-\epsilon; R_c, \rho||Q_X)$. As a result, having \eqref{eq:definition_achievability_distortion}, there exists an $(n,R,R_c)$ coding scheme which satisfies
\begin{align} \label{eq:tight-distortion-constraint}
    X^n \sim Q_{X}^n, \quad
     d_n(\rho^{\otimes n}, \~D_n \circ \~E_n)\leq D.
\end{align}

The corresponding n-letter collective measurements $\Upsilon^{(M)}$ selected by a shared random number $M$  on Alice's side, 
 with the relation in eq. \eqref{eq:definition_achievability_distortion}, results in the outcome $L$ which is sent to Bob. Finally, Bob uses a batch decoder $P_{X^n|L,M}(x^n|l,m)$, to generate the reconstructed state.  The resulting n-letter encoding quantum measurement composite state is
\begin{align} \label{eq:nletter_converse_codingscheme}
    \omega^{R^n L M} = \sum_{l,m} \Tr_{A^n}\left\{ (\text{id}\otimes \Upsilon_l^{(m)})(\Psi_\rho^{ R A})^{\otimes n} \right\}
    &\otimes \frac{1}{|\~M|} \ketbra{m}{m}
    \otimes \ketbra{l}{l}.
\end{align}

\subsubsection{Rate Inequalities}
Assuming the above system is achievable in the sense of definition \ref{def:achievable_pair}, using the technique of \cite{wilde2012information_theoretic_costs}, we obtain 
the following rate inequalities:
\begin{align}
% line 1
    n R &\geq H(L)_\omega \geq I(L;M R^n)_\omega= I(L M; R^n)_\omega + I(L; M)_\omega - I(M; R^n)_\omega \nonumber\\
%line 3
    &\stackrel{a}{\geq} I(L M; R^n)_\omega 
    \geq \sum_{k=1}^n I(L M ; R_k)_\omega \nonumber\\
%line 7
    &\stackrel{b}{=} n I(L M ; R_K | K)_\sigma
% %line 8
%     &= n I(L M ; R_K|K )_\sigma + n I(R_K;K)_\sigma\label{eq:addandremoveRKandKmutualinformation}\\
%line 9
    \stackrel{c}{=}n I(L M K; R_K)_\sigma, \label{eq:converse-rate-inequality-R}
\end{align}
where (a) follows because common randomness $M$ is independent of the source, and 
% (b) holds by the sub-additivity of conditional quantum entropy \cite{wilde2013quantum_information_theory_book} and the fact that the source is a product state. 
(b) follows by defining $K$ as a uniform random variable over the set $\{1,2,...,n\}$ which represents the index to the selected system. Then the overall state of the system can be redefined with $K$ being a random index as
\begin{align} 
    \sigma^{R L M K} 
    = \sum_{k,l,m} \Tr_{R^{[n]\setminus k} A^n}\left\{ (\text{id}\otimes \Upsilon_l^{(m)})(\Psi_\rho^{ R A})^{\otimes n} \right\}
    \otimes \frac{1}{|\~M|} \ketbra{m}{m}
    \otimes \ketbra{l}{l} \otimes \frac{1}{n} \ketbra{k}{k}.      \label{eq:overallstatewithKindex}
\end{align}
\ifmycomments
\mynote{They use Alicki-Fannes theorem to show that $I(R;K) \leq \epsilon'$, due to the faithfulness. Refer to \cite{wilde2012information_theoretic_costs} converse for more details. }
\myqu{Why do they use Alicki-Fannes to bound to $\epsilon$? I think the $I(R;K)$ is already equal to zero because input $R$ is tensor state and thus independent of $K$ !} \myans{So I think this $\epsilon$ is not needed in this part, and also not in the next part for our case. I think it was needed only for the $R+S$ in Wilde's paper because the output is near IID there. But the input is perfect i.i.d. so there is no need for $\epsilon $ bound there. Professor Pradhan also confirmed that $I(R;K)=0$ so there is no need for this Alicki-Fannes theorem.}
\fi
Finally, (c) holds because the reference and the index state are independent, i.e., $I(R;K)_\sigma=0$. This can be easily verified by tracing out other systems in \eqref{eq:overallstatewithKindex}:
\begin{align*}
    % line 1
    \sigma^{RK} &= \sum_{k,l,m}\frac{1}{ |\~M|}  \Tr_{R^{[n]\setminus k} A^n} \Big\{     (\text{id}\otimes \Upsilon_l^{(m)})   (\Psi_\rho^{ R A})^{\otimes n}  \Big\} \otimes\frac{1}{n} \ketbra{k}{k}\\
    % % line 3
    % &= \sum_{k,l,m}\frac{1}{ |\~M|}  \Tr_{R^{[n]\setminus k}} \Big\{ \sqrt {\rho^{\otimes n}} \Upsilon_l^{(m)} \sqrt {\rho^{\otimes n}}    \Big\} \otimes\frac{1}{n} \ketbra{k}{k}\\
    % line 5
    &= \sum_{k} \Tr_{R^{[n]\setminus k}} \Big\{\rho^{\otimes n}   \Big\} \otimes\frac{1}{n} \ketbra{k}{k}= \rho^R \otimes \left(\frac{1}{n} \sum_k \ketbra{k}{k}\right)^K.
\end{align*}

\ifmycomments
\myqu{An important question that still holds is that this $R$  state, where is it? Because in the expression there is sum of $R_K$. right?(for similar notations look up page 37, 38 of Wilde's \cite{wilde2012information_theoretic_costs}).}\myans{Unless, after taking the trace with $R_k$, the result be presented as the $R$ space for all of $R_k$. In that case, because reference is a tensor state, the trace result can be fixed and the composite state will be a product state.}

\myqu{In the above expression, the $R^{[n] \setminus k}$ still depends on the value $k$. It is different from the classical case because here each of them is a different space. So you cannot just replace it with a fixed term independent of $k$. Maybe when applying the entropy function in that case it will become independent.}

\myans{We are just taking the average of the states. Refer to meeting 13 for more info.}
\fi

Note that the spaces $\~W$ and $\~X$ are classically quantum registers as they are the outcomes of the measurement. Therefore, the following inequalities hold for the classical entropy and Shannon's mutual information:
\begin{align}
    n(R + R_c) &\geq H(L M) \geq I(L M ; X^n) \nonumber\\
    &\geq \sum_k I(X_k; L M)=n I(X_K; L M | K)\nonumber\\
    &\geq n I(X_K; L M | K) + n I (X_K; K)
    = n I (L M K; X_K).  \label{eq:converse_rate_inequality_Rc}
\end{align}
The arguments for the above inequalities are similar to \eqref{eq:converse-rate-inequality-R}, with the exception that $I (X_K; K)=0$ is directly implied from the assumption that $X^n$ is exactly IID from the statement of the theorem.

Finally, by combining \eqref{eq:converse-rate-inequality-R} and \eqref{eq:converse_rate_inequality_Rc}, and by defining $W := (L, M, K)$, we observe that a quantum Markov chain of the form  $R - (L,M,K) -  X$ exists that satisfies the rate inequalities \eqref{eq:info_inequality1}, \eqref{eq:info_inequality2}. Also, the single-letter encoder POVM $\Lambda$  as applied on an arbitrary source state $\xi^A$ is of the form 
\begin{align*}
    &\xi^A \stackrel{\Lambda}{\longrightarrow} \frac{1}{n |\~M|} \Tr_{A_1^{k-1} A A_{k+1}^n} \left\{    (\text{id}\otimes \Upsilon_l^{(m)})^{A^n}(\Psi_\rho^{ R A})^{\otimes k-1} \otimes \xi^A \otimes  (\Psi_\rho^{ R A})^{\otimes n-k}  \right\}.
\end{align*}
Then by using the tensor decomposition of the collective measurement $\Upsilon_{l}^{(m)} = \sum_{t} \left(\Upsilon_{l,t}^{(m,1)}\right)^{A_1} \otimes \cdots \otimes 
\left(\Upsilon_{l,t}^{(m,n)}\right)^{A_n}$, then according to the composite state \eqref{eq:overallstatewithKindex}, 
the single-letter encoder POVM is given by 
\[
\Lambda_{w:=(l,m,k)} :=  \frac{1}{n|\~M|}\sum_{t} 
\left(\Upsilon_{l,t}^{(m,k)}\right)^{A_k}
\prod_{i\neq t} \Tr[\left(\Upsilon_{l,t}^{(m,i)}\right)^{A_i} \rho^{A_i}], 
\]
and the randomized decoder is provided below
\begin{align*}
    &P_{X|W}(x|w:=(l,m,k)) :=  \sum_{x^{[n]\setminus i} \in \~X^{n-1}} D_n(x^{k-1},x,x^{n-k}|l,m), \quad \forall x\in \~X, w\in \~W.
\end{align*}

\subsubsection{Distortion Constraint}

Using the tight distortion bound in \eqref{eq:tight-distortion-constraint}, we provide the following bound on the single-letter distortion which completes the proof for the converse of the Theorem \ref{th:maintheorem}:
\begin{align}
    d(\rho, \Delta_{RX}) &= \Exx{X}{\Tr{\rho_X \Delta_R(X)}} = \frac{1}{n}\sum_{k=1}^n \Exx{X}{\Tr{\rho_X \Delta_R(X)}|K=k} = \Exx{X^n}{\Tr{\rho_{X^n} \Delta^{(n)}(X^n)}} \leq D.
\end{align}
The cardinality bound on the alphabet of the auxiliary random variable in provided in Appendix \ref{appE}.

%%%%%%%%%%%%%%%%%%%%%%%%%%%%%%%%%%%%%%%%%%%%%%%%%%%%%%%%%%%%%%%%%%
%%%%%%%%%%%%      Section:  Cardinality Bound   %%%%%%%%%%%%%%%%%%
%%%%%%%%%%%%%%%%%%%%%%%%%%%%%%%%%%%%%%%%%%%%%%%%%%%%%%%%%%%%%%%%%%

\section{Continuous-Variable Quantum System} \label{sec:continuous-systems}
%%%%%%%%%%%%%%%%%%%%%%%%%%%%%%%%%%%%%%%%%%%%%%%%%%%%%%%%%%%%%%%%%
%%%%%%%%%%%      Section:  Continuous case     %%%%%%%%%%%%%%%%%%
%%%%%%%%%%%%%%%%%%%%%%%%%%%%%%%%%%%%%%%%%%%%%%%%%%%%%%%%%%%%%%%%%
% \myqu{Need for discussing the issue of reference and purification in the continuous space.}

In this section, we consider the measurement coding for the CV quantum systems with the infinite-dimensional separable Hilbert space $\~H_A = \~L^2(\mathbb{R})$  \cite[Chapters 11, 12]{Holevo_2019_quantum_systems_book}. 
The proof of the achievability of the random coding argument in previous sections does not directly apply to the continuous quantum systems. The first reason is that the operator Chernoff bound as defined in \cite{wilde2012information_theoretic_costs},
is applicable only for a finite-dimensional Hilbert space. 
Secondly, in infinite-dimensional systems, an observable with a non-discrete set of outcomes cannot define a quantum channel. 
% This is because the set of coherent states forms an overcomplete set.
Therefore, it is not possible to represent the outcome space using quantum registers defined on separable Hilbert spaces \cite{Holevo_2019_quantum_systems_book}.
% This is in contrast to the finite-dimensional system for which, the set of all outcome states forms a complete orthonormal set.
% As a result, the quantum mutual information is not defined for such continuous measurement systems. 
Instead, we keep the output system as classical and use the generalized ensemble representation.
In order to properly define the continuous system model, we first provide the following generalized definitions. 

\subsection{Generalized Definitions of Continuous Quantum Systems} \label{subsec:generalized-definitions}

\begin{definition}
\cite[Definition 11.22]{Holevo_2019_quantum_systems_book} 
The generalized ensemble is defined as a Borel probability measure $\pi$ on the subspace of density operators $\~G(\~H_A)$. Then the average state of the ensemble is defined as $\bar \rho_\pi = \int \rho \cdot \pi(d\rho )$.

\end{definition}
In contrast to the finite-dimensional Hilbert space for which the POVM is defined for all possible outcomes, in the continuous quantum measurement systems, the generalized POVM is defined over the subset of $\sigma$-algebra of Borel subsets $\~B$.
\begin{definition}
\cite[Definition 11.29] {Holevo_2019_quantum_systems_book} 
A POVM is generally defined on a measurable space $\~X$ with a $\sigma$-algebra of measurable subsets $\~B(\~X)$, as a set of Hermitian operators $M = \{ M(B), B \in \~B(\~X) \}$ satisfying the following conditions:
    (I) $M(B)\geq 0, B \in \~B(\~X)$,
    (II) $M(\~X) = I$, and 
    (III)  For any countable (not necessarily finite) decomposition of mutually exclusive subsets $B = \cup B_j$ ($B_i \cap B_j = \emptyset, \; i \neq j$), the sum of the measures converge in weak operator sense to measure of the combined set; i.e. $M(B) = \sum_j M(B_j)$.

\end{definition}
An observable $M$ acting on a CV quantum state $\rho$, with outcomes in measurable space  $\~X$, results in the following probability measure $\pi_\rho^M (B) = \Tr{\rho M(B)}$ for all $ B\in \~B(\~X)$.
It is also necessary to have a proper definition of post-measurement states. The a posteriori average density operator for a subset $B\in \~B$ is defined in \cite{ozawa1985_aposteriori_state} for a general POVM $M$ as 
\begin{align}
    \rho_B = \frac{\sqrt{M(B)} \rho \sqrt{M(B)}}{\Tr{\rho M(B)}}.
\end{align}
Based on that, Ozawa defines the post-measurement state for a continuous quantum system in the following proposition.
\begin{proposition}
\cite[Theorem 3.1.]{ozawa1985_aposteriori_state} \label{prop:aposteriori-density-operator-ozawa}
For any observable $M$ and input density operator $\rho$, there exists a family of a posteriori density operators $
\{\rho_x; x\in \~X\}$, defined with the following properties:
    (I) For any $x \in \~X$, $\rho_x$ is a density operator in $\~H_A$;
    (II) The mapping $x \mapsto \rho_x$ is strongly Borel measurable;
    (III) For any arbitrary observable $N$ with outcome space $\~Y$, and any Borel sets $A\in \~B(\~Y)$ and $B \in \~B(\~X)$, the joint probability is given by
    \begin{align*}
        \Pr(X\in B, Y \in A) =  \Tr{\sqrt{M(B)} \rho \sqrt{M(B)} N(A)}
        = \int_B \Tr{\rho_x N(A)} \pi(dx),
    \end{align*}
    where $\pi(B) = \Tr{M(B) \rho}, \; B\in \~B$ is the probability measure of the outcome space.
\end{proposition}

If the $M$ POVM is such that the measure $M(B)$ is absolutely continuous with respect to some measure $\mu$ on $\~X$, then there exists a weakly measurable function $y \mapsto m(y)$ with values in the cone of bounded positive operators of $\~H$ (Radon-Nikodym derivative) such that \cite[Section IV]{holevo2020-gaussian-maximizers-observables}:
\begin{align*}
    \~E': \qquad \pi(B) = \Tr{\rho M(B)}, \quad \rho_y = \frac{\rho^{1/2} m(y) \rho^{1/2}}{\Tr{\rho m(y)}}, \quad M(B) = \int_B m(y) \mu(dy).
\end{align*}

Using the definition of post-measured ensemble one can define the information gain introduced by Groenwold \cite{groenewold1971_information_gain_measurement} for an input state $\rho$ and output ensemble $\{\rho_B, \pi_\rho^M(B)\}_{B\in \~B}$ as \cite{Holevo_2019_quantum_systems_book}:
\begin{align}
    I_g(\rho, X)  := H(\rho) - \int_\~X H(\rho_x) \pi_\rho^M(dx).
\end{align}
It is worth mentioning that the information gain equals the quantum mutual information of the finite-dimensional QC systems. 
We further redefine the Definition \ref{def:source-coding} (the discrete source-coding scheme) to match the CV quantum systems as follows.
\begin{definition}\label{def:source-coding-continuous}
    An $(n,R,R_c)$ source-coding scheme for the continuous quantum-classical system is comprised of an encoder $\~E_n$ on Alice's side and a decoder $\~D_n$ on Bob's side, with the 
 detailed description provided in the Definition \ref{def:source-coding}. The final output sequence $X^n$ is generated by the decoder in the output space $\~X^n$ with the probability measure $\{\pi_{X^n}(B), \; B\in \~B(\~X^n)\}$. Therefore, the average PMR state $\{\hat \rho^{R^n}_B\}_{B \in \~B(\~X^n)}$ along with the probability measure, form the output ensemble over their corresponding Borel subset $\~B(\~X)$, where 
\begin{align} \label{eq:post-measured-reference-nletter-continuous}
    \hat \rho^{R^n}_{B}
    &=\frac{1}{P_{X^n}(B)} \sum_{m,l} \frac{1}{|\~M|} \Tr_{A^n}\left\{(\text{id} \otimes \Upsilon^{(m)}_l)[\Psi^\rho_{RA}]^{\otimes n})\right\} \mathcal{D}_n(B|l,m) ,\\
    \pi_{X^n}(B) &= \sum_{m,l} \frac{1}{|\~M|} \Tr{ \Upsilon^{(m)}_l\rho^{\otimes n}} \mathcal{D}_n(B|l,m).
\end{align}
\end{definition}

We further define the average n-letter distortion for the source coding system with encoder-decoder pair $(\~E_n, \~D_n)$ as the average single-letter distortion of the local PMR states, given the set of distortion observable operators  $\Delta(x), x\in \~X$, and a continuous memoryless source  state $\rho^{\otimes n}$ as
\begin{align}   \label{eq:average-distortion-definition-continuous}
     d_n(\rho^{\otimes n}, \~D_n \circ \~E_n) &:= \frac{1}{n} \sum_{i=1}^n \Exx{\pi_{X_i}}{\Tr{\hat \rho_{X_i}^{R_i}\Delta_R(X_i) }} = \frac{1}{n} \sum_{i=1}^n \int_{x\in  \mathbb R} \Tr_{R_i} \left[\hat\rho_x^{R_i} \, \Delta_R(x)\right] \pi_{X_i}(dx),
\end{align}
where $\hat \rho_{x_i}^{R_i} := \Exx{x^{n \setminus [i]}}{\Trxx{n\setminus [i]}{\hat \rho^{R^n}_{X^n}}}$ is the $i$-th local PMR state, conditioned on local outcome $x_i\in\~X$ and $\pi_{X_i}$ is the marginal probability measure of the $i$-th local outcome. 

\begin{definition} \label{def:uniform-integ}
Consider a QC system with a distortion observable $\Delta_{RX}$ with operator mapping $x \mapsto \Delta_R(x), x\in \~X$, and an input quantum state $\rho$ forming $(\Delta_{RX}, \rho)$. The pair is called uniformly integrable if for any $\epsilon>0$ there exists a $\delta>0 $ such that 
\begin{align}
    \sup_{\Pi} \sup_{\Lambda} \Exx{ X}{\Tr_R \left\{ \Pi_X\rho_{ X}^{R} \Pi_X \Delta_{R}(  X)\right\}} \label{eq:unif_integ_def} \leq \epsilon,
\end{align}
where the supremum is over all POVMs $\Lambda \equiv \{\Lambda_x\}_{x\in \~X}$ and all projectors of the form $\Pi=\sum_{x} \Pi_x \otimes \ketbra{x}{x}$ such that $\Exx{X}{\Tr(\rho_X\Pi_X)} \leq \delta$, and $\rho^R_x$ is the PMR state of $\rho$ given the outcome $x$ with respect to $\Lambda$. 
\end{definition}

\begin{remark}
For the continuous-variable quantum system, we 
make the following assumptions. 
(I) The density operator $\rho$ belongs to a compact subset $K$ of $\~G(\~H)$ (see \cite[Definition 11.2]{Holevo_2019_quantum_systems_book}. 
(II) The system is uniformly integrable. 
(III) the operator mapping of the distortion $x \mapsto \Delta_R(x)$ is uniformly continuous with respect to the trace norm.
\end{remark}
\subsection{Main Results: Achievable Rate Region for CV Quantum Systems}\label{subsec:theorem-continuous} 

In the light of the above definitions and formulations, we redefine the Definition \ref{def:achievable_pair} of achievability for the continuous systems as follows.

\begin{definition} \label{def:achievable_pair-continuous}
Given a probability measure $\pi_X$ on $(\~X, \~B(\~X))$ and a distortion level $D$, and assuming a product input state of $\rho^{\otimes n}$ of infinite-dimensional separable Hilbert space, a rate pair $(R,R_c)$ is defined as achievable if, for any sufficiently large $n$ and any positive value $\epsilon>0$, there exists an $(n, R, R_c)$ coding scheme  comprising of a measurement encoder $\~E_n$ and a decoder $\~D_n$ as described in Definition \ref{def:source-coding-continuous} that satisfy the following conditions,
\begin{align}
    X^n \sim  \pi_X^n,  \quad 
     d_n(\rho^{\otimes n}, \~D_n \circ \~E_n) \leq D+\epsilon.
\end{align}
\end{definition}
The following theorem provides a single-letter characterization of the achievable rate region for the continuous quantum system.
\begin{theorem}\label{th:continuous-theorem}
    Given a pair $(\pi_X, D)$ and having a product input state $\rho^{\otimes n}$ of continuous infinite-dimensional Hilbert space with limited von Neumann entropy, a rate pair $(R,R_c)$ is inside the achievable rate region in accordance with the  definition \ref{def:achievable_pair-continuous} if and only if there exists an intermediate state $W$ with a corresponding measurement POVM $M_W= \{M_W(B), B\in \~B_\~W\}$ where $\~B_\~W$ is the $\sigma$-algebra  of the Borel sets of $\~W$, and randomized post-processing channel $P_{X|W}$  which satisfies the rate inequalities
    \begin{align} \label{eq:info_inequality1_continuous}
		    R &\geq I_g(R;W), \\
		    R+R_c &\geq  I(W;X),\label{eq:info_inequality2_continuous}
    \end{align}
    where $W$ constructs a quantum Markov chain of the form $R - W - X$, generating the ensemble $E_w := \{\hat \rho_w, \pi_W(w)\}_{ w\in \~W}$  on the intermediate space and the ensemble $E_x := \{\hat \rho_x, \pi_X(x)\}_{x\in \~X}$  on the output space, (as described by Proposition \ref{prop:aposteriori-density-operator-ozawa}) according to the following set:
    \begin{align} \label{eq:set_feasible_theorem_continuous}
		    \~M_c(\~D) = \left\{ (E_w, E_x) \middle\vert 
		    \begin{array}{l}
                    \int_w P_{X|W}(A|w) \pi_W(dw)   =  \pi_X(A),  \quad \text{for }A\in \~B(\~X)\\
                    \qquad \qquad \ \pi_W(B) = \Tr{ M(B) \rho}, \quad \text{for }B\in \~B(\~W)\\
                    \int_{x\in  \mathbb R} \Tr_R \left[\rho_x \, \Delta_R(x)\right] \pi_X(dx) \leq D
            \end{array}
            \right\}.
    \end{align}

\end{theorem}
Note that the cardinality bound does not exist in the continuous case.
\begin{remark}
    The classical channel $P_{X|W}: \~W \times \~B(\~X) \to \~X$, is a mapping such that for every $w \in \~W$, $P_{X|W}(.|w)$ is a probability measure on $\~B(\~X)$ and for every $B\in \~B(\~X)$, $P_{X|W}(B|.)$ is a Borel-measurable function. 
\end{remark}
\begin{remark}
    The converse of the finite quantum system is applicable to the continuous variable quantum systems, except for the cardinality bound.
\end{remark}

\subsection{Proof of Achievability in the Continuous System}\label{subsec:achievability-continuous}
 Given the pair $(\pi_X, D)$, assume there exists a continuous intermediate state $W$ forming a quantum Markov chain $R - W - X$, with a corresponding continuous POVM $M_W= \{M_W(B), B\in \~B(\~X)\}$ with outcomes in $\~W$, defined as a set of Hermitian operators in $\~H_A$ satisfying conditions of Theorem \ref{th:continuous-theorem}, along with a corresponding classical post-processing channel $P_{X|W}$.
 
 % Then according to Proposition \ref{th:aposteriori-density-operator-ozawa}, for any proper measurement POVM $M$, there exists a family of post-measurement density operators $\rho_w$ along with a probability measure $\mu^M_\rho$ (or equivalently, $\rho_x$ along with $\mu_X$ if we consider the overall POVM $\Lambda(B) \equiv \int_\~W M(dw) P(B|w), \; B\in \~B(\~X)$), such that 
% \begin{align}
%     % &I(R;W) < R  \label{eq:continuous_rate_const1}\\
%     % & I(W;X) < R + R_c \label{eq:continuous_rate_const2}\\
%     &\int_w P_{X|W} (A|w) \mu^M_\rho (dw) = \mu_X(A), \quad \text{for all }A \subseteq \~B(\~X),\\
%     &d(R,X):= \int_x \Tr{\rho_x \Delta_{R}(x)} \mu_X(dx) \leq D.
% \end{align}

In this proof, we plan to employ our discrete source coding theorem from the previous section, by first applying a clipping projection to truncate the source state into a finite-dimensional space, and then using Theorem \ref{th:maintheorem}.
Referring to  \cite[Theorem 11.2]{Holevo_2019_quantum_systems_book}, the considered compact subset $K$ of density operators
has the property that for any $\varepsilon > 0 $, there exists a finite-rank projector $P_\varepsilon$ such that $\Tr{P_\varepsilon S} \geq 1 - \varepsilon$ for all $S \in K$.

Thus, similar to the proof of \cite[Lemma 11.5]{Holevo_2019_quantum_systems_book}, we can form a finite-rank spectral projection $\Pi_{k_1}$ onto the eigenspace corresponding to the first $k_1$ eigenvalues of the source state, so that $\Pi_{k_1} \uparrow I$.
This spectral projection obviously commutes with the source state $\rho^A$. Then based on the projection operator, the following sharp measurement is formed,
$C_{k_1} \equiv \left\{\Pi_{k_1} , \quad I - \Pi_{k_1} \right\}$.
An example of this projection is the truncation of the Fock basis to $\{\ket{n}\}_{n=0}^N$ for a thermal Gaussian state using the energy test \cite{ghorai2019asymptotic,lin2019asymptotic,renner2009finetti}. 
The outcome of the clipping measurement is the clipping error
\begin{align} \label{eq:indicatorAk1}
    A_{k_1} := \begin{cases}
            0 & \text{if } \rho_A \in \Pi_{k_1}\\
            1 & \text{otherwise}
        \end{cases},
\end{align}
and the probability of the clipping error
    $P_{k_1} := \Pr(A_{k_1} = 1 ) = \Tr{(I - \Pi_{k_1}) \rho_A}$.
Then the conditional post-clipping state given $A_{k_1} = 0, 1$ are respectively: 
\begin{align}\label{eq:a-posteriori-clipping}
    \hat\rho^{A'}_0 = \frac{\rho_A \Pi_{k_1}}{\Tr{\Pi_{k_1} \rho_A}}, 
    \quad \hat\rho^{A'}_1 = \frac{\rho_A (I - \Pi_{k_1})}{1 - \Tr{\Pi_{k_1} \rho_A}},
 \end{align}
where we used the fact that the projection operator $\Pi_{k_1}$ commutes with source state $\rho_A$.
% {\color{blue} The continuous quantum system can be represented by $\{\ket{n}\}_{n=0}^\infty$, the number operators of the Fock basis, which is a countable infinite dimensional Hilbert space. We plan to use the source coding theorem of the discrete system from the previous section. To use Theorem \ref{th:maintheorem}, we first follow a similar approach as \cite{ghorai2019asymptotic,lin2019asymptotic,renner2009finetti} to perform a clipping projection, that truncates the state into the finite-dimensional space of the first $k_1+1$ Fock states via an energy test:
% \begin{align} \label{eq:gentle-clipping-povm}
%         C_{k_1} \equiv \left\{\Pi_{k_1} := \sum_{n=0}^{k_1} \ketbra{n}{n}, \quad I - \Pi_{k_1} \right\}.
% \end{align}
% Therefore, for any small $\epsilon_c >0$, there exists a large enough $k_1$ such that the probability of the state projecting to the first subspace is $\epsilon_c$-close to unity, 
% $    \Tr{\rho \Pi_{k_1}} \geq 1 - \epsilon_c$. For a detailed description of the spectral decomposition of CV quantum systems, refer to \cite{holevo_hirota1999capacityofquantumgaussianchannels, weedbrook_gaussian_2012}.
% }
%%%%%%%%%%%%%%%%%%%%%%%%%%%%%
%%%%%%%%%%%%%%%%%%%%%%%%%%%%%

\subsubsection{Information Processing Task}
It is possible to extend the original Markov chain $R - W -X$ by quantizing the classical output $X$ with quantizer $Q_{\bar{k}_2}$ where $\bar{k}_2 = (k_2, k_2')$ is the pair of clipping region parameters creating the cut-off range $[-k_2, k_2]$ and $k_2'$ is the precision parameter making $2^{k_2'}$ levels forming the quantized output space $\~X_{\bar{k}_2}$, such that
\begin{align}
    R \stackrel{M_W}{\longrightarrow} W \stackrel{P_{X|W}}{\longrightarrow} X \stackrel{Q_{\bar{k}_2}}{\longrightarrow} X_{\bar{k}_2}.
\end{align}

However, the $R$ system of the above quantum Markov chain cannot be easily extended to the clipped space $R_{k_1}$, because the forward path from $R_{k_1}$ to $R$ would require performing the inverse of the clipping projection POVM which is not straightforward.
Instead, using a new approach, we prepare a separate Markov chain by clipping the quantum source state with clipping projection $C_{k_1}$ and then directly feeding that clipped state $R_{k_1}$ into the same continuous measurement POVM $M_W$, as shown in figure \ref{fig:continuous-markov-chain}. This is specifically possible as the clipped input state lies in a subspace of the original quantum Hilbert space.

\begin{figure}
    \centering
    \includegraphics[width=0.8\textwidth]{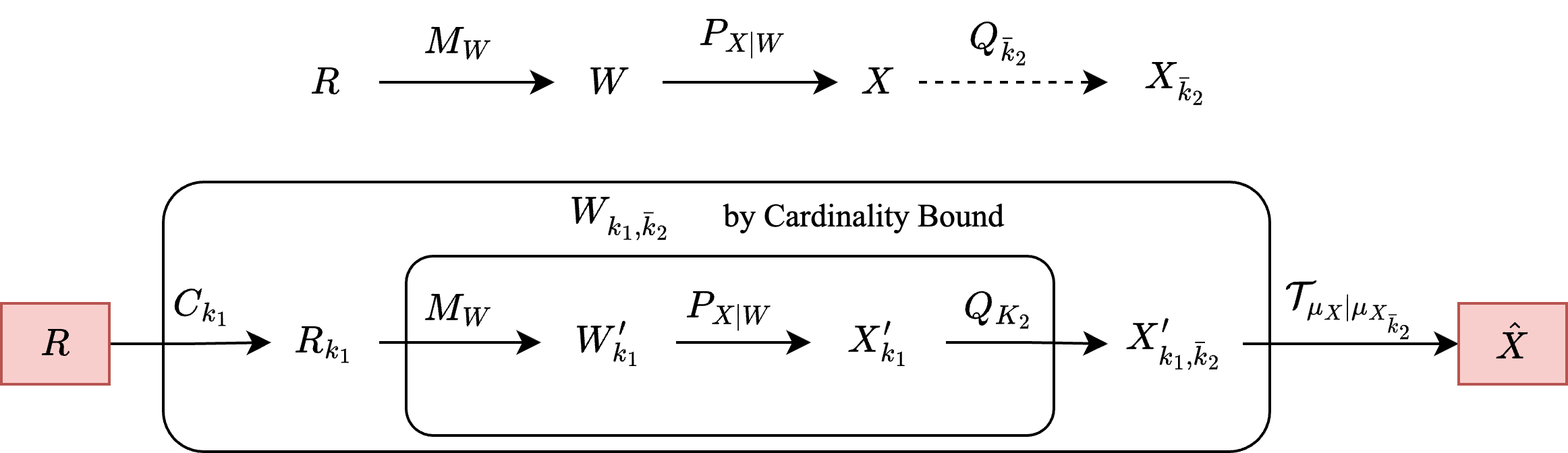}
    \caption[Markov Chain of the Clipping Approach in Continuous System]{Markov Chain of the new Approach: The upper diagram shows the original Markov chain provided by the single-letter intermediate state $W$. The lower diagram shows the single-letter Markov chain of the new approach in which the clipped source state is directly fed into the same continuous measurement $M_W$ and the discrete output is obtained by quantizing $X'_{k_1}$. Finally, the optimal transport transforms the discrete output back to the continuous output $\hat X$.}
    \label{fig:continuous-markov-chain}
\end{figure}

%%%%%%%%%%%%%%%%%%%%%%%%%%%%%%%%%%%%%%

After applying the clipping measurement, the system has two possible scenarios which can be helpful in subsequent analysis. In the first scenario, the system reveals the outcome $A_{k_1}$, and based on its value, if the source state is inside the clipping subspace $(A_{k_1} = 0)$, then $\hat \rho^{A'}_0$ is sent to  $M_W$ POVM to produce the classical outcome. Otherwise if $(A_{k_1} = 1)$, we throw out the a-posteriori state $\hat \rho^{A'}_1$ and only notify the receiver by asserting $A_{k_1}$. This scenario produces the following new Markov chain:
\begin{align} \label{eq:Markov-chain-alternative-method}
    R\stackrel{C_{k_1}}{\longrightarrow} R_{k_1} \stackrel{M_W}{\longrightarrow} W'_{k_1} \stackrel{P_{X|W}}{\longrightarrow} X'_{k_1} \stackrel{Q_{k_2}}{\longrightarrow} X'_{k_1,\bar{k}_2} .
\end{align}

The second scenario is when the outcome $A_{k_1}$ is not revealed and thrown away. Then the system performs $M_W$ POVM on the post-projection state. Using \eqref{eq:a-posteriori-clipping}, the final conditional a-posteriori average density operator for an event $B \in \~B$ given the clipping POVM outcomes are
 \begin{align*}
      \hat\rho_{0}^{A''}(B) = \frac{\sqrt{M_W(B)} \rho_A \Pi_{k_1} \sqrt{M_W(B)}}{\Tr{M_W(B) \rho_A\Pi_{k_1}} }, \quad
      \hat\rho_{1}^{A''}(B) = \frac{\sqrt{M_W(B)} \rho_A (I - \Pi_{k_1}) \sqrt{M_W(B)}}{\Tr{M_W(B) \rho_A ( I - \Pi_{k_1})} }.
 \end{align*}
 Because $A_{k_1}$ is hidden, the average final a-posteriori density operator is
 \begin{align}
    % line 1
     \hat \rho^{A''}(B) = \hat\rho_{0}^{A''}(B) \cdot \Pr{A_{k_1} = 0 | B} + \hat\rho_{1}^{A''}(B) \cdot \Pr{A_{k_1} = 1 | B} =  \frac{\sqrt{M_W(B)} \rho^A \sqrt{M_W(B)}}{\Tr{M_W(B) \rho^A}}. 
 \end{align}
The above equality shows that by throwing out the outcome $A_{k_1}$, the system acts as if no clipping projection was performed, which reclaims the original Markov chain $R - W - X$. 
\begin{remark}
The importance of this second scenario is revealed in the next subsection while proving the rate inequalities. It is needed to establish a single probability space in which both unclipped outcome $W$ and $A_{k_1}$ exist. However, in reality, the continuous protocol only employs the first scenario.
\end{remark}

%%%%%%%%%%%%%%%%%%%%%%%%%%%%%%%%%%%%%%
% \subsubsection{Proof of Rate Inequalities}
As we had to define a separate Markov chain, we cannot directly harness the data processing inequality to prove the rate inequalities for the quantized and clipped systems. Therefore we provide the following proposition to show the rate inequalities still hold after clipping and quantization.
\begin{proposition} \label{prop:continuous_rate_inequalities}
    Suppose having a quantum Markov chain of the form $R-W-X$ satisfying the conditions in Theorem \ref{th:continuous-theorem}. Using the new clipping method as described by \eqref{eq:Markov-chain-alternative-method}, the clipped states still satisfy the following rate inequalities in the asymptotic regime
\begin{align}
   \lim_{k_1 \to \infty} I_g( R_{k_1}; W'_{k_1}) \leq I_g(R;W), \quad \lim_{k_1 \to \infty} I(  W'_{k_1}; X'_{k_1,\bar{k}_2}) \leq I(W;X). \label{eq:clippedrate-inequality}
\end{align}
\end{proposition}
\begin{proof}
Proof in Appendix \ref{appA-cont-rate-inequalities}
\end{proof}

 \subsubsection{Source-Coding Protocol for Continuous States}
Having a quantum source generating a sequence of $n$ independent continuous states $\rho^{\otimes n}$, we apply a coding protocol on the source states, described as follows. First, separate the input states into \textit{proper} and \textit{improper} states (the states residing inside and outside the clipping region respectively) by applying the clipping POVM $C_{k_1}$. This generates a  sequence of error bits $A^n_{k_1}\equiv \{A_{i,k_1}\}_{i=1}^n$ where $A_{i,k_1}$ is the clipping error of the $i$-th state as defined in \eqref{eq:indicatorAk1}. 
Thus, according to \textit{Weak Law of Large Numbers}, for any fixed $\epsilon_{cl}>0$ and $k_1 \in \mathbb{N}$, there exists a value $N_0(\epsilon_{cl}, k_1)$ large enough such that for any $n\geq N_0(\epsilon_{cl}, k_1)$, the number of proper states
\begin{align}
    T:=  \sum_{i=1}^n \left(1 - A_{i,k_1}\right),
\end{align}
is within the range $T \in [n(1 - P_{k_1} \pm \epsilon_{cl})]$ with probability no less than $1 - \epsilon_{cl}$. Then for any sequence with $T < t_\text{min}$ where $t_\text{min}:= n(1 - P_{k_1} -\epsilon_{cl})$, we do not perform source coding,  and instead assert a \textit{source coding error} event $E_{ce}$. Upon receiving the coding error event, Bob will locally generate a sequence of random outcomes $\hat X^n_\text{local}$ with the desired $\mu_X^n$ output distribution. This ensures that in every sequence for which the coding is performed, there are $T \geq t_\text{min}$ independent source states which are inside the clipping region, for which we apply the coding scheme. For the rest of the $(n-T)$ states, we do not perform the coding, instead, we throw away the source state and send the error-index to the receiver. 

Next, the error bits sequence is coded into indices of size ${\binom{n}{t_\text{min}}}+1$ where the extra index is the event of a source coding error $E_{ce}$. Note here that the required classical rate, in this case, will be $R + \log_2\left(  {\binom{n}{t_\text{min}}}+1  \right)$. The following limit
\begin{align}
    \lim_{\epsilon_{cl}\to 0} \lim_{k_1 \to \infty} \lim_{n \to \infty}  \frac{1}{n}\log_2\left(  {\binom{n}{n(1 - P_{k_1} -\epsilon_{cl})}}+1  \right) = 0,
\end{align}
ensures that the extra error handling rate can be made arbitrarily small.

Then at Bob's side, the classical sequence $X'^t_{k_1,\bar{k}_2}$ is constructed using the discrete coding scheme, and is fed to a memoryless optimal transport block $\~T_{\mu|\mu_{\bar{k}_2}}$ to generate the final continuous sequence $\hat X^t_{k_1,\bar{k}_2}$. Finally, the sequence is padded with the $(n-T)$ locally generated independent values at the error positions to create the final $\hat X^n_{k_1,\bar{k}_2}$. Bob then uses this sequence to prepare his final quantum states. Figure \ref{fig:coding-protocol} shows the block diagram of this coding protocol.
\begin{figure}[ht]
    \centering
    \includegraphics[width=\linewidth]{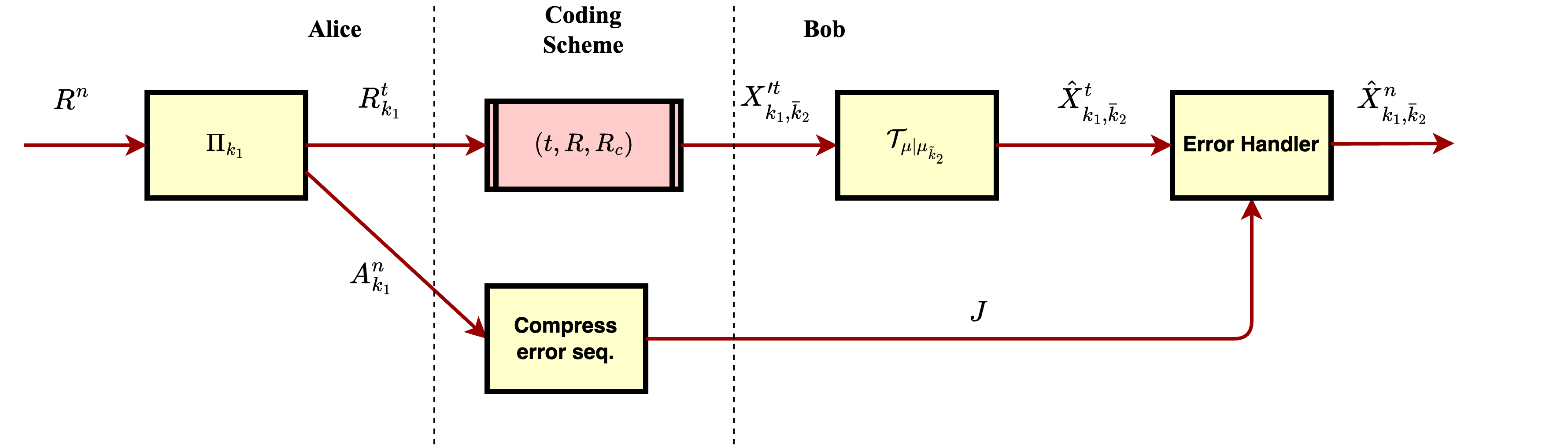}
    \caption[Coding Protocol]{Continuous Coding Protocol}
    \label{fig:coding-protocol}
\end{figure}

\subsubsection{Proof of Distortion Constraint}
The end-to-end average distortion for the above system is written as
\begin{align}
    % line 1
    d_n(R^n, {\hat X}^n_{k_1,\bar{k}_2}) &= d_n\left(R^n, {\hat X}^n_{k_1,\bar{k}_2} \Big|E_{ce}\right) \Pr(E_{ce}) + d_n\left(R^n, {\hat X}^n_{k_1,\bar{k}_2}\Big|\neg E_{ce}\right) \left( 1-   \Pr(E_{ce})  \right) \nonumber\\
    % line 2
    &\leq  d_n\left(R^n, \hat X_\text{local}^n| E_{ce} \right) \Pr(E_{ce}) + d_n\left(R^n, {\hat X}^n_{k_1,\bar{k}_2} \Big| \neg E_{ce}\right) \nonumber\\
    % line 3
    &= \frac{1}{n} \sum_{i=1}^n  d\left(R_i, \hat X_{i,\text{local}}| E_{ce} \right) \Pr(E_{ce}) + d_n\left(R^n, {\hat X}^n_{k_1,\bar{k}_2} \Big| \neg E_{ce}\right) ,
\end{align}
where $\hat  X_\text{local}^n$ is generated locally at Bob's side according to the fixed IID output distribution $\mu_X$ in the event of coding error $E_{ce}$. Therefore, in the first term above, for each $i$-th sample of the system, the uniform integrability of the distortion observable implies that it can be made arbitrarily small by selecting the proper value of $\epsilon_{cl}$. As for the second term, we use the following lemma to provide a single-letter upper bound:
\begin{lemma} \label{lemma:single-letter-distortion-upperbound}
    The end-to-end average n-letter distortion of the continuous system conditioned on the event of no coding error is bounded from above by the following single-letter distortion for any fixed value of $k_1, k_2, k_2' > 0$ and $\bar{\varepsilon}>0$ as:
\begin{align}
    d_n(R^n, {\hat X}^n_{k_1,\bar{k}_2}|\neg E_{ce}) 
    & \leq  d(R_{k_1} , X'_{k_1, \bar{k}_2})   + \bar{\varepsilon} +  \varepsilon(k_1,k_2),
    \label{eq:single-letter-upbnd}
\end{align}
for all sufficiently large $n$. Furthermore, $\varepsilon \downarrow 0$ as 
$k_1,k_2$ become large.
% where $\delta>0$ can be made arbitrarily small given sufficiently large $n$ and  $\varepsilon$ can be made arbitrarily small given sufficiently large $k_1, k_2$.
\end{lemma}

\begin{proof}
    See Appendix \ref{appA-singleletter-distortion}
\end{proof}

 Next, note that the following inequality holds by definition for the single-letter discrete distortion:
\begin{align}
    d(R,X_{\bar{k}_2}) &= d(R,X_{\bar{k}_2} | A_{k_1 } = 0) \Pr(A_{k_1} = 0) +  d(R,X_{\bar{k}_2} | A_{k_1 } = 1) \Pr(A_{k_1} = 1) \nonumber\\
    &= d(R_{k_1},X'_{k_1,\bar{k}_2}) \Pr(A_{k_1} = 0) +  d(R,X_{\bar{k}_2} | A_{k_1 } = 1) \Pr(A_{k_1} = 1).
\end{align}
Then, by having the probability of clipping approach zero $\lim_{k_1 \to \infty} \Pr(A_{k_1} = 1) = 0$, the second term above goes to zero as a direct result of uniform integrability, and we have the following asymptotic limit:
\begin{align}
    \lim_{k_1 \to \infty} d(R_{k_1}, X'_{k_1,\bar{k}_2}) &= d(R, X_{\bar{k}_2}):=  \int_{\~X}  \Tr{ \sqrt{\rho } \, \Lambda_X( d z) \sqrt{\rho}\,\Delta\big(Q_{\bar{k}_2}(z)\big)}.
\end{align}
Then as $k_2, k_2' \to \infty$, we can bound from above this RHS distortion value by
\begin{align}
    \lim_{k_2,k'_2 \to \infty} d(R,X_{\bar{k}_2}) &\leq  D + \lim_{k_2,k'_2 \to \infty} \left(d(R,X_{\bar{k}_2}) - d(R,X)\right) \nonumber\\
    &= D + \lim_{k_2,k'_2 \to \infty} \int_{\~X}  \Tr{ \sqrt{\rho } \, \Lambda_X( d z) \sqrt{\rho}\,\Big(\Delta\big(Q_{\bar{k}_2}(z)\big) - \Delta(z)\Big)} = D,
\end{align}
where the last equality follows similarly from continuity and uniform integrability of the distortion observable operator $\Delta(x)$ as a function of $x\in \~X$. Combining the above bounds to the single-letter expression in \eqref{eq:single-letter-upbnd} shows
$\lim_{n \to \infty}  d_n(R^n, {\hat X}^n_{k_1,\bar{k}_2}) \leq D$, which completes the proof of achievability.

\section{Evaluation of the Qubit-Binary System} \label{sec:evaluations_qubit}
%%%%%%%%%%%%%%%%%%%%%%%%%%%%%%%%%%%%%%%%%%%%%%%%%%%%%%%%%%%%%%%%%%
%%%%%%%%%%%%%%%%%%%%%%      Section:  Evaluations  %%%%%%%%%%%%%%%
%%%%%%%%%%%%%%%%%%%%%%%%%%%%%%%%%%%%%%%%%%%%%%%%%%%%%%%%%%%%%%%%%%
In this section, we study the example of  qubit-binary systems. Having the qubit source state $\rho$, a Bernoulli output distribution $Q_X$, and entanglement fidelity as the distortion measure, we aim to find the OC rate-distortion function $R(D;\infty,\rho||Q_X)$ for the case of unlimited common randomness. By inverting this function we then achieve the RLOT cost $D(R; \infty, \rho || Q_X)$ function. We then provide the numerical results for a few numerical examples and plot the rate-distortion function.
\subsection{Qubit System with Unlimited Common Randomness}
For the case of the qubit QC system, we employ entanglement fidelity \cite{barnum_quantum_2000,datta_wilde_2012_quantum_ratedistortion_reverseshanon} as the distortion measure, which can be written as
\begin{align}
    % line 1
   \Tr{\Delta_{RX}\tau_{RX}} = 1 - \bra{\psi^{RA}} \left(    \sum_x \sqrt{\rho} M_x \sqrt{\rho} \otimes \ketbra{x}{x}    \right) \ket{\psi^{RA}}.\label{eq:eval_entanglement_fidelity1}
\end{align}
Using the spectral decomposition of $\rho = \sum_{t=1,2} P_T(t) \ketbra{\varphi_t}{\varphi_t}$, with the eigenbasis ${\ket{\varphi_t}}_{t=1}^2$, and by substituting the canonical purification of the above decomposition into \eqref{eq:eval_entanglement_fidelity1}, it simplifies to
\begin{align}
%line 1
    \Tr{\Delta_{RX}\tau_{RX}} &= 1 - \sum_{x} \bra{x}\rho M_x^{\~T_\varphi}
        \rho\ket{x} ,\label{eq:entanglement-fidelity-distortion-simplified}
\end{align}
where $M_x^{\~T_{\varphi}}$ is the transpose of $M_x$ with respect to the $\{\varphi_t\}$ basis, defined as
\begin{align} \label{eq:transpose_finite_case}
    M_x^{\~T_{\varphi}} = \sum_{t,s} \braket{\varphi_t}{M_x|\varphi_s} \ketbra{\varphi_t}{\varphi_s}.
\end{align}

\begin{remark} \label{remark:transpose_entropy}
Although the distortion formula in \eqref{eq:entanglement-fidelity-distortion-simplified} has transposed POVM operator $M_0^{\~T_\varphi}$, we can ignore this transpose in the distortion constraint of the above optimization problem. The reason is that the eigenvalues of $\rho_x$ are preserved under the transposition with respect to any basis. This implies that the entropy function also does not change under transposition.
    % $H(\hat \rho_x) = H(\sqrt{\rho}M_x^\~T \sqrt{\rho}/\Tr{\rho M_x})$.
% Therefore, we may remove the transpose from all the terms in the optimization problem.
\end{remark}

\ifmycomments
\mynote{Also, note that $\ket{\varphi_t}$ eigenbasis of decomposition is generally different from $\ket{x}$ basis of the output measurement. Because, if they were the same basis, the problem would be simplified to classical source coding.}
\fi

In the presence of an unlimited amount of common randomness, the only effective rate becomes $I(W;R)$, which is lower-bounded by $I(X;R)$ because of the data processing inequality and the Markov chain $R-W-X$ \cite{tamas_output_constrained_2015}. Thus $W = X$ minimizes the mutual information, which means no local randomness is required at the decoder.  Therefore, using the main theorem, for a qubit system with input state $\rho$ and  Bernoulli($q_1$) output distribution, the OC rate-distortion function is obtained by
\begin{align} \label{eq:qubit_ratedist_problem}
    R(D;\infty, \rho||\text{Bern}(q_1)) = &\min_{M_x^A} I(R;X)_\tau, = \min_{M_x^A} \left[H(\rho) - \Big(q_0 H(\hat \rho_0) + q_1 H(\hat \rho_1)\Big)\right]\\
    \text{subject to: }
    &\begin{array}{l}
        \Tr{ M_x \rho} =  q_x, \quad \quad x=0,1, \\
                    \sum_x \Tr_A\left\{(\text{id}_R\otimes M_x) \psi^{R A}\right\} = \rho,\\
                   1 - \sum_{x} \bra{x}\rho M_x^{\~T_\varphi}\rho\ket{x} \leq D.
                \end{array} 
\end{align}

% \begin{align} \label{eq:qubit_ratedist_problem}
%        R(D;\infty, \rho||\text{Bern}(q_1)) = &\min_{M_x^A}      I(X;R)_\tau,
%        \end{align}
%        such that 
%        \begin{align}
%                       \Tr{ M^A_x \rho} =  q_x, \quad \quad \forall x\in \~X, \ \
%                     \sum_x \Tr_A\left\{(\text{id}_R\otimes M_x^A) \psi^{R A}\right\} = \rho, \ \ 
%                    \Tr{\Delta_{RX}\tau_{RX}} \leq D.
%         \end{align}
where the information quantity is with respect to the composite state
    $\tau_{R X} = \sum_x \hat\rho_x \otimes \ketbra{x}{x}^X.$

\subsubsection{Rate-Limited Optimal Transport for Qubit Measurement System}
By addressing the above optimization problem, we obtain Theorem \ref{th:qubit_ratedist_theorem}, which yields a transcendental system of equations that determines the OC rate-distortion function for this system.
We first define the quantum source $\rho$ and some operator, $N(n,s)$ in the following matrix format:
\begin{align} \label{eq:qubit_evaluation_Nopt_matrix}
        \rho = \begin{bmatrix}
        \rho_1 & \rho_2\\
        \rho_2^* & 1 - \rho_1
    \end{bmatrix}, \quad
    N(n,s) := \begin{bmatrix}
            n & s \rho_2 /|\rho_2| \\
            s \rho_2^* /|\rho_2| & q_0 - n
        \end{bmatrix},
\end{align}
 for some $\rho_1,\rho_2,n$, and $s$,  where $n$ and $s$ depend on the measurement operator $M_0$ and need to be determined by solving the optimization problems. We also define the following parameters:
    \begin{align} \label{eq:E1}
        E_1(n,s) &:=\sqrt{\left(n-\frac{q_0}{2}\right)^2 +s^2}, \qquad \qquad \qquad \qquad  E_2(n,s) :=\sqrt{\left(n-\rho_1 + \frac{1-q_0}{2}\right)^2 +(s - |\rho_2|)^2},\\ 
    % \end{align}
    % $a,b,c$ are fixed parameters of system based on the input and output states $\rho, q_0$ defined as
    % \begin{align*}
        a &:= 1 - \frac{4 |\rho_2|^2}{1+2k},
        \qquad\qquad\qquad\qquad \qquad \qquad \qquad b := \frac{2|\rho_2|(2\rho_1 -1)}{1+2k},\\
        c &:= q_0 \left(\rho_1 - 1 + \frac{2|\rho_2|^2}{1+2k} \right) + \braket{1|\rho^2}{1} - 1 +D,\ \
         k := \sqrt{\det{\rho}}.
    \end{align}
For the qubit measurement system, let $D_{OT}:= D(R = \infty, R_c, \rho||\text{Bern}(q_1))$ denote the QC optimal transportation cost, which is characterized in Theorem \ref{th:qubit_optimal_transport} in the next subsection. 
With these definitions of the parameters in place, we can characterize the OC rate-distortion function as follows.

\begin{theorem} \label{th:qubit_ratedist_theorem}
    For the case of qubit input state $\rho$ 
    % \begin{align*}
    % $\rho = \begin{bmatrix}
    %     \rho_1 & \rho_2\\
    %     \rho_2^* & 1 - \rho_1
    % \end{bmatrix},
    % \end{align*}
    % and fixed output with Bernoulli$( 1-q_0)$ distribution, with the presence of an unlimited amount of common randomness, and using entanglement fidelity distortion measure, 
    the OC rate-distortion function and the corresponding optimal POVMs $M_0, M_1$ are provided as follows: 
    % For any distortion level $D$ above the threshold $D\geq D_{R_0}$ where 
    % \begin{align}
    %     D_{R_0} := 1 - q_0 \braket{0}{\rho^2 |0} -  (1 - q_0) \braket{1}{\rho^2 |1},
    % \end{align}
    % the output state can be generated independently, thus $R(D;\infty, \rho||\text{Bern}(q_1))=0$ with $M_{0,R_0} = q_0 I$.
    % Otherwise if $ D_{OT}\leq D < D_{R_0}$, (where $D_{OT}$ is the optimal transport cost) then, 
    \begin{align} \label{eq:rate-dist-function-theorem}
        R(D;\infty, \rho||\text{Bern}(q_1)) &=\begin{cases}
            H(\rho) - q_0 H\left(N_{opt}/q_0\right) - (1 - q_0) H\left((\rho - N_{opt})/q_1  \right), & D_{OT}\leq D \leq D_{R_0},\\
            0, &  D_{R_0} \leq D
        \end{cases} 
        \end{align}
        where    
        $D_{R_0} := 1 - q_0 \braket{0}{\rho^2 |0} -  (1 - q_0) \braket{1}{\rho^2 |1}$. In the first case, POVM elements are $M_0=q_0I$, and $M_1=q_1I$. In the second case, they are 
    $M_0 = \sqrt{\rho}^{-1} N_{opt} \sqrt{\rho}^{-1}, \quad M_1 = I - M_0$.  
    % The parameter $N_{opt} / q_0$ is the optimal PMR state conditioned on outcome 0,
    % expressed by
    % \begin{align*}
    %     N_{opt} := \begin{bmatrix}
    %                     n & s \rho_2 /|\rho_2| \\
    %                     s \rho_2^* /|\rho_2| & q_0 - n
    %                 \end{bmatrix}   ,
    % \end{align*}
    Also, $N_{opt}/q_0 = N(n,s)/q_0$ with the optimal $n,s$ values, is the optimal PMR state conditioned on outcome zero, and that optimal $n,s$ satisfy the transcendental system of equations:
    \begin{align}
    \begin{cases}
        \frac{-as + b(n - q_0/2)}{E_1}\ln{\frac{q_0/2 + E_1}{q_0/2 - E_1}}
        +  \frac{-a (s - |\rho_2|) + b(n - \rho_1 + \frac{1-q_0}{2})}{E_2} \ln{\frac{\frac{1 - q_0}{2} + E_2}{\frac{1 - q_0}{2} - E_2}}&=0.\\
        a n + b s + c &=0.
    \end{cases}
    \end{align}

\end{theorem}
\begin{proof}
See Appendix \ref{appB}.
\end{proof}

\begin{remark}
 The following two special source states result in interesting  $N_{opt}$ optimal matrices.  
     (I) Pure input state: 
        For the case of pure input state, the rate-distortion curve reduces to a single point where the rate is $R =0$, the optimal $N_{opt} = q_0 \rho$ and $D = D_{R_0}$. This is because the pure input state has no correlation with the reference, so the receiver can simply use local randomness in its decoder.  
 (II)    Diagonal (among canonical eigen-basis) quantum  input state: 
        In this case, the optimal operator $N_{opt}$ will also be diagonal. 
\end{remark}

        \ifmycomments
        {\color{blue}
        One can simply find that in this case when $D \leq D_{R_0}$, the optimal operator is given by
        \begin{align*}
            N^{cl}_{opt} = \begin{bmatrix}
                1 - D + (1-\rho_1)( q_0 + \rho_1 - 1) & 0\\ 0 & D + \rho_1(q_0 + \rho_1 - 2)
            \end{bmatrix}.
        \end{align*}
        }
        \fi

\subsubsection{Optimal Transport for Qubit Measurement System}
The optimal transport scheme provides the minimum achievable distortion when the rate of information is not limited. The following theorem provides this value for the problem of the qubit measurement system.
\begin{theorem} \label{th:qubit_optimal_transport}
% For the case of qubit input state $\rho$ and fixed output distribution $ \text{Bern}(q_1)$, the optimal transport cost $D_{OT}$ is given as follows. 
Defining the parameter $Q:=  \frac{\rho_1 - 1/2}{\sqrt{1 - 4|\rho_2|^2}}$, the optimal transportation cost $D_{OT}$ is given by:
% and the optimal operator $N_{OT}:=\begin{bmatrix}
%         n_{OT} & s_{OT}\rho_2/\abs{\rho_2}\\ s_{OT} \rho_2^*/\abs{\rho_2} & q_0 - n_{OT}
%     \end{bmatrix}$ given by:
    \begin{align}
    D_{OT} &= \begin{cases}
        q_0(1-\rho_1) + \det(\rho) + \frac{1-q_0}{2}\left(1 - \sqrt{1 - 4|\rho_2|^2}\right), &  \text{if }  \quad  Q \leq \frac{ \det(\rho)}{1 - q_0} - \frac{1}{2}\\
        (1-q_0)\rho_1 + \det(\rho) + \frac{q_0}{2}\left( 1 - \sqrt{1 - 4|\rho_2|^2}\right), & \textbf{ if }\quad Q \geq \frac{1}{2} -  \frac{\det(\rho)}{q_0}\\
        1 - q_0 \left(\rho_1 - 1 + \frac{2|\rho_2|^2}{1+2k}\right) - \braket{1|\rho^2}{1} - a \alpha - b \beta, & \text{ if } \quad \frac{ \det(\rho)}{1 - q_0} - \frac{1}{2} \leq Q \leq \frac{1}{2} -  \frac{\det(\rho)}{q_0}
    \end{cases}
\end{align}
where $\alpha$ and $\beta$ are the parameters given below and $\Delta' :=\left(q_0(1 - q_0) - \det(\rho)\right)\det(\rho)$,
\begin{align}
    \alpha &= \frac{\left(q_0 - 2\det(\rho)\right) |\rho_2| + \text{sgn}\{a-b\} (1 - 2\rho_1) \sqrt{\Delta'}}{4|\rho_2|^2 + (1 - 2\rho_1)^2},\nonumber\\
    \beta &= \frac{2q_0 |\rho_2|^2 - (1 -2\rho_1)(\rho_1 q_0 - \det(\rho)) + \text{sgn}\{a-b\}2|\rho_2|\sqrt{\Delta'}}{4|\rho_2|^2 + (1 - 2\rho_1)^2}.\nonumber
\end{align}
\ifmycomments
{\color{blue}
\begin{enumerate}
    \item if  $\quad  Q \leq \frac{ \det(\rho)}{1 - q_0} - \frac{1}{2}$ then

        \begin{align}
            D_{OT} &= q_0(1-\rho_1) + \det(\rho) + \frac{1-q_0}{2}\left(1 - \sqrt{1 - 4|\rho_2|^2}\right),\\
            s_{OT} &= \frac{b}{\sqrt{1 - 4|\rho_2|^2}}\frac{1-q_0}{2} + |\rho_2|, \quad
            n_{OT} =  \left(\frac{a}{\sqrt{1 - 4|\rho_2|^2}}-1\right)\frac{1-q_0}{2}+ \rho_1. \nonumber
        \end{align}
        
    \item Else if $\quad Q \geq \frac{1}{2} -  \frac{\det(\rho)}{q_0}$ then
        \begin{align}
        D_{OT} &= (1-q_0)\rho_1 + \det(\rho) + \frac{q_0}{2}\left( 1 - \sqrt{1 - 4|\rho_2|^2}\right),\\
            s_{OT} &=\frac{b}{\sqrt{1 - 4|\rho_2|^2}} \frac{q_0}{2},
            \quad n_{OT} = \left(\frac{a}{\sqrt{1 - 4|\rho_2|^2}} + 1\right) \frac{q_0}{2}.\nonumber
        \end{align}
    
    \item Else if $\quad \frac{ \det(\rho)}{1 - q_0} - \frac{1}{2} \leq Q \leq \frac{1}{2} -  \frac{\det(\rho)}{q_0}$ then,      
            \begin{align}
            D_{OT} &= 1 - q_0 \left(\rho_1 - 1 + \frac{2|\rho_2|^2}{1+2k}\right) - \braket{1|\rho^2}{1} - a n_{OT} - b s_{OT},\\
                s_{OT} &= \alpha,\nonumber\\
                n_{OT} &= \beta.\nonumber
            \end{align}
        where $\Delta' :=\left(q_0(1 - q_0) - \det(\rho)\right)\det(\rho)$.
\end{enumerate}
}
\fi
\end{theorem}
\begin{proof}
See Appendix \ref{appC}.
\end{proof}
Interestingly, when the input state is prepared with the diagonal density operator along the eigenbasis of the output state (i.e. $\rho_2 =0$), the optimal matrices $N_{opt}$ and $\rho- N_{opt}$ encompasses the classical binary optimal transport scheme. One can see that the first and second conditions reduce to $q_0\geq \rho_1$ and $q_0 < \rho_1$ and the third condition is empty. However, the distortion value will be obviously different.

\ifmycomments
{\color{blue}
\subsubsection{Minimum Required Rate for Optimal Transport}
Finally, it is important to note that an unlimited communication rate is not required for optimal transport. The minimum required rate $R_{\text{min, OT}}$ for the optimal transport scheme can be obtained by substituting optimal values $s_{OT}$ and $n_{OT}$ in \eqref{eq:rate-dist-function-theorem}, which results
\begin{align}
    R_{\text{min, OT}} = H(\rho) - q_0 H_b\left(\frac{1}{2} - \frac{E_1(n_{OT},s_{OT})}{q_0}\right) - (1- q_0)H_b\left(\frac{1}{2} - \frac{E_2(n_{OT},s_{OT})}{1 - q_0}\right).
\end{align}
where $H_b(.)$ is the binary entropy function, and $E_1(n_{OT},s_{OT})$ and $E_2(n_{OT},s_{OT})$ are functions defined in \eqref{eq:E1}.
}
\fi

\subsection{Numerical Results}
We used the CVX package \cite{cvx,gb08} to find numerical solutions for the examples of this convex optimization problem. Also  \cite{cvxquad} provides the CVX functions for the von-Neumann entropy functions. The OC rate-distortion function is numerically evaluated for the following set of examples (a), (b), (c) and (d) with fixed $q_0=1/2$ and $\rho_1 = 1/2$ parameters and different off-diagonal values: 
\begin{align*}
      &\text{Ex1: } \rho_a = \begin{bmatrix}1/2 &0\\ 0 & 1/2\end{bmatrix}, 
       &&\text{Ex2: }\rho_b = \begin{bmatrix}1/2 &0.132 - 0.036i\\ 0.132 + 0.036i & 1/2\end{bmatrix},\\
      &\text{Ex3: }\rho_c = \begin{bmatrix}1/2 &0.075 - 0.231i\\ 0.075 + 0.231i & 1/2\end{bmatrix}, 
     &&\text{Ex4: }\rho_d = \begin{bmatrix}1/2 &-0.140 - 0.387i\\ -0.140 + 0.387i & 1/2\end{bmatrix}, 
\end{align*}
with the purity values $\Tr{\rho_a^2} =  0.500$, $\Tr{\rho_b^2} =  0.537$, $\Tr{\rho_c^2} =  0.618$ and $\Tr{\rho_d^2} =  0.839$.
These rate-distortion functions are plotted in Figure \ref{fig:figure_ratedist1}, which shows that starting from a maximally mixed state (Ex.1.), as the source state becomes purer, it requires less communication rate to maintain the same level of entanglement fidelity. 

In the case of a pure source state, the rate-distortion function reduces to a single point at the no transmission rate. This is intuitively acceptable as the pure state is independent of the reference state. So the receiver can generate random outcomes independent of the source. However, the entanglement fidelity distortion will not be zero because the measurement collapses the state into deterministic outcomes and hence it will not fully recover the source state. On the contrary, the maximally mixed state has the maximum dependence on the reference state which requires the maximum rate of transmission to recover the state with the same level of distortion. 
\begin{figure}[ht]
\centering
\includegraphics[width=0.7\textwidth]{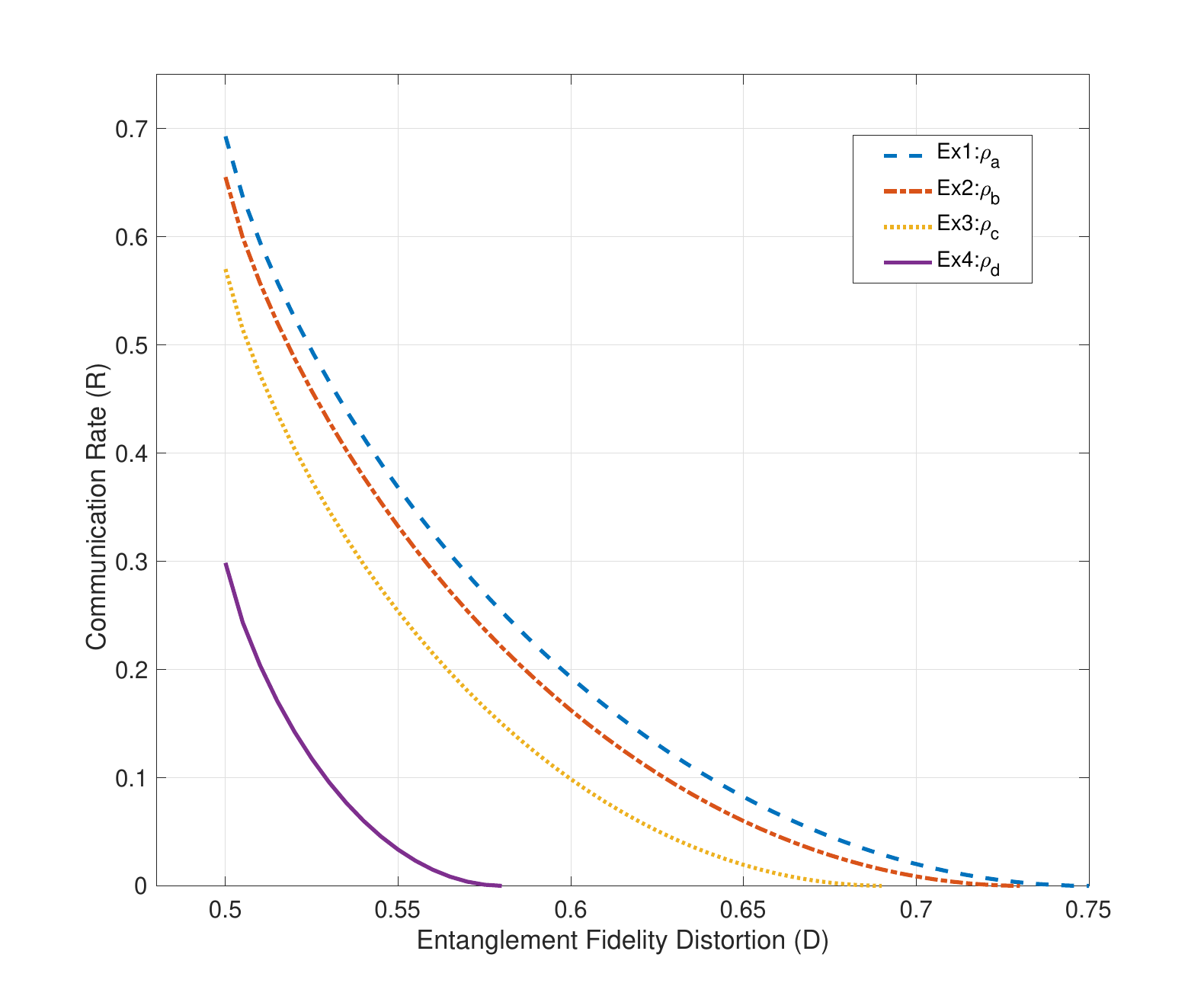}
\caption[Plot of OC Rate-Distortion Function for a Qubit System]{The OC rate-distortion function with unlimited common randomness for the examples}
\label{fig:figure_ratedist1}
\end{figure}

\ifmycomments
\mynote{Interesting facts: For the case with diagonal input state, all of the optimal operators $N$ are also diagonal as well. I must investigate if in this case, this rate distortion is the same as the classical case or not.}

\mynote{Another fact is that for all pure states, the rate-distortion is a single point and $R=0$. Is that right?}

\myqu{Just want to double check. If $N = \sqrt{\rho} M_0 \sqrt{\rho}$ is psd, then $M_0$ is also psd? }\myans{Yes, but this is only true when $\rho $ is not singular matrix. The proof simply follows from the definition of semidefiniteness:
\begin{align}
    X \succeq 0 \; \Longleftrightarrow \; z^H X z \geq 0 \quad \forall z \in \mathbb{C}^n
\end{align}}

\myqu{Why even for the classically quantum case with same input and output classical distributions, the optimal transport distance is not zero? Is there a bug with the distortion constraint?}\myans{This happens as a result of entanglement fidelity distortion. Although the input is considered classical, still it is a classically quantum state. As performing a measurement eliminates the connection with the reference, there are always some levels of distortion. The only case in which there is no distortion is when the input is pure state and the output is also deterministic. In general, if we would like to achieve zero distortion, a quantum channel should be used instead of the current classical channel.}
\fi

% two sections for gaussian, the intro and the eval

\section{Evaluation of the Gaussian Quantum States} \label{sec:evaluation_gaussian}

\subsection{A Brief Overview of Gaussian Quantum Systems}
Before providing the evaluation of the Gaussian systems, we first introduce the principal definitions and provide a brief overview of the Gaussian quantum systems.
\subsubsection{Gaussian Quantum Systems}\label{subsec:overview_gaussian_systems}
Let $\~H_A \equiv \~L^2(\mathbb{R}^s)$ be the separable infinite-dimensional Hilbert space of the square-integrable functions $\psi(\eta)$ for $\eta \in \mathbb{R}^s$ corresponding to $s$ Harmonic oscillators. 
The canonical observable operators of this system are $Q_1$, $P_1$, $Q_2$, $P_2, ..., Q_s$, $P_s$ with continuous eigenspectra, where $Q_i, P_i$ are the position and momentum quadrature operators of the $i$-th harmonic oscillator. In this section, we consider the energy-bounded harmonic oscillator system as described by \cite{depalma_quantum_optimal_transport_2021,Holevo_2019_quantum_systems_book}. Then referring to Lemma 11.5 of Holevo's \cite{Holevo_2019_quantum_systems_book}, we note that a subset of energy-bounded states $\~G_E \subset \~G(\~H)$ with a corresponding energy operator $F$ defined as
$\~G_E : = \{S: \Tr{SF} \leq E\}$,
is compact.
% has the property that for any $\varepsilon > 0 $, there exists a finite-rank projector $P_\varepsilon$ such that $\Tr{P_\varepsilon S} \geq 1 - \varepsilon$ for all these energy-bounded density operators $S\in \~G_E$. 
% {\color{blue}
% Thus, having the wave-function $\psi \in \~L^2(\mathbb{R}^s)$, these observable operators act on this wave-function as follows \cite{depalma_quantum_optimal_transport_2021}:
%  \begin{align}
%      (Q_i \psi)(q) = q_i \psi(q), \quad (P_i \psi)(q) = -i \frac{\partial }{\partial q_i} \psi(q).
%  \end{align}
%  }
 The canonical observables satisfy the Canonical Commutation Relation (CCR) resulting from the Heisenberg uncertainty principle. For more convenience, we combine the quadrature operators as elements of a $2s\times 2s$ vector operator,
     $R_1 = Q_1,\; R_2 = P_1 ,\;  \cdots, \; R_{2s-1} = Q_{m}, \; R_{2s} = P_{s},$
which redefines the CCR as $[R_i,R_j] = i \Delta_{ij} I_\~H,$ for $i,j = 1 , ...,2m,$
with $\Delta$ being a non-degenerate skew-symmetric symplectic matrix defined as
     % \begin{align} \label{eq:symplectic_matrix}
$         \Delta = \bigoplus_{k=1}^s \begin{bmatrix}
             0 & 1 \\ -1 & 0
         \end{bmatrix}.$
     % \end{align}
     
The Weyl operator is further defined by $W(z) = \exp{iRz} $ where $z:= [q_1, p_1, ..., q_s, p_s], z\in \mathbb{R}^{2s}$ is a vector of values in some phase space corresponding to the eigenvalues of the quadrature operators.
The displacement operator is also defined by $D(m) := W(\Delta^{-1} m)$.
The \textit{Wigner characteristic function} is further given by $\phi_\rho (z) = \Tr{\rho W(z)}$. The domain of the Wigner function $\~Z = \mathbb{R}^{2s}$ together with the symplectic matrix $\Delta$ form a \textit{symplectic space} $\~K := (\~Z, \Delta)$ which is called the phase space.

\ifmycomments
{\color{blue}
We define the Weyl operator $W(z) = \exp{iRz} $ where $z:= [q_1, p_1, ..., q_s, p_s], z\in \mathbb{R}^{2s}$ is a vector of eigenvalues of the corresponding quadrature operators. The CCR can be characterized using the Weyl operators in the following Weyl-Segal format:
\begin{align}
    W(z) W(z') = \exp{-\frac{i}{2} \Delta(z,z')}W(z + z')
\end{align}
where 
\begin{align} \label{eq:bilinear_form}
    \Delta(z,z') = z^T \Delta z
\end{align}
is the canonical symplectic form and $\Delta$ is the symplectic matrix. Then for a density operator $\rho$ of an arbitrary quantum state, the Wigner characteristic function is defined as
\begin{align}
    \phi_\rho (z) = \Tr{\rho W(z)}.
\end{align}
By taking a Fourier transform of the characteristic function with respect to a variable $x \in \~X \equiv \~Z = \mathbb{R}^{2s}$, we obtain the Wigner quasi-probability distribution
\begin{align} \label{eq:wigner_quasi_prob}
    \Phi_\rho(x) = \int_Z  \phi_\rho (z) \exp{-ix^T z} \frac{d^{2s} z}{(2\pi)^{2s}}.
\end{align}
which is a real function normalized to 1, but not necessarily non-negative. 
For a detailed definition of the Wiegner function see \cite[Section II.A]{weedbrook_gaussian_2012}  and \cite[Chapter 4]{serafini_quantum_2017}.

 Hence, the Wigner characteristic function and Wigner distribution are the unique phase space representations of the quantum state $\rho$ according to Stone-von Neumann's uniqueness theorem.
A transformation matrix $T$ is called a symplectic transformation if it maps the symplectic space into itself; i.e. preserving the symplectic form
$\Delta(Tz,Tz') = \Delta(z,z'),$ for all $ z, z' \in \~Z.  $
}
\fi
% \subsubsection{Statistical Moments}
For a density operator $\rho$, the mean vector and covariance matrix are given by 
% \begin{align}
    $m = \Tr{\rho R}$, and  $\alpha - \frac{i}{2} \Delta = \Tr{(R - m)^T\rho \ (R- m)}$,
% \end{align} 
respectively, where $\alpha$ is the covariance matrix of the Wigner quasi-probability distribution. 
% {\color{blue}The density operator $\rho$ can be expressed as a displacement of its zero-mean version $\rho_0$ using a displacement operator $D(m) := W(\Delta^{-1} m)$ by $\rho = D(m) \rho_0 D(m)^\dagger$.}
Further, the following inequality holds (in the positive semi-definite sense) for the covariance matrix of the quantum state which is a result of the uncertainty principle:
\begin{align} \label{eq:uncertainty_relation}
    \alpha \geq \pm \frac{i}{2} \Delta.
\end{align}

The eigenvalues of $\alpha$ depended on the, and have no significance. However, for the operator $\tilde \alpha = \abs{\Delta^{-1} \alpha}$, the eigenvalues $\pm i\gamma_j$, for $j=1,\cdots,s$ are called the \textit{symplectic eigenvalues}. The uncertainty relation \eqref{eq:uncertainty_relation} translates to $\alpha_j>1/2$
or described as $\left( \Delta^{-1} \alpha \right)^2 \geq - \frac{1}{4} I$,
with equality when the system is in a pure state \cite{holevo2001evaluating-capacities-bosonic}.
In the special case when the state is a product state of the modes, which is when $\alpha$ is a block-diagonal  of $s$ separate 2-by-2 matrices, the symplectic eigenvalues are simply the determinant of each mode 
    $\gamma_j = \det{\alpha_j}$ \cite{holevo_hirota1999capacityofquantumgaussianchannels}.
The quantum Gaussian state is defined on $\~H$ as a state whose Wigner characteristic function and the Wigner quasi-probability function are Gaussian and given by:
\begin{align*}
    \phi_\rho(z) &= \exp[-\frac{1}{2} z^T \alpha z + i m^T z ],\quad \Phi_\rho(x) = \frac{\exp[ -\frac{1}{2} (z- m)^T \alpha^{-1}(z- m)]}{(2\pi)^s \sqrt{\det \alpha}} .
\end{align*}
Similar to classical Gaussian distributions, for a Gaussian quantum state, the displacement and covariance matrix are sufficient to fully represent the Gaussian state. For a more detailed review of the Gaussian quantum systems see \cite{Holevo_2019_quantum_systems_book,weedbrook_gaussian_2012,hafez2023thesis,serafini_quantum_2017}.

\subsubsection{Gaussian Observables}

The measurement POVM $M$ has a general form as defined in \cite{Holevo_2019_quantum_systems_book}. For any observable  POVM $M$ with the PMR ensemble $\{\rho_z, \pi_Z(dz)\}_{z\in \~Z}$, the information gain function is defined as 
\begin{align}
    I_g(\rho_N;M) = H(\rho_N) - \int_\~Z H(\rho_z) \pi_Z(dz).
\end{align}
An important group of observables is the general form of the covariant Gaussian observable as provided in \cite{holevo2020-gaussian-maximizers-observables}:
\begin{align} \label{eq:general-gaussian-observ}
    \tilde{M}(d^{2s} z) = D(Kz) \rho_G D(Kz)^\dagger \frac{\abs{\det K}^2 d^{2s}z}{\pi^s}, 
\end{align}
where $\rho_G$ is a density operator that is the parameter of the Gaussian observable and is different from the input density operator $\rho_N$. The von-Neumann entropy of the general Gaussian state with covariance matrix $\alpha$ is given by \cite{holevo2001evaluating-capacities-bosonic},
\begin{align} \label{eq:gaussian-quantum-entropy}
    H(\rho_N) = H_G(\alpha) := \frac{1}{2} \text{Sp}  g\left(\abs{\Delta^{-1} \alpha} - \frac{I}{2}\right).
\end{align}
where $\sp{.}$ is the matrix trace as opposed to the $\Tr(.)$ the trace in Hilbert space and that $g(x) = (x+1) \log(x+1) - x\log x$, for $x>0$, and $g(0 )  = 0$
is the Gordon function.

\subsubsection{Distortion Observable}
To introduce the distortion observable, we first provide the following definitions. Let $\~H^*$ be the space of continuous linear functionals on Hilbert space $\~H$ \cite{depalma_quantum_optimal_transport_2021}. We consider the following definition of the transposition.
\begin{definition}
    For any operator $X\in \~L(\~H)$, let $X^T$ be the linear operator on $\~H^* $ given by 
    \begin{align}
        X^T \bra{\phi} = \bra{\phi} X \quad \text{for all } \bra{\phi} \in \~H^*.
    \end{align}
\end{definition}
For finite-dimensional Hilbert spaces, this reduces to the definition of transpose specified in  \eqref{eq:transpose_finite_case}.
% Let $\{  \ket{i}  \}_{i=1}^d$ be a basis in the Hilbert space $\~H$. Having an operator $A\in \~H$, the operator $A^\~T\in \~H^*$ is defined as the transposition of $A$ with
%     \begin{align}
%         A^\~T := \sum_{i}^d\sum_j^d \braket{j|A}{i}  \ketbra{i}{j}_{H^*}.
%     \end{align}

\begin{remark}
   Assume having a source state $\rho \in \~H_A$ and a measurement POVM with outcomes $z \in \mathcal{Z}$. Then the conditional PMR states 
   $(\hat \rho_z^\~T)^R$ and the unrevealed PMR state $(\rho^\~T)^R$ live in 
   $\~H^*$
\cite{wilde2012information_theoretic_costs, depalma_quantum_optimal_transport_2021}. 
   % \begin{align}
   %     &(\hat \rho_u^\~T)^R := \frac{(\sqrt{\rho} M_u \sqrt{\rho})^\~T}{\Tr{M_u \rho}}, \quad \forall u\in \~U,
   %    \qquad  (\rho^\~T)^R =  \sum_u \Tr{M_u \rho} \cdot (\hat \rho_u^\~T)^R.
   % \end{align}
    Furthermore, as was mentioned in remark \ref{remark:transpose_entropy}, the entropic quantities are invariant under the transposition.
\end{remark}

The following distortion observable operator was introduced by \cite{depalma_quantum_optimal_transport_2021} using the quadrature operators,
\begin{align} \label{eq:distortion_operator}
    C = \frac{1}{2s} \sum_{i=1}^{2s} (R_i^\~T \otimes I_\~H - I_{\~H^*} \otimes R_i)^2.
\end{align}
This operator was defined to address the problem of quantum optimal transport for Gaussian quantum states.
The significance of the operator is that it acts on the composite state of the reference and the destination system in the $\~H^* \otimes \~H$ space.
As a result, the optimal coupling that minimizes this distortion measure has a unique correspondence to the physical quantum channels. 
Note that, in a fully quantum setting, applying this observable to a purified state $\Psi^{RA}$ yields zero distortion, i.e., $\Tr{\Psi^{RA}C}$=0, because the observations on the local state and on the reference state have equal values with probability $1$.
The transpose operation is essential because the reference state resides in the $\~H^*$ Hilbert space, therefore, if we want to compare the outcomes of the states, we must use the $R_i^\~T$ operator for the reference so that $\Tr{R^\~T \rho^\~T}  = \Tr{R \rho}$. 
For a more detailed explanation refer to \cite{hafez2023thesis, depalma_quantum_optimal_transport_2021}.

\ifmycomments
{\color{blue}
\begin{lemma}
    Let $\rho \in \~H_A$ be the source state and $\~M \equiv \{M\}_u\in \~U$ be a measurement POVM where $\~U$ is the space of the outcomes. Then, the conditional post-measured reference states and the unrevealed post-measured reference state are all in the transpose Hilbert space $\~H^*$.
\end{lemma}
\begin{proof}
    Let $\Lambda = M_u^\dagger M_u$, and $\lambda_u = P(U = u) = \Tr{M_u \rho}$. Then the state of the system after measurement given that $\{X = u\}$ is 
    \begin{align}
        \ket{\psi_u^{RA'}} := \frac{(I \otimes M_u) \ket{\phi_{RA}}}{\sqrt{\lambda_u}},
    \end{align}
    and its density operator is 
    \begin{align}
        \psi_u^{RA'}  = \frac{(I \otimes M_u) \ketbra{\phi_{RA}}{\phi_{RA}} (I \otimes M_u^\dagger)}{\lambda_u}.
    \end{align}
    Then the state of the reference given that $\{U = u\}$ is
    \begin{align}
        \Trxx{A'}{\psi_u^{RA'}} &= \frac{1}{\lambda_u} \sum_{i,j} \braket{j|\sqrt\rho M_u^\dagger M_u \sqrt{\rho}}{i}_{A'} \ketbra{i}{j}_R\\
        &= \frac{\left(  \sqrt{\rho} \Lambda_u \sqrt{\rho}  \right)^\~T}{\lambda_u} :=  \left(\hat \rho_u^\~T \right)^R
    \end{align}
    If the outcomes are not revealed, then the reference state is 
    \begin{align}
        \sum_u \lambda_u \frac{\left(  \sqrt{\rho} \Lambda_u \sqrt{\rho}  \right)^\~T}{\lambda_u} = {(\rho^\~T)}^R.
    \end{align}
    The proof simply extends to the continuous space POVMs as well.
\end{proof}
}
\fi
 We use this distortion observable that acts on $\~H^* \otimes \~H$ and modify it in a way that fits our QC Gaussian measurement system. 
% Toward this consider a finite-dimensional Hilbert space and a measurement POVM $\Lambda$ with the outcomes in a finite set $\~X$. 
% The  post-measured composite state of this system is
% \begin{align}
%     \Psi^{RX} &= \sum_{x\in \~X} \sqrt{\rho_A}^\~T \Lambda_x^\~T \sqrt{\rho_A}^\~T \otimes \ketbra{x}{x} 
%     = \sum_{x\in \~X} \hat \rho_x^\~T \otimes \pi(x) \ketbra{x}{x} ,
% \end{align}
Consider a POVM $\Lambda$ with outcomes in $\mathcal{Z}$.
Let $\pi_Z(B) = \Tr{\rho^A \Lambda(B)}$ for all $B\in \~B(\~Z)$. 
We define the distortion of the PMR ensemble $\{\{\hat \rho_z\}_{z\in\~Z},\pi_Z\}$, in a similar fashion to the distortion observable mapping $x\mapsto\Delta_R(x)$, described in Subsection \ref{subsec:distortionmeasure} as 
\begin{align} 
    % line 1
    d(\rho_A, \Lambda) %&= \Tr{\left( \sum_x {(\hat \rho_x^\~T)}^R \otimes \pi(x ) \ketbra{x}{x} \right)\left[ \frac{1}{2s} \sum_{i=1}^{2s}\left(\Delta_i^\~T \otimes I_\~H  - I_{\~H^*} \otimes \Delta_i\right)^2\right]} \\
    % % line 3
    % &= \frac{1}{2s} \sum_{i=1}^{2s} \sum_z  \Tr \Bigg\{\left(  {(\rho_z^\~T)}^R \otimes \pi(z ) \ketbra{z}{z} \right)\Bigg[\left(R_i^\~T \otimes I_\~H - \bar m_i(z) + \bar m_i(z) - I_{\~H^*} \otimes R_i\right)^2 \nonumber\\
    % & \qquad \qquad \qquad \qquad \qquad \qquad \qquad \qquad  + 2 \bigg(R_i^T \otimes I_\~H - \bar m_i(z)\bigg)\bigg( \bar m_i(z) - I_{\~H^*} \otimes R_i \bigg)\Bigg]\Bigg\} \\
    % line 4
    &:= \frac{1}{2s} \sum_{i=1}^{2s} \int_\~Z  \Bigl[ \Tr{\hat\rho_z^R (R_i - \bar m_i(z) )^2} + \braket{z| (R_i - \bar m_i(z) )^2}{z} \Bigr] \pi_Z(dz) \label{eq:distortion_continuous_firstdefinition}\\
    % % line 5
    % &= \frac{1}{2s} \sum_{i=1}^{2s} \sum_{z\in \~Z} \pi(z) \left[ \Tr{\rho_z^R (R_i - \bar m_i(z) )^2} + \braket{z_1 \cdots z_{2s}\Bigg| \left(\sum_{z_i} z_i \ketbra{z_i} - \bar m_i(z) \right)^2}{z_1 \cdots z_{2s}} \right]\\
    % line 6
    &= \frac{1}{2s} \sum_{i=1}^{2s} \int_\~Z  \Bigl[ \left\{\Sigma(z)\right\}_{ii} + (z_i - \bar m_i(x) )^2 \Bigr] \pi_Z(dz)\\
    % line 7
    &= \frac{1}{2s} \int_{\~Z}  \left(\text{Sp}\{\Sigma(z)\}  +  {\norm{z - \bar m(z) }_2^2} \right) \pi_Z(dz) := d(\rho_Z, \pi_Z),
\end{align} 
where $\bar m_i(z) = \Tr{\hat \rho^R_z R_i } = \Tr{{(\hat \rho_z^\~T)}^R R^\~T_i }$ is the first moment and $\Sigma(z)$ is the covariance matrix of the state $\hat{\rho}^R_z$. 
% The distortion further extends to the continuous systems when $\rho$ is a CV quantum system, $Z$ is a continuous output distribution and $R_i$ are the continuous quadrature operators.
% \begin{align} 
%     d(\rho_A, \Lambda) = \frac{1}{2s} \int_{z\in \~Z}  \left(\text{Sp}\{\Sigma(z)\}  +  {\norm{z - \bar m(z) }^2}  \right)\pi_Z(dz) 
% \end{align}
Further, note that a Gaussian measurement system with the above distortion measure, satisfies the uniform integrability property defined in \ref{def:uniform-integ}. For the proof see Appendix \ref{appA-uniforminteg}.

\subsection{Problem Formulation}
In this subsection, we formulate the QC optimal transport for the case of Gaussian quantum systems with Gaussian source states and Gaussian output distribution. 
Recall the single-letter conditions of the rate pair in Theorem \ref{th:continuous-theorem}.
We further restrict our evaluation to the systems with unlimited common randomness $(R_c = \infty)$. In that case, 
we can restrict ourselves to $W=Z$ as in Section \ref{sec:evaluations_qubit}.
 % This reduces the problem into a constrained optimization over all possible measurement POVMs to minimize  $I_g(R;Z)$ subject to the marginal and distortion constraints.
Moreover, due to the ensemble-observable duality \cite{holevo2020-gaussian-maximizers-observables}, instead of searching for the optimal measurement itself, we can search for the optimum measurement outcome ensembles $\{\{\hat \rho_z\}_{z \in \~Z},\pi_Z\}$.
% When we choose measurement POVM $M(dw)$ as the variable of optimization, the input marginal constraint becomes a trivial result of $\int_\~W M(dw) = I$. Because multiplying both sides by $\sqrt{\rho}$ from left and right will give the input marginal constraint $\int_\~W \sqrt{\rho} M(dw) \sqrt{\rho} = \rho$. 
In that case, the \textit{OC rate-distortion function} is defined by the following optimization problem:
\begin{align}
    R(D; R_c = \infty, \rho || \pi_Z) := &\min_{\{\hat \rho_z:\ z \in \~Z\}} \left[  H(\rho) - \int_\~Z H(\hat \rho_z) \pi_Z(dz) \right]\\
    \text{subject to: } &\int_\~Z \hat \rho_z \pi_Z(dz) = \rho,\\
    &\frac{1}{2s} \int_{\~Z}  \left(\mathrm{Sp}\{\Sigma(z)\}  +  {\norm{z - \bar m(z) }_2^2} \right) \pi_Z(dz) \leq D.
\end{align}
% The above optimization problem is equivalent to the \textit{Rate-Limited Wasserstein distance} in definition \ref{def:ratelim_wasserstein}.

\begin{definition}\label{def:ratelim_wasserstein}
The \textit{rate-limited Wasserstein distance} $ W_2 (R, \rho || \pi_Z)$ of the system is defined as the square root of the inverse of the above OC rate-distortion function within the proper range.
% Furthermore, the 2-Wasserstein distance for the QC system is defined when no rate-limit is implied on the system.
 The QC Wasserstein distance of order 2 is defined by $W_2 (\rho || \pi_Z) := W_2 (\infty, \rho || \pi_Z)$.
 Then the rate of QC Wasserstein distance is defined as the lowest rate achieving the Wasserstein distance:
\begin{align}
    R_{W_2}(\rho ||\pi_Z):= \inf\Big\{R: W_2 (R, \rho || \pi_Z) = W_2 ( \rho || \pi_Z)\Big\}.
\end{align}
\end{definition}

% \textbf{(QC Rate-Limited 2-Wasserstein Distance)}
%     Having a source state $\rho \in \~L^2(\mathbb{R}^s)$ and a classical destination distribution $\pi_Z$, the quantum-to-classical rate-limited Wasserstein distance of order 2 is
%     \begin{align}
%         W_2^2 (R, \rho || \pi_Z) := &\min_{\hat \rho_z,\ z \in \~Z} \int_{z\in  \mathbb R} \Tr_R \left[\hat \rho_z \, \Delta_R(z)\right] \pi_Z(dz)\\
%         \text{subject to: } &I_g(R;X) = H(\rho) - \int_\~Z H(\hat \rho_z) \pi_Z(dz) \leq R,   \\
%         & \frac{1}{2s} \int_{\~Z}  \left(\mathrm{Sp}\{\Sigma(z)\}  +  {\norm{z - \bar m(z) }^2} \right) \pi_Z(dz) = \rho.
%     \end{align}

% Furthermore, the 2-Wasserstein distance for the QC system is defined when no rate-limit is implied on the system.
% \begin{definition}
%     The QC Wasserstein distance of order 2 is defined by $W_2^2 (\rho || \pi_Z) := W_2^2 (\infty, \rho || \pi_Z)$.
% \end{definition}

\subsection{Main Result: Optimality of Gaussian Quantum Measurement}
We provide a Gaussian measurement optimality theorem which shows that for the aforementioned OC rate-distortion quantum measurement system, with Gaussian source state and Gaussian output distribution and unlimited common randomness, the Gaussian observables achieve the optimal rate.  We first provide the following lemma which is the QC version of the \textit{law of total variance}. Consider a measurement that is applied on a source state $\rho$ and results in the PMR ensemble $\{\{\hat \rho_z^\~T\}_{z\in\~Z}, \pi_Z\}$. In this setting,  $\bar m(Z) = \Tr{\hat \rho_Z R}$ can be interpreted as the classical estimator of the source state $\rho$.
\footnote{This is analogous to the estimator $\tilde X = \Ex{X|Y}$ for an observation $Y$ of a source $X$ in a fully classical system.} 
The following lemma provides the relation for the covariance matrix of this estimator.

\begin{lemma}\label{lemma:lawoftotalvariance} \textbf{(QC Law of Total Variance)} For a source quantum state $\rho$  with first moment $\bar m_\rho$ and covariance $\Sigma_\rho$ and a measurement POVM with the corresponding PMR ensemble $\{\{\hat \rho_z^\~T\}_{z\in\~Z},\pi_Z\}$, we have,
\begin{align}
    \Sigma_\rho = \int \hat \Sigma(z) \pi_Z(dz) + \tilde \Sigma,
\end{align}
where,
\begin{align}
    \Sigma_\rho - \frac{i}{2} \Delta =& \Tr{\rho (R- \bar m_\rho)(R - \bar m_\rho)^T}\\
    \hat \Sigma(z) - \frac{i}{2} \Delta =& \Tr{\hat \rho_z(R - \bar m(z))(R - \bar m(z))^T }\quad  \text{for all } z\in \~Z\\
    \tilde \Sigma :=& \text{Cov}(\bar m(Z)) = \Exx{Z}{\bar m(Z) \bar m(Z)^T} - \bar m_\rho \bar m_\rho^T 
\end{align}
are the covariance matrix of the source, the conditional covariance of PMR state $\hat \rho_z$ and the covariance of the classical estimator $\bar m(Z)$ respectively.
\end{lemma}
\begin{proof}
Starting with the second moment of the source state 
    \begin{align}
    \Tr{\rho R R^T} &= \Tr{\left(  \int \hat \rho_z \pi_Z(dz)  \right) R R^T} = \int \left(\hat \Sigma(z) - \frac{i}{2} \Delta + \bar m(z) \bar m(z)^T \right) \pi_Z(dz),
\end{align}
then the covariance of the source state is 
\begin{align}
    \Sigma_\rho &= \Tr{\rho R R^T} - \bar m_\rho \bar m_\rho^T + \frac{i}{2} \Delta = \int \left(\hat \Sigma(z)  \right) \pi_Z(dz) + \Exx{Z}{ \bar m(z) \bar m(z)^T}- \bar m_\rho \bar m_\rho^T.
\end{align}
\end{proof}

Using the above lemma we provide the following Gaussian optimality theorem.
\begin{theorem}\label{th:gaussian_optimality_theorem}
In a quantum measurement system, suppose having a Gaussian quantum source state $\rho \sim \~{QN}(\bar m_\rho, \Sigma_\rho)$ and a classical Gaussian output distribution $\pi_Z \equiv \~N(\mu_Z, \Sigma_Z)$. Also, suppose using the quadratic distortion observable operator \eqref{eq:distortion_operator} as the distortion measure. Then, for any feasible distortion value $D$, the information gain $I_g(R;Z)$ is minimized by a Gaussian observable. This optimal Gaussian observable has the output ensemble representation $\{\tilde \rho_{N_G,z}, \pi_Z(dz)\}_{z\in \~Z}$ with post-measured states
\begin{align}
    \tilde \rho_{N_G,z}:= D(Kz) \rho_{N_G} D(Kz)^\dagger
\end{align}
where $\rho_{N_G} \sim \~{QN}(0,\Sigma_N)$ is a zero-mean Gaussian state representing the measurement noise with some $\Sigma_N$ 
such that $\Sigma_N \leq \Sigma_\rho$ to be determined and $K$ is a transformation matrix of the form
\begin{align}\label{eq:K-classical-optimal-transport-mapping}
    K = \Sigma_Z^{-1/2} \left(  \Sigma_Z^{1/2} (\Sigma_\rho - \Sigma_N)  \Sigma_Z^{1/2} \right)^{1/2}  \Sigma_Z^{-1/2}.
\end{align}
\end{theorem}
 \begin{proof}
With no loss of generality, we first assume that the quantum source state and the classical output distribution are both transported to the origin having zero means.
Assume having an arbitrary set of PMR states $\{\hat \rho_z\}_{z\in \~Z}$ with mean values $\bar m(z)$ and covariance matrices $\hat \Sigma(z)$. According to lemma \ref{lemma:lawoftotalvariance}, the following relation holds for the covariances:
\begin{align} \label{eq:lawofvariance_intheorem}
    \Sigma_\rho = \int \hat \Sigma(z) \pi_Z(dz) + \tilde \Sigma.
\end{align}
We further define a centralized version of the PMR states 
\begin{align}
    \hat \rho_{c,z} := D^{\dagger}(\bar m(z)) \ \hat \rho_z \ D(\bar m(z)).
\end{align}
 Thus, we have the following upper bound for the conditional entropy of PMR states
\begin{align}
    H(R|Z) =  \int H\left(\hat \rho_{c,z }\right) \pi_Z(dz)
    \leq   H\left( \hat \rho_N := \int \hat \rho_{c,z } \ \pi_Z(dz)\right)
    \leq H(\tilde \rho_{N_G}), 
    \label{eq:conditional-entropy-ineq}
\end{align}
where the first inequality follows from the concavity of von Neumann entropy \cite[Property 11.4.1]{wilde2013quantum_information_theory_book}. We defined $\hat \rho_N$ to be the average state with zero mean and the covariance matrix $ \Sigma_N$, where
\begin{align}
     \Sigma_N &= \Tr{\hat \rho_N R R^T } 
     = \Tr{\left(  \int D^{\dagger}(\bar m(z)) \ \hat \rho_z \ D(\bar m(z)) \pi_Z(dz)  \right) R R^T } 
     = \int \hat \Sigma(z) \pi_Z(dz). 
     \label{eq:noise_variance}
\end{align}
Then by comparing \eqref{eq:noise_variance} with \eqref{eq:lawofvariance_intheorem}  we find the covariance of the classical estimator,
% \begin{align}
   $ \tilde \Sigma =  \Sigma_\rho -  \Sigma_N$.
% \end{align}

Also in \eqref{eq:conditional-entropy-ineq}, in the last inequality, $\tilde \rho_{N_G}$ is introduced as a zero-mean Gaussian quantum state with the same covariance matrix $ \Sigma_N$  which appeals to the quantum Gaussian entropy maximization theorem \cite[Property (i)]{holevo_hirota1999capacityofquantumgaussianchannels}. Further, the distortion function for the arbitrary PMR ensemble is
\begin{align} 
    d(\hat \rho_z, \pi_Z) =  \frac{1}{2s} \ \int \left[\sp{\hat \Sigma(z) } +  \norm{z - \bar m(z)}_2^2 \right] \pi_Z(dz)= \frac{1}{2s} \ \sp{ \Sigma_N} +  \frac{1}{2s} \int \norm{z - \bar m(z)}_2^2 \pi_Z(dz). \label{eq:distortion_continuous_simplified}
\end{align}

We next form a different set of PMR ensembles equivalent to the quantum Gaussian measurement of the form
% \begin{align} \label{eq:gaussianform_pmr}
    $\tilde \rho_{z} = D(K z) \tilde \rho_{N_G} D(Kz)^\dagger$,
% \end{align}
with $K$ being some transformation matrix to be determined. 

Note that by using $\tilde \rho_z$, because the covariance matrix $ \Sigma_N$ is fixed, the first term of the distortion \eqref{eq:distortion_continuous_simplified} does not change.
 We further limit the selection of the $K$ matrix such that it satisfies the relation $\tilde \Sigma = K \Sigma_Z K^T$ with the covariance matrix of the estimator. This, again using lemma \ref{lemma:lawoftotalvariance}, makes sure that the marginal constraint is satisfied $\rho = \int \tilde \rho_z \pi_Z(dz)$  because the source quantum state is assumed to be Gaussian and can be fully determined by its first moment and covariance matrix \cite{depalma_quantum_optimal_transport_2021}.
 \ifmycomments
 {\color{blue}
 The detailed proof is as follows. Define $\tilde \rho_\text{avg}:= \int \tilde \rho_z \pi_Z(dz)$. Then for the first moment
 \begin{align}
     \bar m_{\tilde \rho_\text{avg}} = \Tr{\tilde \rho_\text{avg} R} = \int\Tr{ \tilde \rho_z R} \pi_Z(dz) = 0 = \bar m_\rho
 \end{align}
 And for the second moment, from the lemma \ref{lemma:lawoftotalvariance} we have
 \begin{align}
     \Sigma_{\tilde \rho_\text{avg}} &= \int \Sigma(z) \pi_Z(dz) + K \Sigma_Z K^T\\
     &= \Sigma_N + \tilde \Sigma\\
     &= \Sigma_\rho
 \end{align}
 where the second equality is forced by us fixing $K$ such that the relation holds and that for a fixed $z$ the conditional covariance matrix is $\Sigma_N$.
 }
 \fi
 The second term of distortion \eqref{eq:distortion_continuous_simplified} is only a function of $\bar m(z)$. 
 
 So by preserving the covariance matrix $\tilde \Sigma = \Sigma_\rho -  \Sigma_N$, we can provide the following lower bound 
\begin{align}
    \Exx{Z}{\norm{\bar m(Z) - Z}_2^2} \geq W_2^2\left(\~N(\Sigma_\rho -  \Sigma_N) \; , \; \~N(\Sigma_Z)\right) = \Exx{Z}{\norm{Z - KZ}_2^2}.
\end{align}
where the inequality appeals to  
the following arguments. The $l_2$ norm depends only on the first two moments of the distributions, hence, one can replace them to Gaussian distributions. The lower bound on this distance is then the Wasserstein distance between these Gaussian distributions. Then the last equality follows by choosing $K$ to be the transportation that achieves the minimum distance from Gaussian distribution with $\Sigma_Z$ to Gaussian distribution with $\tilde \Sigma$ \cite{knott1984ontheoptimalmappingofdistributions}, i.e., 
 \begin{align} \label{eq:Kmapping}
     K := \Sigma_Z^{-1/2} \left(  \Sigma_Z^{1/2} (\Sigma_\rho -  \Sigma_N)  \Sigma_Z^{1/2} \right)^{1/2}  \Sigma_Z^{-1/2}.
 \end{align}
\ifmycomments
{\color{blue}
\begin{lemma}
    Assuming a classical system having a set of source distributions $\Omega_X (0, \Sigma_X)$ and a set of destination distributions $\Omega_Y (0, \Sigma_Y)$ with zero mean and fixed given covariance, the minimum MSE distance between the two sets is the Wasserstein distance between two Gaussian distributions with given moments. i.e.,
    \begin{align}
        \min_{\substack{X \sim \~X \in \Omega_X \\ Y \sim \~Y \in \Omega_Y}} \Ex{\norm{X- Y}^2} = W_2^2(\~N(0, \Sigma_X), \~N(0,\Sigma_Y)).
    \end{align}
    that is achieved by a linear mapping $Y = K X$ such that 
    \begin{align}
       K = \Sigma_Z^{-1/2} \left(  \Sigma_Z^{1/2} (\Sigma_\rho - \Sigma_N)  \Sigma_Z^{1/2} \right)^{1/2}  \Sigma_Z^{-1/2}.
    \end{align}
The desired conclusion follows from the observation that $\Ex{\norm{X- Y}^2}$ is preserved when  $(X,Y)$ is replaced by the jointly  Gaussian pair $(X_G, Y_G)$ with the same joint covariance matrix.
\end{lemma}
}
\fi
Then the following lower bound holds for the distortion
\begin{align}
    d(\hat \rho_z , \pi_Z) \geq \frac{1}{2s} \sp{\Sigma_N} + \Exx{Z}{\norm{Z - KZ}_2^2}.
\end{align}
This shows that for a fixed noise covariance $ \Sigma_N$, the PMR ensemble of the Gaussian form $\tilde \rho_z$ with $\bar m(z) = K z$, does not decrease the conditional entropy while also preserving the distortion and marginal constraints. This proves that the Gaussian observable achieves the optimality.
\end{proof}

\begin{corollary}
   The following rate and distortion are achievable:
\begin{align}
    I_g(R;X) &= H_G(\Sigma_\rho) - H_G(\Sigma_N),\\
    d(\tilde \rho_{N_G,z},  \pi_Z) &= \frac{1}{2s}\left( \sp{\Sigma_N} + \Exx{\bar \pi_Z}{Z - KZ} + \norm{\bar m_\rho - \mu_Z}_2^2 \right),
\end{align}
where the expectation in the second term is with respect to $\bar \pi_Z(B) = \pi_Z(B - \mu_Z)$ for $B \in \~B(\~Z)$.
\end{corollary}

\subsection{The QC Gaussian Optimal Transport}

To obtain the OC rate-distortion function, by using the above Gaussian optimality theorem, it suffices to find the optimal  noise covariance matrix $\Sigma_N$ as follows:
\begin{align} \label{eq:Gaussian_optimization_problem1}
    R(D; R_c = \infty, \rho|| \pi_Z) &= \min_{\Sigma_N} H_G(\Sigma_\rho) - H_G(\Sigma_N)\\
    \text{subject to:}\quad&  \frac{1}{2s} \ \sp{\Sigma_N} +  \frac{1}{2s} \int \norm{z - Kz}_2^2 \bar \pi_Z(dz) + \frac{1}{2s} \norm{\bar m_\rho - \mu_Z}_2^2 \leq D, \nonumber
\end{align}
where $H_G(\Sigma_\rho)$ is given in \eqref{eq:gaussian-quantum-entropy}, $K$ is given by \eqref{eq:K-classical-optimal-transport-mapping}. By further simplifying the distortion constraint, the optimization problem is reduced to
\begin{align} 
   R(D; R_c = \infty, \rho|| \pi_Z) := &\min_{\Sigma_N}  \frac{1}{2} \text{Sp } g\left(\abs{\Delta^{-1} \Sigma_\rho} - \frac{I}{2}\right) - \frac{1}{2} \text{Sp } g\left(\abs{\Delta^{-1}\Sigma_N} -  \frac{I}{2}\right)  \label{eq:gaussian_optimization_objective}\\
    \text{subject to: }
    &\; \frac{1}{s } \ \sp{ (\Sigma_Z^{1/2} (\Sigma_\rho - \Sigma_N) \Sigma_Z^{1/2})^{1/2}} \geq D_\text{max} - D \label{eq:gaussian_optimization_distortionconst}\\
    &  \pm\frac{i}{2}\Delta \leq \Sigma_N \leq \Sigma_\rho,    \nonumber
\end{align}
where $\abs{A}:= \sqrt{A A^T}$ is the absolute value of the matrix $A$, and where 
\begin{align} \label{eq:Dmax_Gaussian}
    D_\text{max} := \frac{1}{2s} \ \sp{\Sigma_\rho + \Sigma_Z} + \frac{1}{2s} \norm{\bar m_\rho - \mu_Z}_2^2
\end{align}
is the zero-crossing point. This means no transmission rate is required to achieve any distortion value above  $D_\text{max}$; i.e. the distortion tolerance is high enough to allow Bob to generate the output independent of the source according to the given distribution. 
Next, we evaluate this OC rate-distortion function for the following example systems.

\subsubsection{The isotropic Gaussian source and destination}
In this part, we formulate the RLOT problem for the isotropic Gaussian systems. Assume having a Gaussian source state $\rho$ and a Gaussian destination distribution $\pi_Z$ with mean vectors $\bar m_\rho = m_{\rho,1} \bs{1}_{2s}$ and $\mu_Z = \mu_{Z,1} \bs{1}_{2s}$, and covariance matrices $ \Sigma_\rho = \sigma_\rho^2 I_{2s\times 2s}$ and  $\Sigma_Z = \sigma^2_Z I_{2s \times 2s}$, respectively. Then the following theorem directly follows.
\begin{theorem}
    The OC rate-distortion function of the QC Gaussian isotropic system is given by
    \begin{align}
    R\bigg(D; \infty,  \~{QN}(m_{\rho,1}, \sigma_\rho^2) || \~N(\mu_{Z,1} , \sigma_Z^2)\bigg) = s \cdot \left[ g\Big(\sigma_\rho^2 - 1/2 \Big) - g\Big(n^* \sigma_\rho^2 - 1/2\Big)\right]. \label{eq:gaussian_isotropic_rate}
\end{align}
where $n^*(D)$ is the noise-to-signal power ratio of the Gaussian observable
\begin{align}
    n^*(D) = \begin{cases}
        1 - \left(\frac{ D_\text{max} - D }{2 \sigma_\rho \sigma_Z} \right)^2 \quad  & D_\text{OT }\leq D \leq  D_\text{max},\\
        1 & D_\text{max} < D.
    \end{cases}
\end{align}
The upper threshold value 
$D_\mathrm{max}$
% $D_\mathrm{max} = (\sigma_\rho^2 + \sigma_Z^2) + (m_{\rho,1} - \mu_{Z,1})^2$ 
is given by \eqref{eq:Dmax_Gaussian} and the lower threshold value is the 2-Wasserstein distance of the QC Gaussian isotropic system given by 
\begin{align} \label{eq:gaussian_isotropic_DOT}
    W_2^2\bigg(\~{QN}(m_{\rho,1}, \sigma_\rho^2) \;, \; \~N(\mu_{Z,1}, \sigma_Z^2)\bigg) = D_\text{OT} = D_\text{max} - 2\sigma_Z \sqrt{ (\sigma_\rho^2 - \frac{1}{2})}.
\end{align}
\end{theorem}
\begin{proof}
    As the marginal covariance matrices are both isotropic, it easily follows that the noise covariance $\Sigma_N$ is isotropic as well. Then let $\Sigma_N = n \sigma_\rho^2 I_{2s\times 2s}$, where $n$ is the variable of optimization. Then by simplifying the objective function in \eqref{eq:gaussian_optimization_objective}, it reduces to the form of \eqref{eq:gaussian_isotropic_rate} with $n^*$ to be determined. Next, the distortion constraint \eqref{eq:gaussian_optimization_distortionconst} gives
    $n \leq 1 - \left(\frac{D_\text{max} - D}{2 \sigma_Z \sigma_\rho}\right)^2$,
    when having $D \leq D_\text{max}$ and $n = 1$ otherwise. Finally, note that the feasibility constraint $\Sigma_N \geq \pm \frac{i}{2} \Delta $ imply $n \geq 1 / (2\sigma_\rho^2)$. This gives the corresponding minimum attainable distortion to be of the form \eqref{eq:gaussian_isotropic_DOT}.
\end{proof}

\subsubsection{The one-mode Gaussian case}
Solving the optimization problem for the one-mode systems gives the following-closed form expression for the rate-distortion function. We define a few parameters that are necessary for the theorem. Define  a matrix $\Sigma = \Sigma_Z^{1/2} \Sigma_\rho \Sigma_Z^{1/2}$ 
with eignen decomposition $\Sigma = U \text{diag}[\sigma_1^2, \sigma_2^2] U^T$.
 
Define $X=U \text{diag}[x_1, x_2] U^T$,
where $x_1$ and $x_2$ satisfy
\begin{align} \label{eq:ratedist_onemode_solution}
    \frac{x_1}{\sigma_1^2 - x_1^2} = \frac{x_2}{\sigma_2^2 - x_2^2}, \quad x_1 + x_2 = D_\text{max} - D.
\end{align}
Also define $X_{OT} = U \text{diag}[y_1, y_2] U^T$,  where $y_1$ and $y_2$ satisfy
\begin{align} \label{eq:gaussian_wasserstein_syseq}
     \frac{y_1}{\sigma_1^2 - y_1^2} = \frac{y_2}{\sigma_2^2 - y_2^2} ,\quad
         (y_1^2 - \sigma_1^2)(y_2^2 - \sigma_2^2) = \frac{1}{4} \det{\Sigma_Z}.
\end{align}

\begin{theorem} \label{th:gaussian_onemode}
For a one-mode Gaussian quantum measurement  system ($s=1 $), the OC rate-distortion function for the interval $D_{OT} \leq D \leq D_\text{max}$ is given by
\begin{align} 
    R(D; \infty,  \~{QN}(0, \Sigma_\rho) || \~N(0, \Sigma_Z)) =  \ g\left(\sqrt{\det{\Sigma_\rho}} - 1/2\right) - g\left(\sqrt{\det{\Sigma_N}} - 1/2\right)
\end{align}
where the noise covariance matrix 
 is $\Sigma_N = \Sigma_\rho - \Sigma_Z^{-1/2} X^2 \Sigma_Z^{-1/2}$.
  Furthermore, the QC Wasserstein distance $D_{OT}$ for the one-mode case is obtained by
 \begin{align}
     W_2^2\bigg(\~{QN}(0, \Sigma_\rho) \;,\; \~N(0, \Sigma_Z)\bigg) = D_{OT} = D_{\text{max}} - \sp{X_{OT}}.
 \end{align}
\end{theorem}
\begin{proof}
    See Appendix \ref{appD}
\end{proof}

\begin{remark}
 In the classical systems, the 2nd-order Wasserstein distance between Gaussian distributions occurs with infinite communication rate ($R = \infty$). In  contrast, the required communication rate for QC Wasserstein distance is bounded, and obtained as,
 \begin{align}
     R_{W_2}(\rho ||\pi_Z) = g\left(\sqrt{\det{\Sigma_\rho}}- 1/2\right),
 \end{align}
and is achieved when $\abs{\Delta^{-1}\Sigma_N} =  {I}/{2}$.%$\det{\Sigma_N} = 1/4$.
  In other words, increasing the rate beyond this point cannot reduce the distortion any further. This result can also be interpreted as a restatement of the famed Heisenberg uncertainty principle which says the position and momentum cannot be determined with arbitrary accuracy. 
 % Moreover, this also aligns with Holevo's result on the maximum accessible information of the Gaussian quantum ensembles in \cite[Theorem 2]{holevo2020-gaussian-maximizers-observables}.
\end{remark}

\ifmycomments
\myqu{How do you claim that this is because of Heisenberg and the noise having a lowerbound?}
\myans{Recall that in the classical setting, we had $X$ and $Y$, then we would write $Y - X = Y - \hat X + \hat X - X$, where $\hat X = \Ex{X|Y}$. Then we would name $N = X - \hat X$ to be the noise. Then with some calculations, we would arrive to the rate $I(X;Y) = H(X) - H(N)$. Then in classical case, this would be infinity because $\sigma_N =0$ and this resulted in $H(N) = -\infty$. But in quantum case, the minimum value of $H(N) = 0$ because the noise level cannot drop less than the $1/2 $ threshold because of Heisenberg's inequality}

\myqu{Why Holveo \cite[Theorem 2]{holevo2020-gaussian-maximizers-observables} has different formula?}\myans{Because it is a totally different thing. The accessible information he works with is the classical information. That is he prepares classical data with quantum ensemble and then performs measurement on that. All his entropy functions are Shannon entropy and the formulations are also $\log \det (.)$ which aligns with classical formulation. Here we are considering the quantum information of the system which is why we have the $g(.)$ function.}

\myqu{The information gain is maximum of the mutual information! not the minimum!} \myans{In our OC rate-distortion theorem, we reduce the distortion by allowing the system to have lower noise covariance. It shows that when we allow minimum possible distortion, which corresponds to the minimum possible noise, the corresponding rate is the information gain (maximum achievable rate)}
\fi

The extension of this problem to the multi-mode Gaussian quantum system with independent modes is provided in the thesis \cite{hafez2023thesis}.

\ifmycomments
{\color{blue}
\subsubsection{The multi-mode Gaussian case with independent modes}
We can easily extend the previous one-mode Gaussian system to $s$-mode independent Gaussian system, where source state is an $s$-mode Gaussian state with a block diagonal matrix $\Sigma_\rho = \oplus_{i=1}^s \Sigma_{\rho_i}$, where $\Sigma_{\rho_i}$ is the covariance matrix of $i$-th mode. Similarly, we assume the output distribution to be a multivariate Gaussian system with $s$ independent 2-dimensional Gaussian distributions comprising the covariance $\Sigma_Z = \oplus_{i=1}^s \Sigma_{Z_i}$. Therefore, in this case the optimization problem in \eqref{eq:quantum_gaussian_optimization} is generalized as follows
\begin{align} \label{eq:nmode-quantum_gaussian_optimization}
    R(D; \infty,  \~{QN}(0, \Sigma_\rho) || \~N(0, \Sigma_Z)) = & \min_{\stackrel{\Sigma_{N_i},}{i=1,\cdots, s}}  \sum_{i=1}^s \left[ g\left(\sqrt{\det{\Sigma_{\rho_i}}} - 1/2\right) - g\left(\sqrt{\det{\Sigma_{N_i}}} - 1/2\right) \right]\\
    \text{subject to }&\; \sum_{i=1}^s \frac{1}{s} \sp{ (\Sigma_{Z_i}^{1/2} (\Sigma_{\rho_i} - \Sigma_{N_i}) \Sigma_{Z_i}^{1/2})^{1/2}} \geq D_\text{max} - D.
\end{align}
Again, similar to the proof in Appendix \ref{appD}, we change the variables of optimization to the set of variables 
    $X_i:= \left(  \Sigma_{Z_i}^{1/2} (\Sigma_{\rho_i} - \Sigma_{N_i})  \Sigma_{Z_i}^{1/2}   \right)^{1/2}$ 
 for $i = 1,\cdots, s$.
This results in the following equivalent optimization problem:
\begin{align}
    \min_{\substack{X_i,\\i=1,...,s}} -  \sum_{i=1}^s  \ & g\left(  \sqrt{ \frac{\det{\Sigma_i - X_i^2}}{\det{\Sigma_{Z_i}}}}    - \frac{1}{2}  \right),\\
    &  \det(\Sigma_i - X_i^2 ) \geq \frac{1}{4} \det(\Sigma_{Z_i}) ,\quad  0\leq X_i^2,\\
    &  \sum_i\sp{X_i} \geq c,
\end{align}
where $\Sigma_i := \Sigma_{Z_i}^{1/2}\Sigma_{\rho_i} \Sigma_{Z_i}^{1/2}$ and  $c = D_\text{max} - D$. The above problem is solved with a similar approach as the one-mode case by considering $X_i$ to be same diagonalizable as $\Sigma_i$, which results in an extension of the one-mode solution \eqref{eq:onemode_solution} as follows:
\begin{align}
        \phi(X_i) \cdot \log\frac{\phi(X_i) + 1/2}{\phi(X_i) - 1/2} \frac{x_{ij}}{\sigma_{ij}^2 - x_{ij}^2}&= \pi, \quad \text{for }i = 1, \cdots, s, \quad j = 1,2 \label{eq:pi_xi_equation_nmode}\\
    \sum_{i=1}^s \sum_{j=1}^2 x_{ij} &= c.
\end{align}
where for the $i$-th mode, $X_i = U_i \text{diag}(x_{ij})_{j=1,2} U_i^T$ and $\Sigma_i = U_i \text{diag}(\sigma^2_{ij})_{j=1,2} U_i^T$.
}
\fi

% \subsection{The zero-crossing point ($D_\text{max}$)}
% The zero-crossing point is the extreme point at which there is no need for any transmission rate; i.e. $R=0$ to achieve the required distortion level. For any distortion value above this point also the rate will be zero. Because in this region, the acceptable distortion level is high enough to allow  Bob to  generate a completely random output with the given distribution, independent of the source, and satisfy the distortion constraints, so no measurement and transmission are required, i.e.
% \begin{align}
%     M_\text{ind}(dz) = I_\~H \pi(dz).
% \end{align}
% The zero-crossing distortion for a source state $\rho \sim \~{QN}(m_\rho, \Sigma_\rho)$ where $m$ is the mean and $\Sigma_\rho$ is the covariance matrix, and output distribution $Z\sim \~N(\mu_Z, \Sigma_Z)$ is given by
% \begin{align}
%     D_\text{max} := d(\rho, M_\text{ind}) &= \frac{1}{2} \text{Sp}\{\Sigma_\rho\} + \frac{1}{2} \Exx{\pi}{\norm{Z - m_\rho}^2}\\
%     &= \frac{1}{2} \left[  \text{Sp}\{\Sigma_\rho + \Sigma_Z \} +  \norm{m_Z - m_\rho}^2 \right]
% \end{align}
% which is similar to the distance of two independent classical Gaussian distributions.

\subsection{Numerical Example of Quantum-Classical Gaussian system}
We consider an example of the above system when the input state is a Gaussian quantum state and output system is a classical Gaussian distribution with covariance matrices
\begin{align} \label{eq:gaussian_numerical_example}
    \Sigma_\rho = \begin{bmatrix}
         1.0783 & -0.8976\\
       -0.8976&  1.3155 
    \end{bmatrix}, \quad
    \Sigma_Z = \begin{bmatrix}
        1.7471 & -1.2224\\
        -1.2224 &  0.8583
    \end{bmatrix}.
\end{align}
One can find the Symplectic eigenvalues of the input covariance matrix as $\alpha_s = 0.7828$ and the eigenvalues of the output matrix as $\lambda_{1,2} = 0.002, 2.6034$. The displacement is assumed to be zero for both source and destination. For this system, the suboptimal OC rate-distortion function generated by calculating the noise covariance from the above expressions is given in Figure \ref{fig:gaussian-rate-dist-plot}. The interesting observation is that in contrast to the classical Gaussian optimal transport for which the Wasserstein distance is achieved when the information rate is infinite $(R = \infty)$, in the QC setting, due to the Heisenberg Uncertainty principle, the maximum rate cannot be infinite as the accessible information of measurement is limited according to Theorem 2 of \cite{holevo2020-gaussian-maximizers-observables}.
\begin{figure}
    \centering
    \includegraphics[width=0.7\textwidth]{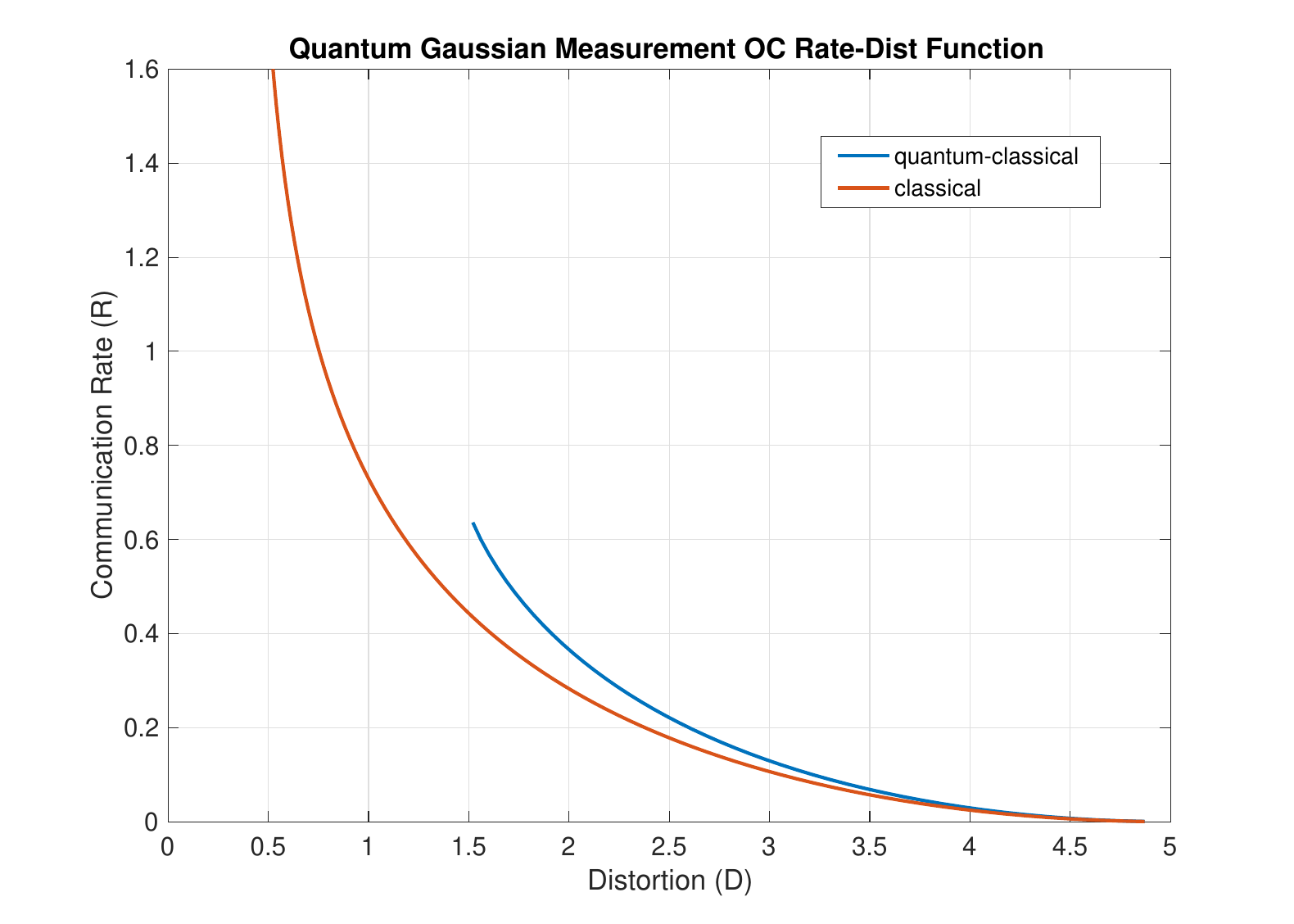}
    \caption[Plot of Rate-Limited QC Wasserstein Distance of 2nd Order]{Plot of the OC rate-distortion function of (Blue): the QC Gaussian measurement system with zero mean and marginal covariance matrices in \eqref{eq:gaussian_numerical_example}, and (Red): the  fully classical Gaussian system with Gaussian marginal distributions with zero mean and same covariance matrices as \eqref{eq:gaussian_numerical_example}.}
    \label{fig:gaussian-rate-dist-plot}
\end{figure}
The simulations are performed using Strawberry fields \cite{Killoran_2019strawberry_sim} and Walrus \cite{Gupt2019walrus_sim} packages and Matlab.

% discussion and conclusion

\section{Conclusion}
We established a Quantum-to-Classical Rate-Limited Optimal Transport problem formulation involving both discrete and continuous-variable quantum systems.
% We further used this source coding theorem to establish the concept of quantum-to-classical RLOT. 
We provided a single-letter computable characterization of the performance limits of this problem using quantum information quantities.
The performance limit is given by the achievable rate region which is the collection of all communication and common randomness rates $(R,R_c)$ fulfilling the distortion level in accordance with a generally defined form of distortion observable. 
% The theorem was further extended to the CV quantum systems by introducing a CV quantum coding protocol that uses clipping projections to truncate the infinite-dimensional Hilbert space and then apply the discrete coding protocol.
Considering a general form of distortion observable in the coding theorem, made it possible to use different distortion measures corresponding to discrete and continuous-variable QC systems.
Next, the example of Qubit-Bernoulli, the RLOT was evaluated for the unlimited common randomness and entanglement fidelity distortion measure.
% The RL optimal transportation cost function was provided analytically for the Qubit measurement systems and the Gaussian measurement systems in the presence of an unlimited amount of common randomness. 
For the Gaussian measurement system, by reformatting the quadratic cost operator of \cite{depalma_quantum_optimal_transport_2021}, we introduced the QC rate-limited Wasserstein distance of order 2. Our Gaussian optimality theorem further showed that the Gaussian RL 2-Wasserstein distance is achieved by the Gaussian measurements. Based on that, the closed-form expressions for the rate-limited QC 2-Wasserstein distance of Gaussian isotropic systems were provided. The conventional 2-Wasserstein distance was also obtained as the special case when $R=\infty$.

As described in the introduction, in the finite-dimensional systems, the fully quantum version of rate-distortion coding has been studied for the memoryless sources in works like \cite{datta_wilde_2012_quantum_ratedistortion_reverseshanon,wilde2013quantum_ratedist_auxiliary,devetak2002quantum_ratedist_memorylesss,Salek_quantumratedist_relevant}, and for one-shot setting in \cite{datta2013oneshot_quantum_compression,atifpradhan2023quantumsoftcovering}.
We plan to address as a part of our future work the QC RLOT problem in one-shot settings for both finite-dimensional and continuous-variable systems, in addition to the evaluations of the Quantum Gaussian Measurement systems. Another important line of work that may be of interest is to consider the fully quantum counterpart of these RLOT problems both in asymptotic IID and one-shot regimes in finite-dimensional and continuous-variable settings as well as their evaluations. We also plan to address the applications of these concepts in developing new quantum measure concentration inequalities and their relation to the capacity of the quantum relay channels.

\begin{appendices}

\section{Proof of Useful Lemmas } \label{appA}

\subsection{Proof of Proposition \ref{prop:continuous_rate_inequalities}} \label{appA-cont-rate-inequalities}
The following lemmas are used in this proof.
% Shirokov's proposition of convergence of information gain
\begin{lemma} \label{lemma:information_gain_convergence}
    \cite[Proposition 6]{shirokov2021lower_semicontinuity}
    Let $\{{\rho^{(n) A}}\}$ be a sequence of states converging to a state ${\rho^{(0) A}}$ (See Section 11.1 of \cite{Holevo_2019_quantum_systems_book}). Also, let $\{M_n\}$ be any arbitrary sequence of POVMs weakly converging to a POVM $M_0$, with the outcome space $\~W$. If either $\abs{\~W} < +\infty$ or $\lim_{n \to \infty} H({\rho^{(n) A}}) = H({\rho^{(0) A}}) < +\infty$ then 
    % \begin{align}
       $ \lim_{n \to \infty } I_g(M_n, {\rho^{(n) A}})  = I_g(M_0, {\rho^{(0) A}})$,
    % \end{align}
    where $I_g(M,\rho)$ is the information gain as defined by Groenwold.
\end{lemma}

% lemma of weak convergence
\begin{lemma} \label{lemma:weak_convergence_M}
    The sequence of $\hat M_{k_1}$ POVMs converges weakly to the $M_W$ POVM.
\end{lemma}
\begin{proof}
    See Appendix \ref{appA-weak_convergence_M}.
\end{proof}
Then the proof of proposition \ref{prop:continuous_rate_inequalities} is as follows.
\begin{proof}
To prove the first inequality of \eqref{eq:clippedrate-inequality} we directly use Lemma \ref{lemma:information_gain_convergence} by Shirokov.
To use this lemma in our system, we demonstrate a refined POVM by combining the clipping projection and the measurement POVM $M_W$ into,
 % \begin{align*}
     $\hat M_{k_1} = \left\{  \left\{  M_W(B) \Pi_{k_1}, B\in \~B  \right\}, I -\Pi_{k_1}  \right\}$.
 % \end{align*}
It is straightforward that this refined POVM is a valid POVM. Next, Lemma \ref{lemma:weak_convergence_M} proves the weak convergence of this refined POVM.
Therefore, by applying Lemma \ref{lemma:information_gain_convergence} and gentle measurement lemma, it follows that $\lim_{k_1 \to \infty} I_g(R_{k_1},W'_{k_1}) = I_g(R;W)$.
% \begin{align}
%     \lim_{k_1 \to \infty} I_g(R_{k_1},W'_{k_1}) = I_g(R;W).
% \end{align}
To prove the second rate inequality of \eqref{eq:clippedrate-inequality}, we first apply the chain rule for mutual information, $I(W,A_{k_1};X_{\bar{k}_2}) = I(W;X_{\bar{k}_2}) + I(A_{k_1};X_{\bar{k}_2}|W) \leq H(A_{k_1}) + I(W;X_{\bar{k}_2})$.
% \begin{align*}
%     I(W,A_{k_1};X_{\bar{k}_2}) = I(W;X_{\bar{k}_2}) + I(A_{k_1};X_{\bar{k}_2}|W) \leq H(A_{k_1}) + I(W;X_{\bar{k}_2}).
% \end{align*} 

Hence,
% \begin{align}
    $I(W,A_{k_1};X_{\bar{k}_2}) - H(A_{k_1}) \leq I(W;X_{\bar{k}_2}) \leq I(W,A_{k_1};X_{\bar{k}_2})$.
% \end{align}
As $H(A_{k_1})$ decays to zero when $k_1 \to \infty$, from squeeze theorem we have the following limit
\begin{align}\label{eq:A-equals-first-system}
    \lim_{k_1 \to \infty} I(W,A_{k_1}; X_{\bar{k}_2}) = I(W;X_{\bar{k}_2}).
\end{align}
Further, applying another chain rule we get
\begin{align}
     I(W,A_{k_1}; X_{\bar{k}_2})% &=  I(A_{k_1}; X_{\bar{k}_2}) + I(W; X_{\bar{k}_2} | A_{k_1}) \nonumber \\
    &= I(A_{k_1}; X_{\bar{k}_2}) +  (1 - P_{k_1}) \cdot I(W; X_{\bar{k}_2}|A_{k_1}=0)  + P_{k_1} \cdot I(W; X_{\bar{k}_2}|A_{k_1}=1) .
\end{align}
As the clipping region grows to infinity, the probability of clipping decays to zero $\lim_{k_1 \to \infty} P_{k_1} = 0$. Then, because $A_{k_1}$ is a simple Bernoulli random variable with an arbitrarily small probability, we have $I(A_{k_1}; X_{\bar{k}_2}) \leq H(A_{k_1}) \to 0$. 
Also for the third term, the following asymptotic bound exists:
\begin{align}
    \lim_{k_1 \to \infty } P_{k_1} \cdot I(W;X_{\bar{k}_2}|A_{k_1}=1) \leq \lim_{k_1 \to \infty } P_{k_1} \cdot H(X_{\bar{k}_2})  = 0,
\end{align}
where we appealed to the fact that quantized output with limited alphabet has limited entropy.

Finally, consider that $I(W; X_{\bar{k}_2}|A_{k_1}=0) = I(W'_{k_1}, X'_{k_1,\bar{k}_2})$ holds by definition because when the input state is inside the clipping subspace, the system behaves as if no clipping was performed. Combining all together, for any fixed $\bar{k}_2$, we have
\begin{align} \label{eq:A-equals-second-system}
    \lim_{k_1 \to \infty} I(W,A_{k_1};X_{\bar{k}_2}) = \lim_{k_1 \to \infty} I(W'_{k_1}, X'_{k_1,\bar{k}_2}) .
\end{align}
Then \eqref{eq:A-equals-first-system} and \eqref{eq:A-equals-second-system} together show that for any fixed $\bar{k}_2$ we have $\lim_{k_1 \to \infty} I(W'_{k_1}, X'_{k_1,\bar{k}_2}) = I(W,X_{\bar{k}_2}) \leq I(W;X)$,
% \begin{align}
%     \lim_{k_1 \to \infty} I(W'_{k_1}, X'_{k_1,\bar{k}_2}) = I(W,X_{\bar{k}_2}) \leq I(W;X),
% \end{align}
where the inequality follows directly from the Data Processing Inequality and completes the proof. In addition, note that as $k_2,k_2' \to \infty$, $X_{\bar{k}_2}$ converges weakly to $X$. Therefore, using lower semi-continuity of mutual information \cite{pinsker1964_information_stability, posner1975random_coding_stratergies}, combined with the above data-processing inequality we further have $\lim_{k_2,k_2' \to \infty}  \lim_{k_1 \to \infty} I(W'_{k_1}, X'_{k_1,\bar{k}_2}) = \lim_{k_2,k_2' \to \infty} I(W,X_{\bar{k}_2}) = I(W;X)$.
% \begin{align}
%     \lim_{k_2,k_2' \to \infty}  \lim_{k_1 \to \infty} I(W'_{k_1}, X'_{k_1,\bar{k}_2}) = \lim_{k_2,k_2' \to \infty} I(W,X_{\bar{k}_2}) = I(W;X).
% \end{align}
\end{proof}

% % %%%%%%%%%%%%%%%%%%%%%%%%%%%%%%%%%%%%%%%%

\subsection{Proof of Lemma \ref{lemma:weak_convergence_M}} \label{appA-weak_convergence_M}
By definition of the weak convergence of operators in  \cite[Section 11.1]{Holevo_2019_quantum_systems_book}, it suffices to show that for each subset $B \in \~B$ and any two arbitrary states $\phi, \psi \in \~H$, it holds that $\lim_{k_1 \to \infty} \braket{\psi \Big|\hat M_{k_1}(B)}{\phi} = \braket{\phi|M(B)}{\psi}$.
% \begin{align*}
%     \lim_{k_1 \to \infty} \braket{\psi \Big|\hat M_{k_1}(B)}{\phi} = \braket{\phi|M(B)}{\psi}.
% \end{align*
Starting with $I- \Pi_{k_1}$, we show that this operator vanishes to zero when $k_1 \to \infty$, for any $\psi, \phi \in \~H$
\begin{align*}
    \lim_{k_1 \to \infty} \Big|\braket{\phi|(I - \Pi_{k_1})}{\psi} \Big|  &\leq \lim_{k_1 \to \infty}\norm{(I - \Pi_{k_1})\ket{\psi}}_2 . \norm{(I - \Pi_{k_1})\ket{\phi}}_2 \nonumber \\
    &= \lim_{k_1 \to \infty} \sqrt{\braket{\phi|(I - \Pi_{k_1})}{\phi}} \cdot \sqrt{\braket{\psi|(I - \Pi_{k_1})}{\psi}},
\end{align*}
where the inequality appeals to Cauchy-Schwartz's inequality for Hilbert space. By using the definition of the projector operator and Bessel's inequality and the fact that $\ket{n}$ are orthonormal set, we have 
% \begin{align*}
$\braket{\phi|\Pi_{k_1}}{\phi} = \sum_{n=1}^{k_1} |\braket{n}{\phi}|^2    \leq \norm{\phi}^2$.
% \end{align*}
The inequality turns into equality when the orthonormal set is a complete orthonormal basis (Parseval's identity), which is when $k_1 \to \infty$. Substituting this into the inner-product expression proves that 
% \begin{align*}
    $\lim_{k_1 \to \infty} \Big|\braket{\phi|(I - \Pi_{k_1})}{\psi} \Big| = 0$.
% \end{align*}
Moreover, for the other operators, we have $\lim_{k_1 \to \infty} \braket{\phi| M(B) \Pi_{k_1}}{\psi} =\lim_{k_1 \to \infty} \braket{\phi'|\Pi_{k_1}}{\psi}= \braket{\phi'}{\psi}\nonumber = \braket{\phi|M(B)}{\psi}$,
% \begin{align*}
%     \lim_{k_1 \to \infty} \braket{\phi| M(B) \Pi_{k_1}}{\psi} =\lim_{k_1 \to \infty} \braket{\phi'|\Pi_{k_1}}{\psi}= \braket{\phi'}{\psi}\nonumber = \braket{\phi|M(B)}{\psi},
% \end{align*}
where $\ket{\phi'} = M(B)\ket{\phi}$, and the argument is similar to previous operator. These together prove that the sequence of POVMs $\hat M_{k_1}$ weakly converge to the $M$ POVM.

\subsection{Proof of Lemma \ref{lemma:single-letter-distortion-upperbound}} \label{appA-singleletter-distortion}
We condition the average n-letter distortion on the sequence of clipping errors $A_{k_1}^n$ as follows:
\begin{align}
    % line 1
    &d_n(R^n, {\hat X}^n_{k_1,\bar{k}_2}|\neg E_{ce}) = \Exx{A_{k_1}^n | \neg E_{ce}}{d_n(R^n, \hat X^n_{k_1,\bar{k}_2} | A_{k_1}^n)} \nonumber\\
    % % line 2
    % &=\Exx{A_{k_1}^n | \neg E_{ce}}{\frac{1}{n} \sum_{i=1}^n d(R_i, \hat X_{i,k_1,\bar{k}_2} | A_{k_1}^n) }\nonumber\\
    % line 3
    &=\Exx{A_{k_1}^n | \neg E_{ce}}{\frac{1}{n} \sum_{i: A_{i,k_1}=1} d(R_i, \hat X_{i,k_1,\bar{k}_2} | A_{k_1}^n) + \frac{1}{n} \sum_{i: A_{i,k_1}=0} d(R_i, \hat X_{i,k_1,\bar{k}_2} | A_{k_1}^n) }\nonumber\\
    % line 4
    &=\Exx{T | \neg E_{ce}}{\frac{n-T}{n} d(R,  \hat X_\text{local} | A_{k_1}=1) + \frac{T}{n} d_T(R_{k_1}^T, \hat X^T_{k_1,\bar{k}_2})  }\nonumber\\
    % line 5
    &\leq  (P_{k_1} + \epsilon_{cl}) d\left(R, \hat X_\text{local} | A_{k_1}=1\right) +    \Exx{T | \neg E_{ce}}{d_T(R_{k_1}^T, \hat X^T_{k_1,\bar{k}_2})} \label{eq:continuous-nletter-distortion-noEce}
\end{align}
where in the last equality, we generate a local random value $\hat X_\text{local}$ at Bob's side for any sample with an asserted error bit $A_{i,k_1}=1$. To determine the second term, note that the event $\neg E_{ce}$ means $T \geq t_\text{min}$. Then for any value of $\bar\varepsilon > 0$, by the Theorem \ref{th:maintheorem}, 
for all sufficiently large $n$ (resulting in a sufficiently large value of $t_\text{min}$) and for any   $t \geq t_\text{min}$, we have that $d_t(R^t_{k_1}, {X'}^t_{k_1,\bar{k}_2}) \leq d(R_{k_1}, {X'}_{k_1,\bar{k}_2}) + \bar\varepsilon$, thus,
\begin{align}
    d_t(R^t_{k_1}, {\hat X}^t_{k_1,\bar{k}_2})
    %&= d_t(R^t_{k_1}, {X'}^t_{k_1,\bar{k}_2}) + \Big(d_t(R^t_{k_1}, {\hat X}^t_{k_1,\bar{k}_2}) -  d_t(R^t_{k_1}, { X'}^t_{k_1,\bar{k}_2}) \Big) \nonumber\\
    %3rd line
    &\leq d(R_{k_1}, {X'}_{k_1,\bar{k}_2}) + \bar\varepsilon + \Big(d_t(R^t_{k_1}, {\hat X}^t_{k_1,\bar{k}_2}) -  d_t(R^t_{k_1}, { X'}^t_{k_1,\bar{k}_2}) \Big). \label{eq:cont-nletter-separation}
\end{align}
 % Note that we have coupled the parameters $n, k_1, \epsilon_{cl}$ in a way that the value of $t_\text{min}$ is sufficiently large. 

The second part of \eqref{eq:cont-nletter-separation} is the distortion caused by the dequantizer block in the receiver. 
However, we cannot use the triangle inequality to sum up the distortion of the blocks of the system. Instead, we expand the expressions of each term and find the difference. Define $\{\Lambda^{\~X^t_{\bar{k}_2}}_{x^t}\}_{x^t \in \~X^t_{\bar{k}_2}}$ as the $t$-collective measurement of the $t$-letter discrete coding scheme which generates discrete output sequence with perfect IID pmf $\mu_{X_{\bar{k}_2}}$.  Also define the continuous measurement POVM $\hat \Lambda \equiv \{\hat \Lambda(B), B\in \~B(\~X^t)\}$ by combining the discrete measurement with the dequantizer optimal transport block,  which is expressed for any event $B \subseteq \~B(\~X^t)$ by
% \begin{align*}
    $\hat \Lambda(B) := \sum_{x^t \in \mathcal{X}_{\bar{k}_2}^t} \Lambda^{\~X^t_{\bar{k}_2}}_{x^t} \pi^t_{OT} ( B| x^t ).$
% \end{align*}

The conditional distribution $\pi_{OT}(.|x^t)$ for $x^t\in \~X^t_{\bar{k}_2}$ representing the dequantizer optimal transport is a memoryless channel because the discrete coding produces IID discrete pmf $\mu_{X_{\bar{k}_2}}$ and the final output is required to have IID continuous distribution $\mu_X$. For  any event $B \subseteq \~B(\~X^t)$ it can be expressed by
% \begin{align}
    $\pi^t_{OT} ( B | x^t ) = \mu^t_X(A \cap {RQ}_{\bar{k}_2}(x^t)) \; / \; \mu^t_X\left(\~{RQ}_{\bar{k}_2}(x^t)\right),$
% \end{align}
where  $\~{RQ}_{\bar{k}_2}(x^t)$ is the quantization region of $x^t \in \~X_{\bar{k}_2}^t$.
Next, the distortion of the discrete coding block for the $i$-th local state is expanded as follows:
\begin{align}
    % line 1
    &d(R_{i,k_1} , X'_{i,k_1,\bar{k}_2}) = \sum_{x^t \in \~X^t_{\bar{k}_2}} \Tr{ \Trxx{[t]\setminus i}{\omega_{k_1} \Lambda^{\~X^t_{\bar{k}_2}}_{x^t} \omega_{k_1} }  \Delta(x_i)} \nonumber\\
    % line 2
    &= \sum_{x_i \in \~X_{\bar{k}_2}} 
    \Tr\bigg\{ 
        \Tr_{[t]\setminus i}\bigg\{
            \omega_{k_1}\bigg(
                \sum_{x^{[t]\setminus i} \in \~X^{t-1}_{\bar{k}_2}}  \Lambda^{\~X^t_{\bar{k}_2}}_{x^t} 
            \bigg)\omega_{k_1} 
        \bigg\} \Delta(x_i)
    \bigg\} 
    = \sum_{x \in \~X_{\bar{k}_2}} \Tr \bigg\{ \hat \rho^{R_i}_{x}  \Delta(x)\bigg\} \mu_{X_{\bar{k}_2}}(x), \label{eq:distortion-Xprime}
 \end{align}
where $\omega_{k_1}:= \sqrt{\rho_{k_1}^{\otimes t}}$ and $    \hat \rho^{R_i}_{x} 
    :=\frac{1}{\mu_{X_{\bar{k}_2}}(x)} 
    \Tr_{[t]\setminus i}\bigg\{
        \omega_{k_1}\bigg(
                \sum_{
                    x^{[t]\setminus i} \in \~X^{t-1}_{\bar{k}_2}
                }  \Lambda^{\~X^t_{\bar{k}_2}}_{x^t} 
        \bigg)\omega_{k_1} 
    \bigg\}$,  
% \begin{align*}
%     \hat \rho^{R_i}_{x} 
%     :=\frac{1}{\mu_{X_{\bar{k}_2}}(x)} 
%     \Tr_{[t]\setminus i}\bigg\{
%         \omega_{k_1}\bigg(
%                 \sum_{
%                     x^{[t]\setminus i} \in \~X^{t-1}_{\bar{k}_2}
%                 }  \Lambda^{\~X^t_{\bar{k}_2}}_{x^t} 
%         \bigg)\omega_{k_1} 
%     \bigg\} ,
% \end{align*}
is the PMR state of the $i$-th local system after discrete measurement POVM. The state is normalized by $\mu_{X_{\bar{k}_2}}(x)$ because the output-constrained coding implies the IID distribution of the output.
We further have the distortion of the $i$-th local system up to the output of dequantizer as,
\begin{align}
    % line 1
    d(R_{i,k_1} , \hat X_{i,k_1,\bar{k}_2}) &= \int_{\~X^t} \Tr{ \Trxx{[t]\setminus i}{\omega_{k_1} \hat \Lambda(dz^t) \omega_{k_1} }  \Delta(z_i)}  \nonumber\\
    % % line 3
    % &= \int_{\~X^t} \Tr{ \Trxx{[t]\setminus i}{\omega_{k_1} 
    % \left(\sum_{x^t \in \mathcal{X}_{\bar{k}_2}^t} \Lambda^{\~X^t_{\bar{k}_2}}_{x^t} \pi^t_{OT} \Big( dz^t\Big| x^t \Big) \right) 
    % \omega_{k_1} }  \Delta(z_i)}  \nonumber\\
    % line 4
    &= \sum_{x^t\in \~X^t_{\bar{k}_2}}\int_{\~{RQ}_{\bar{k}_2}(x^t)} \Tr{ \Trxx{[t]\setminus i}{\omega_{k_1} \Lambda^{\~X^t_{\bar{k}_2}}_{x^t}\pi^t_{OT} \Big( dz^t\Big| x^t \Big) \omega_{k_1} }  \Delta(z_i)} \nonumber \\
    % line 5
    % &= \sum_{x_i\in \~X_{\bar{k}_2}}\int_{\~{RQ}_{\bar{k}_2}(x_i)} \Tr \Biggl\{ \Trxx{[t]\setminus i}{\omega_{k_1} \left( \sum_{x^{[t] \setminus i }}\Lambda^{\~X^t_{\bar{k}_2}}_{x^t}   \left(\int_{\~{RQ}_{\bar{k}_2}(x^{[t] \setminus i})} \pi^{t-1}_{OT} \Big( dz^{[t]\setminus i} | x^{[t]\setminus i}\Big) \right)\right) \omega_{k_1} } \nonumber \\
    %     &\qquad\qquad \qquad \qquad \qquad \;\cdot\Delta(z_i) \pi_{OT} ( dz_i|x_i) \Biggr\} \nonumber \\
    % line 6
    &= \sum_{x_i\in \~X_{\bar{k}_2}}\int_{\~{RQ}_{\bar{k}_2}(x_i)} \Tr{ \Trxx{[t]\setminus i}{\omega_{k_1} \left( \sum_{x^{[t] \setminus i }}\Lambda^{\~X^t_{\bar{k}_2}}_{x^t}  \right) \omega_{k_1} }  \Delta(z_i) \pi_{OT} ( dz_i|x_i) } \nonumber \\
    % line 7
        &= \sum_{x\in \~X_{\bar{k}_2}}\Tr{ \hat \rho^{R_i}_{x} \left( \int_{\~{RQ}_{\bar{k}_2}(x)} \Delta(z) \pi_{OT} ( dz|x)\right) } \mu_{X_{\bar{k}_2}}(x), \label{eq:distortion-hatX}
\end{align}
 Then having the expressions \eqref{eq:distortion-Xprime} and \eqref{eq:distortion-hatX}, the distortion caused by the optimal transport block is:
\begin{align}
     % line 1
     d_t(R_{k_1}^t , \hat X_{k_1,\bar{k}_2}^t) - d_t(R_{k_1}^t , {X'}^t_{k_1,\bar{k}_2}) %\nonumber\\
     % line 2
     &= \frac{1}{t} \sum_{i=1}^t \sum_{x\in \~X_{\bar{k}_2}}\Tr{ \hat \rho^{R_i}_{x} \left(\int_{\~{RQ}_{\bar{k}_2}(x)} \Delta(z) \pi_{OT} ( dz|x) - \Delta(x)\right)} \mu_{X_{\bar{k}_2}}(x)\nonumber\\
     % line 3
     % &= \sum_{x\in \~X_{\bar{k}_2}}\Tr{ \left(\frac{1}{t} \sum_{i=1}^t  \hat \rho^{R_i}_{x}\right) \left(\int_{\~{RQ}_{\bar{k}_2}(x)} \Delta(z) \pi_{OT} ( dz|x) - \Delta(x)\right)} \mu_{X_{\bar{k}_2}}(x)\nonumber\\
     % % line 4
     % &= \sum_{x\in \~X_{\bar{k}_2}}\Tr{ \bar{\rho}^{R}_{t,x} \left(\int_{\~{RQ}_{\bar{k}_2}(x)} \Big(\Delta(z) - \Delta(x)\Big) \pi_{OT} ( dz|x)\right)} \mu_{X_{\bar{k}_2}}(x)\nonumber\\
     % % line 5
     % &= \sum_{x\in \~X_{\bar{k}_2}} \int_{\~{RQ}_{\bar{k}_2}(x)} \Tr{ \bar{\rho}^{R}_{t,x} \Big(\Delta(z) - \Delta(x)\Big) } \pi_{OT} ( dz|x) \mu_{X_{\bar{k}_2}}(x)\nonumber\\
     % line 6
     &= \frac{1}{t} \sum_{i=1}^t\int_\~X \Tr{  \hat \rho^{R_i}_{Q_{\bar{k}_2}(x)} \Big(\Delta(x) - \Delta(Q_{\bar{k}_2}(x))\Big) } \mu_{X}(dx),
    %  % line 6
    % &\leq \int_\~X \norm{\bar{\rho}^{R}_{Q_{\bar{k}_2}(x)} }_\infty \norm{\Delta(x) - \Delta(Q_{\bar{k}_2}(x))}_1 \mu_{X}(dx),\\
    % % line 7
    % &= \int_\~X  \norm{\Delta(x) - \Delta(Q_{\bar{k}_2}(x))}_1 \mu_{X}(dx). 
    \label{eq:continuous_optimal_transport_distortion}
\end{align}
% where  $\bar{\rho}^{R}_{t,x} := \frac{1}{t} \sum_{i=1}^t  \hat \rho^{R_i}_{x}$ is defined as the t-letter average PMR state for $x\in \~X$.
Finally, substituting \eqref{eq:continuous_optimal_transport_distortion} and \eqref{eq:cont-nletter-separation} into \eqref{eq:continuous-nletter-distortion-noEce} provides the following bound
\begin{align*}
    d_n(R^n, {\hat X}^n_{k_1,\bar{k}_2}|\neg E_{ce}) &\leq  d(R_{k_1} , X'_{k_1, \bar{k}_2})  + \bar\varepsilon + \varepsilon_1(k_1) + \varepsilon_2(n,k_1, \bar{k}_2),
\end{align*}
where 
\begin{align}
    \varepsilon_1(k_1) &:= (P_{k_1} + \epsilon_{cl}) d\left(R, \hat X_\text{local} | A_{k_1}=1\right), \\
    \varepsilon_2(n, \bar{k}_2) &:=   \Exx{T|\neg E_{ce}}{ \frac{1}{T} \sum_{i=1}^T \int_\~X \Tr{   \hat \rho^{R_i}_{Q_{\bar{k}_2}(x)}  \Big(\Delta(x) - \Delta(Q_{\bar{k}_2}(x))\Big) } \mu_{X}(dx)}.
\end{align}
% and $\bar{\bar \rho}^R_{t_\text{min},x} := \Exx{T|\neg E_{ce}}{\bar{\rho}^{R}_{T,x}}$ is defined as the expected average PMR state of with no error events ($\neg E_{ec}$) for $x \in \~X$, and is a function of $t_\text{min}$.

 For $\varepsilon_1(k_1)$, by letting $\epsilon_{cl} \leq P_{k_1}$ and as a result of uniform integrability, for any $\varepsilon'_1>0$ there exists a sufficiently large $k_1$ such that $\varepsilon_1(k_1) \leq \varepsilon'_1$. 
For $\varepsilon_2(n, \bar{k}_2)$, the integral inside the expectation is split into the inside and outside the clipping interval and then bounded by
\begin{align*}
    \varepsilon_2(n, \bar{k}_2) &=   \Exx{T | \neg E_{ce}}{  \frac{1}{T} \sum_{i=1}^T   \int_{-2^{k_2}}^{2^{k_2}}  \Tr{\hat \rho^{R_i}_{Q_{\bar{k}_2}(x)} \Big(\Delta(x) - \Delta(Q_{\bar{k}_2}(x))\Big)} \mu_X(dx)} \\
     &+ \Exx{T | \neg E_{ce}}{  \frac{1}{T} \sum_{i=1}^T\int_{\~X \setminus [-2^{k_2},2^{k_2}]} \Tr{\hat \rho^{R_i}_{Q_{\bar{k}_2}(x)} \Big(\Delta(x) - \Delta(Q_{\bar{k}_2}(x))\Big)} \mu_X(dx)} \leq \varepsilon_{21}(k_1) + \varepsilon_{22}(n,\bar{k}_2)
\end{align*}
where
\begin{align}
    % line 2
    \varepsilon_{21}(k_1) &:=   \int_{-2^{k_2}}^{2^{k_2}}  \norm{\Delta(x) - \Delta(Q_{\bar{k}_2}(x))}_1 \mu_{X}(dx),   \\
    % line 3
    \varepsilon_{22}(n,\bar{k}_2) &:=  \Exx{T | \neg E_{ce}}{  \frac{1}{T} \sum_{i=1}^T \int_{\~X \setminus [-2^{k_2},2^{k_2}]} \Tr{\hat \rho^{R_i}_{Q_{\bar{k}_2}(x)} \Big(\Delta(x) - \Delta(Q_{\bar{k}_2}(x))\Big)} \mu_X(dx) } .
\end{align}
%%%%%%%%%%%%%%%%%%%%%%%%%%%%%%%%%%%%%%%%%%%%%%%%%
% Note that as $ \lim_{k_1 \to \infty} P_{k_1} = 0$, the first term decays to zero by letting $\epsilon_{cl} \leq P_{k_1}$, as a result of uniform integrability of the distortion observable.
%  For the second term, note that The second term also decays to zero as follows:
% \begin{align}
%     % line 1
%     &\lim_{k_2 \to \infty} \lim_{k_2' \to \infty} \int_\~X \Tr{ \bar{\bar \rho}^R_{t_\text{min},Q_{\bar{k}_2}(x)} \Big(\Delta(x) - \Delta(Q_{\bar{k}_2}(x))\Big) } \mu_{X}(dx) \nonumber\\
%     % line 2
%     &\leq \lim_{k_2 \to \infty} \lim_{k_2' \to \infty} \int_{-2^{k_2}}^{2^{k_2}}  \norm{\Delta(x) - \Delta(Q_{\bar{k}_2}(x))}_1 \mu_{X}(dx)  \nonumber \\
%     % line 3
%     &\quad + \lim_{k_2 \to \infty}  \int_{\~X \setminus [-2^{k_2},2^{k_2}]}   \Tr{ \bar{\bar \rho}^R_{t_\text{min},Q_{\bar{k}_2}(x)}  \Big(\Delta(x) - \Delta(Q_{\bar{k}_2}(x))\Big) } \mu_{X}(dx)=0,
% \end{align}

 To obtain $\varepsilon_{21}(k_1)$ we used Holder's inequality on the first term and the fact that the PMR state $\hat \rho_x^{R_i}$ for $x\in \~X_{\bar{k}_2}$ is a state of finite-dimensional space. The expectation term is excluded because the resulting term is not a function of $T$.  Then the first term above can be made arbitrarily small to any $\varepsilon'_{21}$ by choosing a large enough $k_1$, due to the continuity of the $\Delta(x)$ operator ($\varepsilon_{21}(k_1) \leq \varepsilon'_{21}$).
 
  In the second term $\varepsilon_{22}(n, \bar{k}_2)$, for the inner integral, by the definition of the uniform integrability, regardless of the index $i$, for an arbitrarily small value $\varepsilon'_{22} > 0 $ there exists a large enough $k_2$ such that 
 \begin{align}
     \int_{\~X \setminus [-2^{k_2},2^{k_2}]} \Tr{\hat \rho^{R_i}_{Q_{\bar{k}_2}(x)} \Big(\Delta(x) - \Delta(Q_{\bar{k}_2}(x))\Big)} \mu_X(dx) \leq \varepsilon'_{22} \quad \text{for } i = 1, \cdots, T. 
 \end{align}
 Then for any value of $n$ we have that
 % \begin{align}
     $\varepsilon_{22}(n,\bar{k}_2) \leq \Exx{T | \neg E_{ce}}{\frac{1}{T} \sum_{i=1}^T \varepsilon'_{22}} = \varepsilon'_{22}$.
 % \end{align}
 
 % Note that the only term that is a function of $t_\text{min}$ is the expected average PMR state. But again from the uniform integrability of the system, we know that for any value of $k_2$ there exists an upperbound value $\varepsilon''(k_2)$ independent of $t_\text{min}$, where
 % \begin{align}
 %     \int_{\~X \setminus [-2^{k_2},2^{k_2}]}   \Tr{ \bar{\bar \rho}^R_{t_\text{min},Q_{\bar{k}_2}(x)}  \Big(\Delta(x) - \Delta(Q_{\bar{k}_2}(x))\Big) } \mu_{X}(dx) \leq \varepsilon''(k_2).
 % \end{align}

Finally putting it all together gives the single-letter upper bound
\begin{align}
d_n(R^n, {\hat X}^n_{k_1,\bar{k}_2}|\neg E_{ce}) &\leq  d(R_{k_1} , X'_{k_1, \bar{k}_2})  + \bar\varepsilon +  \varepsilon'_{21} + \varepsilon'_{22}.
\end{align}
where $\bar\varepsilon>0$ can be made arbitrarily small given sufficiently large $n$ and  $\varepsilon'_{21}, \varepsilon'_{22}$ can be made arbitrarily small given sufficiently large $k_1, k_2$.

\subsection{Uniform Integrability of Gaussian Measurement System} \label{appA-uniforminteg}

The distortion  function \eqref{eq:distortion_continuous_firstdefinition} can be expressed in a different form  as follows:
\begin{align}
    % line 1
    d(\rho_A, \Lambda) %&=  \frac{1}{2s} \sum_{i=1}^{2s} \Tr \Bigg\{\left( \sum_z {(\rho_z^\~T)}^R \otimes \pi(z ) \ketbra{z}{z} \right)\bigg(R_i^\~T \otimes I_\~H  - I_{\~H^*} \otimes R_i\bigg)^2\Bigg\} \\
    % % line 3
    % &= \frac{1}{2s} \sum_{i=1}^{2s} \sum_z \pi(z) \Big[  \Tr{\rho_z R_i^2} + \Tr{\rho_z} \braket{z|R_i^2}{z} - 2 \Tr{\rho_z R_i} \braket{z}{R_i|z}  \Big]\\
    % line 4
    &= \frac{1}{2s} \sum_{i=1}^{2s} \Exx{Z}{  \Tr{\rho_Z R_i^2} + \Tr{\rho_Z} Z_i^2 - 2 \Tr{\rho_Z R_i} Z_i  }:= d(\rho_Z, \pi_Z).
\end{align} 
\begin{lemma}
Considering a Gaussian quantum Optimal transport measurement system comprised of a Gaussian source state $\rho$ and a Gaussian output distribution $\pi_Z$, with the  PMR ensemble $\{\{\rho_z\}_{z\in \~Z}, \pi_Z\}$. For any set of projection operators $\Pi_Z$ such that $\Exx{Z}{\rho_Z \Pi_Z} \leq \delta$, we have that
% \begin{align}
$d(\Pi_Z\rho_Z \Pi_Z, \pi_Z) \leq f(\delta),$
% \end{align}
where $f(\delta) \downarrow 0$ as $\delta \downarrow 0$.
\end{lemma}
\begin{proof}
  The proof follows a nested use of the Cauchy-Schwartz inequality.  The distortion of the unnormalized projected state is expanded as
  \begin{align} \label{eq:distortion-terms-uniformproof}
      d(\Pi_Z\rho_Z \Pi_Z, \pi_Z) = \frac{1}{2s} \sum_{i=1}^{2s} \Exx{Z}{  \Tr{\Pi_Z\rho_Z \Pi_Z R_i^2} + \Tr{\Pi_Z\rho_Z} Z_i^2 - 2 \Tr{\Pi_Z\rho_Z \Pi_Z R_i} Z_i  }.
  \end{align}
  We further examine each term separately.  The operation $\langle M, N \rangle = \Tr{M^\dagger N}$ is the Hilbert-Schmidt inner product, then it satisfies the Cauchy-Schwartz inequality. Thus, for the first term of  \eqref{eq:distortion-terms-uniformproof} we have,
  \begin{align*}
      \Tr{\Pi_Z\rho_Z \Pi_Z R_i^2} &= \Tr{\sqrt{\rho_Z} \Pi_Z R_i^2 \Pi_Z \sqrt{\rho_Z}} 
      \leq \sqrt{\Tr{\rho_Z \Pi_Z}}\sqrt{ \Tr{\Pi_Z \rho_Z \Pi_Z R_i^4}}.
  \end{align*}
  Then we use another round of Cauchy-Schwartz inequality for its classical expectation:
  \begin{align}
      \Exx{Z}{\sqrt{\Tr{\rho_Z \Pi_Z}}\sqrt{ \Tr{\Pi_Z \rho_Z \Pi_Z R_i^4}}} &\leq \sqrt{\Exx{Z}{\Tr{\rho_Z \Pi_Z}}} \sqrt{\Exx{Z}{\Tr{\Pi_Z \rho_Z \Pi_Z R_i^4}}}  \nonumber \\
      &\leq \sqrt{\delta} \sqrt{\Exx{Z}{\Tr{\Pi_Z \rho_Z \Pi_Z R_i^4}}} := f_{1,i}(\delta).
  \end{align}
  By the Gaussian Optimality Theorem \ref{th:gaussian_optimality_theorem}, the optimal PMR states $\rho_Z$ corresponding to the rate-limited optimal transport measurement are all Gaussian quantum states. Then in the above inequality, $f_{1,i}(\delta)\downarrow 0$ as $\delta \downarrow 0$ because the fourth moment of a Gaussian quantum source state $\rho_Z$ is finite, and thus, the moments of its projection by $\Pi_Z$ will still remain finite. 
  
  The second term of \eqref{eq:distortion-terms-uniformproof} is upper-bounded by the Cauchy-Schwartz inequality as:
  \begin{align}
      \Exx{Z}{Z_i^2 \Tr{\Pi_Z \rho_Z}} &\leq \sqrt{\Exx{Z}{Z_i^4}} \sqrt{\Exx{Z}{\left(\Tr{\Pi_Z \rho_Z}\right)^2}} \nonumber \\
      &\leq \sqrt{\Exx{Z}{Z_i^4}} \sqrt{\Exx{Z}{\Tr{\Pi_Z \rho_Z}}} 
      \leq \sqrt{\Exx{Z}{Z_i^4}} \sqrt{\delta} := f_{2,i}(\delta)
  \end{align}
  where we used the fact that the fourth moment of the output Gaussian distribution $\pi_Z$ is finite. The third term of \eqref{eq:distortion-terms-uniformproof} follows a similar approach to the above steps:
  \begin{align}
      - \Exx{Z}{\Tr{\Pi_Z\rho_Z \Pi_Z R_i} Z_i} &\leq \sqrt{\Ex{Z_i^2}}\sqrt{\Exx{Z}{\left(\Tr{\Pi_Z\rho_Z \Pi_Z R_i}\right)^2}} \nonumber \\
      &\leq \sqrt{\Ex{Z_i^2}}\sqrt{\Exx{Z}{\Tr{\rho_Z \Pi_Z}\Tr{\Pi_Z\rho_Z\Pi_Z R_i^2}}} \nonumber \\
      &\leq \sqrt{\Ex{Z_i^2}}
      \sqrt[4]{
        \Exx{Z}{\left(\Tr{\rho_Z \Pi_Z}\right)^2}\cdot \Exx{Z}{\left(\Tr{\Pi_Z\rho_Z\Pi_Z R_i^2}\right)^2}
        } \nonumber \\
    %&\leq \sqrt{\Ex{Z_i^2}}
     % \sqrt[4]{
      %  \Exx{Z}{\Tr{\rho_Z \Pi_Z}}\cdot \Exx{Z}{\left(\Tr{\Pi_Z\rho_Z\Pi_Z R_i^2}\right)^2}
      %  } \nonumber\\
        &\leq \sqrt[4]{\delta} \sqrt{\Ex{Z_i^2}}
      \sqrt[4]{\Exx{Z}{\left(\Tr{\Pi_Z\rho_Z\Pi_Z R_i^2}\right)^2}} := f_{3,i}(\delta)
  \end{align}
  Putting all the bounds above together shows that the projected distortion in \eqref{eq:distortion-terms-uniformproof} is upper-bounded by a continuous function $f(\delta)$ as follows, where $f(\delta)  \downarrow 0 $ as $\delta \downarrow 0$. Hence by definition, the Gaussian measurement system is Uniformly integrable, $d(\Pi_z\rho_z \Pi_z, \pi_z) \leq  \frac{1}{2s} \sum_{i=1}^{2s} \left(f_{1,i}(\delta) + f_{2,i}(\delta) + 2f_{3,i}(\delta)\right) := f(\delta)$.
  % \begin{align}
  %     d(\Pi_z\rho_z \Pi_z, \pi_z) \leq  \frac{1}{2s} \sum_{i=1}^{2s} \left(f_{1,i}(\delta) + f_{2,i}(\delta) + 2f_{3,i}(\delta)\right) := f(\delta). \nonumber
  % \end{align} 
\end{proof}

\section{Proof of Theorem \ref{th:qubit_ratedist_theorem} } \label{appB}

%%%%%%%%%%%%%%%%%%%%%%%%%%%%%%%%%%%%%%%%%%%%%%%%%%%%%%%%%%
%%%%%%%  Appendix Section:  proof of Qubit Rate-Dist %%%%%
%%%%%%%%%%%%%%%%%%%%%%%%%%%%%%%%%%%%%%%%%%%%%%%%%%%%%%%%%%
\ifmycomments
{\color{blue}
We simplify the distortion term in the following form
\begin{align}
    \bra{0}\rho M_0^T\rho\ket{0} +  \bra{1}\rho M_1^T\rho\ket{1} \geq 1 - D\\
    \braket{1}{\rho^2|1} + \bra{0}\rho M_0^T\rho\ket{0} -  \bra{1}\rho M_0^T\rho\ket{1} \geq 1 - D\\
    \Tr{\rho M_0^T \rho \begin{bmatrix}
        1 & 0\\0&-1
    \end{bmatrix} } \geq 1 - D -  \braket{1}{\rho^2|1}
\end{align}
Note that the above inequality has only real terms therefore we can transpose the LHS. Also note that as the transpose is w.r.t. the eigenbasis of the input state $\rho$, then $\rho^\~T = \rho$ as it is diagonal w.r.t. its own basis. Then we have
\begin{align}
    \Tr{ M_0 \rho \begin{bmatrix}
        1 & 0\\0&-1
    \end{bmatrix}^T \rho} \geq 1 - D -  \braket{1}{\rho^2|1}
\end{align}
Then just like before, by setting $N:= \sqrt{\rho} M_0 \sqrt{\rho}$ and 
\begin{align}
    G:= \sqrt\rho \begin{bmatrix}
        1 & 0\\0&-1
    \end{bmatrix}^T \sqrt\rho
\end{align}
we will have the distortion inequality as 
\begin{align}
        \Tr{N G^T} \geq 1 - D -  \braket{1}{\rho^2|1}
\end{align}
Therefore, the only change is in the transpose above the matrix $G$. If we can show that the matrix $G^T$ can still be shown in the format:
\begin{align}
    G^T:= \begin{bmatrix}
                g_1& g_2 \rho_2/|\rho_2|\\ g_2 \rho_2^*/|\rho_2| & g_3
            \end{bmatrix}
\end{align}
then the optimization is already solved for that.
}
\fi

We start with the optimization problem in \eqref{eq:qubit_ratedist_problem} 
and change the variable of optimization to $N = \sqrt{\rho} M_0 \sqrt{\rho}$, which is the unnormalized conditional PMR state given outcome zero. The optimization problem reduces to the following form
\begin{align} \label{eq:optimization_problem_N}
    \min_{N} {\quad \Tr{N \ln(N/q_0)} + \Tr{(\rho - N) \ln(\frac{\rho - N}{1 - q_0})}},\\
    %line 2
    \text{such that,} \; \bra{0}\sqrt\rho N\sqrt\rho\ket{0} +  \bra{1}\sqrt\rho (\rho - N) \sqrt\rho\ket{1}\geq 1- D, \label{eq:example_distortion_const_notsimplified}\\
    \Tr{N}= q_0, \quad 0 \preceq N \preceq \rho. \label{eq:infinitrandom_const2}
\end{align}
% where $\rho $ is the input state and $ 0 \leq q_0 \leq 1$ is the zero-output probability, which are the given parameters of the problem. 
This is a convex optimization problem, as the objective function comprises negative entropy functions that are concave, and the constraint is linear. The distortion constraint \eqref{eq:example_distortion_const_notsimplified} is further simplified to
\begin{align} \label{eq:distortion_constraint_NG}
    \Tr{N G} \geq 1 - D - \braket{1|\rho^2}{1},
    \quad G := \sqrt{\rho}\begin{bmatrix}1 & 0\\0 & -1\end{bmatrix} \sqrt{\rho}.
\end{align}

% \subsection{Zero-Crossing Point}
We first obtain the zero-crossing point $D_{R_0}$ which is the threshold that any distortion value $D \geq D_{R_0}$, is achieved by zero communication rate ($R=0$). Setting the rate to zero $I(R;X) = H( \rho) - q_0 H(\hat\rho_0) - q_1 H(\hat\rho_1) = 0$ gives $\hat\rho_0 = \hat\rho_1  = \rho$. 
The measurement POVMs corresponding to these PMR states are $M_0 = q_0 I, \; M_1 = (1-q_0) I$.
% The zero communication rate means the source state and destination distribution are independent, which is also observed from the fact that $M_0, M_1$ both apply the same identity operator regardless of the outcome. 
Substituting $\hat \rho_0= \hat \rho_1=\rho$ in the distortion constraint gives the zero-crossing point
      $D_{R_0} := 1 - q_0 \braket{0}{\rho^2 |0} -  (1 - q_0) \braket{1}{\rho^2 |1}$.
% Thus, for all values of entanglement fidelity distortion in the range $D \geq D_{R_0}$, the output and reference state is independent and the communication rate is zero.

% \subsection{Non-Zero Rate Region}
Next, in the non-zero rate region $D_{OT} \leq D \leq D_{R_0}$, the distortion constraint is active. Note that $D_{OT}$ is the optimal transportation cost which is the minimum feasible distortion possible for the system.

We first show here that the optimal solution of $N$ is of the form presented in \eqref{eq:qubit_evaluation_Nopt_matrix}. Recall that the convex problem \eqref{eq:optimization_problem_N} is on the domain of positive semi-definite Hermitian matrices $N \in \~S^n_+$, and has a fixed trace $\Tr{N} = q_0$. So it can be initially considered as the following matrix form:
\begin{align} \label{eq:initialN_noff_format}
     N = \begin{bmatrix}n& n_\text{off} \\ n_\text{off}^*& q_0 - n \end{bmatrix},
\end{align}
where $n \in \mathbb{R}^+$ and $n_\text{off}\in \mathbb{C}$. Further, the PSD constraints of \eqref{eq:infinitrandom_const2} reduce to $\det(N)\geq 0$, and $\det(\rho - N) \geq 0$, and $0\leq \Tr{N} \leq 1$. Thus we have
\begin{align} 
    |n_\text{off}|^2 \leq n(q_0 - n) ,\qquad &|\rho_2 - n_\text{off} |^2 \leq (\rho_1 - n) (1 - \rho_1 - q_0 + n ) \label{eq:noff_constraint}.
\end{align}
Next, by examining the objective function \eqref{eq:optimization_problem_N} and the constraints, we observe a symmetry in $n_\text{off}$ values along the real and imaginary parts of the complex plain. 
Note that the first term in \eqref{eq:optimization_problem_N} is a function of eigenvalues of matrix $N$ which only depends on $|n_\text{off}|$.
Similarly, the second term is only a function of $|\rho_2 - n_\text{off}|$. 
These two expressions can be illustrated as two circles in the complex plane of $n_\text{off}$, which are symmetrical w.r.t. the direction of the vector $\rho_2$ in this complex plane. 
Moreover, by representing $G$ in the following matrix elements format:
\begin{align*}
G:= \begin{bmatrix}
    g_1 & \frac{g_2}{|\rho_2|}\rho_2 \\ \frac{g_2}{|\rho_2|} \rho_2^* & g_3
    \end{bmatrix}    
= \frac{1}{1+2k}
    \begin{bmatrix}
        (\rho_1+k)^2 - |\rho_2|^2 
        & (2\rho_1 -1)\rho_2 \\b
        (2\rho_1 -1)\rho_2^*
        & |\rho_2|^2 - (1 - \rho_1 +k)^2
    \end{bmatrix},
\end{align*}
with $k:= \sqrt{\det{\rho}}$, and \eqref{eq:initialN_noff_format},
we can rewrite the LHS of distortion constraint \eqref{eq:distortion_constraint_NG} as 
\begin{align*}
    \Tr{NG} = (g_1 - g_3)n + \frac{2g_2}{|\rho_2|} \Re{n_\text{off} \,\rho^*_2}+q_0 g_3 \geq  1 - D - \braket{1|\rho^2}{1}.
\end{align*}
 By expanding the term $\Re{n_\text{off} \,\rho^*_2}$ in above expression, the distortion constraint becomes
\begin{align*}
    \Re{n_\text{off}}\Re{\rho_2} + \Im{n_\text{off}}\Im{\rho_2} + f(n,q_0,\rho) \geq 0,
\end{align*}
where $f(n,q_0,\rho)$ is a function comprised of the remaining terms. This is a half-plane in the complex plane of $n_\text{off}$, which is orthogonal to the symmetry line $\rho_2$. Therefore, the objective function and all the constraints are symmetric w.r.t. the line $\rho_2$ in the complex plane. In view of the fact that the problem is a convex program, the solution must occur on the line of symmetry. This means the optimal solution has the shape of the matrix described in \eqref{eq:qubit_evaluation_Nopt_matrix},
where $n, s \in \mathbb{R}$ are the variables of optimization. 
\ifmycomments
{\color{blue} Also, the feasibility constraints \eqref{eq:noff_constraint} reduce to the following constraints on real variables $n,s$,
\begin{align} 
    \det{N} &= n(q_0 - n) - s^2 > 0 \label{eq:detNconst},\\
    \det{\rho - N} &= (\rho_1 - n)(1 - \rho_1 - q_0 + n) - (|\rho_2| - s)^2 > 0. \label{eq:detdconst}
\end{align}
}
\fi
% Then the distortion constraint \eqref{eq:distortion_constraint_NG} is further expanded in the form
% \begin{align}
%      \Big(g_1 - g_3 \Big) n + \Big(2 g_2\Big)s  + \Big(q_0 g_3 + \braket{1|\rho^2}{1} - 1 + D \Big) =0, \label{eq:distexpandedconts}
% \end{align}
Therefore, by rewriting the distortion constraint and the objective function as functions of $n$ and $s$ we arrive at the following optimization problem:
\begin{align*}
    \min_{n,s} &\sum_{i=1,2} \lambda_{N_i}(n,s) \ln(\lambda_{N_i}(n,s)/q_0) + \sum_{i=1,2} \lambda_{d_i}(n,s) \ln(\lambda_{d_i}(n,s)/(1-q_0))\\
    \text{such that,} \;
    &\Big(g_1 - g_3 \Big) n + \Big(2 g_2\Big)s  + \Big(q_0 g_3 + \braket{1|\rho^2}{1} - 1 + D \Big) \geq 0 
\end{align*}
where $n,s \in \mathbb{R}$. The terms $\lambda_{N_1, N_2}(n,s) =   \frac{q_0}{2} \pm E_1(n,s)$ and $\lambda_{d_1, d_2}(n,s) =  \frac{1 - q_0}{2}\pm E_2(n,s)$  
% $\lambda_{N_i}$ and $\lambda_{d_i}$ 
are the eigenvalues of $N$ and $\rho - N$ respectively, with $E_1(n,s)$ and $E_2(n,s)$ defined in \eqref{eq:E1}. 
% \begin{align*}
%     \lambda_{N_1, N_2} =   \frac{q_0}{2} \pm E_1(n,s),
%     \quad \quad \lambda_{d_1, d_2} =  \frac{1 - q_0}{2}\pm E_2(n,s).
% \end{align*}

% \begin{align}
%     E_1(n,s) &:= \sqrt{\left(n-\frac{q_0}{2}\right)^2 +s^2},\\
%     E_2(n,s) &:= \sqrt{\left(n-\rho_1 + \frac{1-q_0}{2}\right)^2 +(s - |\rho_2|)^2}.
% \end{align}
% \subsection{The Transcendental Equation of Optimal Solution}
We solve by substituting the linear constraint into the objective function, then taking the derivative w.r.t. $n $ and equating it to zero.
The final solution becomes in the form of the following transcendental system of equations of $n,s$ 
 which provides the optimal values $n_{opt}$ and $s_{opt}$:
\begin{align*}
    \frac{-as + b(n - q_0/2)}{E_1}\ln{\frac{q_0/2 + E_1}{q_0/2 - E_1}}
    +  \frac{-a (s - |\rho_2|) + b(n - \rho_1 + \frac{1-q_0}{2})}{E_2} \ln{\frac{\frac{1 - q_0}{2} + E_2}{\frac{1 - q_0}{2} - E_2}}=0,\\
     \bigg(  1 - \frac{4 |\rho_2|^2}{1+2k}  \bigg) n
    +\bigg(  \frac{2|\rho_2|(2\rho_1 -1)}{1+2k}  \bigg)  s 
    +\bigg(  q_0\left(  \rho_1 - 1 + \frac{2|\rho_2|^2}{1+2k}  \right) + \braket{1|\rho^2}{1} - 1 +D  \bigg) =0.
\end{align*}

\ifmycomments
\myqu{The way this optimal N depends on the parameters seems inefficient to me. I think it would be a better way of defining the parameters of input. Maybe if we define the $\rho$ matrix as $\rho = p_0 \nu_0 + p_1 \nu_1$ where $\nu_0, \nu_1$ are orthogonal, then it would be more meaningful. }\myans{Probably our parameters are also meaningful because the distortion constraint is finding the inner product based on $\braket{0|.}{0}$ and $\braket{1|.}{1}$. We also saw that the diagonal quantum state reduces to the classical optimal transport solution. So these parameters also make sense.}
\fi

\section{Proof of Theorem \ref{th:qubit_optimal_transport}} \label{appC}

To find the optimal transportation cost $D_{OT}$ of the Qubit system, we assume having an unlimited available rate and find the minimum possible distortion. In this case, the problem reduces to
\begin{align*}
    D_{OT} :=  \min_{M_0, M_1}\; &1- \sum_x \braket{x|\rho M_x \rho}{x},\quad \text{such that: } \; & \Tr{M_0 \rho} = q_0, \quad M_0 + M_1 = I,\quad M_0, M_1 \geq 0.
\end{align*}
Note that we used the same argument as in Appendix A to remove the transpose operator from the formulations. Then again a change of variable $N:= \sqrt{\rho } M_0 \sqrt{\rho}$ reduces the problem to the following semi-definite programming:
% \begin{align*}
    $\min_N \quad  1 - \braket{1|\rho^2}{1} - \Tr{NG},$ such that $ \; \Tr{N} = q_0$, and $ 0 \preceq N \preceq \rho.$
% \end{align*}

Similar to the rate-distortion problem in Appendix \ref{appB}, the symmetry requires the optimal solution $N_{OT}$ to be in the matrix form described in \eqref{eq:qubit_evaluation_Nopt_matrix}. With this assumption, the problem reduces to the scalar optimization below:
\begin{align*}
    \min_{n,s} f_0(n,s) &:= -(g_1 - g_3) n - (2g_2)s + 1- q_0g_3 - \braket{1|\rho^2}{1},\\
\text{such that, }\;f_1(n,s)&:=\left(n - \frac{q_0}{2}\right)^2 + s^2 - \left(\frac{q_0}{2}\right)^2 \leq 0, \\
    f_2(n,s)&:=\left(n - \left(\rho_1 - \frac{1-q_0}{2}\right)\right)^2 + (s - |\rho_2|)^2 - \left(\frac{1-q_0}{2}\right)^2  \leq  0.
\end{align*}
This is a quadratic-constrained linear programming which is a convex problem. The Lagrangian function is 
    $\~L(n,s, \lambda_1,\lambda_2)= f_0(n,s) + \lambda_1 f_1(n,s) + \lambda_2 f_2(n,s)$, with $\lambda_{1,2} \geq 0$.
Then by taking the partial derivatives with respect to $n,s$ and equating to zero, we obtain the optimal variables as a function of $\lambda_i$ Lagrange multipliers,
\ifmycomments
    \mynote{The parameter $a$ is always non-negative ($a\geq 0$). It can be proved simply by writing the terms and multiplying.}
\fi
\begin{align} \label{eq:qubit_Dot_nop_sop_eqsinproof}
    n_{opt }(\lambda_1,\lambda_2) = \frac{a + q_0 \lambda_1 + 2(\rho_1 - \frac{1-q_0}{2})\lambda_2}{2(\lambda_1+ \lambda_2)},
    \quad s_{opt}(\lambda_1,\lambda_2) = \frac{b+ 2\lambda_2|\rho_2|}{2(\lambda_1+ \lambda_2)}.
\end{align}
Then by substituting the above expressions in the complementary slackness conditions
\begin{align*}
    \lambda_1 f_1\Big(n_{opt}(\lambda_1,\lambda_2),\, s_{opt}(\lambda_1,\lambda_2)\Big)=0, \quad \lambda_2 f_2\Big(n_{opt}(\lambda_1,\lambda_2),\, s_{opt}(\lambda_1,\lambda_2)\Big)=0,\quad \lambda_{1,2}\geq 0,
\end{align*}
the following solutions are the results of these separate scenarios.
\begin{enumerate}
    % first scenario
    \item Minimum happens at $f_2$ circle ($\lambda_1 =0, \; \lambda_2 \neq 0$):
    Because $\lambda_2 \neq 0$, then $f_2(n_{opt},s_{opt})=0$. This results in $\lambda_2 = \frac{\sqrt{a^2+b^2}}{1-q_0}$. Substituting $\lambda_{1,2}$ into \eqref{eq:qubit_Dot_nop_sop_eqsinproof} and then substituting the obtained $n_{opt}$ and $s_{opt}$ into objective function gives,
        \begin{align*}
            % s_{opt} &=  \frac{b}{\sqrt{1 - 4|\rho_2|^2}}\frac{1-q_0}{2} + |\rho_2|, \quad n_{opt} = \left(\frac{a}{\sqrt{1 - 4|\rho_2|^2}}-1\right)\frac{1-q_0}{2}+ \rho_1,\\
            D_{OT} %&= 1 - q_0 g_3 - \braket{1|\rho^2}{1} - \frac{1-q_0}{2}\sqrt{a^2 + b^2} - a(\rho_1 - \frac{1-q_0}{2}) - b|\rho_2|\\
            &= q_0(1-\rho_1) + \det(\rho) + \frac{1-q_0}{2}\left(1 - \sqrt{1 - 4|\rho_2|^2}\right).
        \end{align*}
        Also, substituting the $n_{opt}$ and $s_{opt}$ values into  $f_1(n_{opt}, s_{opt})\leq 0$ provides the conditions required for this scenario:
        % \begin{align*}
            $\frac{1-q_0}{2} + \frac{1-q_0}{\sqrt{a^2 +b^2}}\left(    a(\rho_1 - \frac{1}{2}) + b|\rho_2|    \right) \leq \det(\rho).$
        % \end{align*}

    \item Minimum happens at $f_1$ circle ($\lambda_1 \neq 0, \lambda_2 =0$):
        Because $\lambda_1 \neq 0$, then $f_1(n_{opt},s_{opt})=0$. This results in
            $\lambda_1 = \frac{\sqrt{a^2+b^2}}{q_0}$.
        Again, substituting these $\lambda_{1,2}$  into \eqref{eq:qubit_Dot_nop_sop_eqsinproof} gives,
        \begin{align*}
            % s_{opt} &= \frac{b}{\sqrt{1 - 4|\rho_2|^2}} \frac{q_0}{2}, \quad n_{opt} =  \left(\frac{a}{\sqrt{1 - 4|\rho_2|^2}} + 1\right) \frac{q_0}{2}, \\
            D_{OT} %&= 1 - q_0 g_3 - \braket{1|\rho^2}{1} - \frac{q_0}{2}\sqrt{a^2+b^2} - \frac{a q_0}{2}\\
            &= (1-q_0)\rho_1 + \det(\rho) + \frac{q_0}{2}\left( 1 - \sqrt{1 - 4|\rho_2|^2}\right).
        \end{align*}
        Also, substituting the corresponding $n_{opt}$ and $s_{opt}$ values into $f_2(n_{opt}, s_{opt})\leq 0$ provides the conditions required for this scenario:
        % \begin{align*}
            $\frac{q_0}{2} - \frac{q_0}{\sqrt{a^2 +b^2}}\left(    a(\rho_1 - \frac{1}{2}) + b|\rho_2|    \right) \leq \det(\rho)$.
        % \end{align*}
    
    \item Minimum happens at intersection of circles ($\lambda_1 \neq 0, \lambda_2 \neq 0$):
        If neither of the conditions of the above scenarios is satisfied, then the result happens at the intersection of two circles $f_1(n,s)=0$ and $f_2(n,s)=0$,
        % \begin{align*}
        %     \left(n - \frac{q_0}{2}\right)^2 + s^2 = \left(\frac{q_0}{2}\right)^2,  
        %     &&\left(n - \left(\rho_1 - \frac{1-q_0}{2}\right)\right)^2 + (s - |\rho_2|)^2  =  \left(\frac{1-q_0}{2}\right)^2,
        % \end{align*}
         which results in the following expressions under separate conditions:
        \begin{itemize}
        \item If $a\geq b$, then 
            % \begin{align*}
                $s_{opt} = \frac{2B +q_0 A + A \sqrt{\Delta}}{2(1+ A^2)}$, and $n_{opt} = \frac{q_0 - 2A B + \sqrt{\Delta}}{2(1+ A^2)}$.
            % \end{align*}
        \item If $a<b$, then
            % \begin{align*}
                $s_{opt} = \frac{2B +q_0 A - A \sqrt{\Delta}}{2(1+ A^2)}$, and $
                n_{opt} = \frac{q_0 - 2A B - \sqrt{\Delta}}{2(1+ A^2)}$.
            % \end{align*}
        \end{itemize}
        where
        % \begin{align*}
            $A:= \frac{1-2\rho_1}{2|\rho_2|}, \quad
            B:= \frac{\rho_1 q_0 - \det(\rho)}{2|\rho_2|}, \quad 
            \Delta := q_0^2 - 4 B^2 - 4q_0AB$.
        % \end{align*}
        
        Then by substituting the above optimal values in $f_0(n_{opt},s_{opt})$ gives the optimal transportation cost 
        % \begin{align*}
           $ D_{OT} = 1 - q_0 g_3 - \braket{1|\rho^2}{1} - a n_{opt} - b s_{opt}$.
        % \end{align*}
\end{enumerate}

\section{Proof of Theorem \ref{th:gaussian_onemode}} \label{appD}
Recall from subsection  \ref{subsec:overview_gaussian_systems} that for the one-mode system, $\abs{\Delta^{-1} \Sigma_N} = \sqrt{\det{\Sigma_N}} I_{2\times 2}$. Then in the objective function \eqref{eq:gaussian_optimization_objective}, we can simplify $\frac{1}{2} \text{Sp } g\left(\abs{\Delta^{-1} \Sigma_N} - \frac{I}{2}\right)  = g\left(\sqrt{\det{\Sigma_N}} - 1/2\right)$. 
% In this case, the objective function reduces to 
% \begin{align} 
%     R(D; \infty,  \~{QN}(0, \Sigma_\rho) || \~N(0, \Sigma_Z)) := \min_{\Sigma_N} & \ g\left(\sqrt{\det{\Sigma_\rho}} - 1/2\right) - g\left(\sqrt{\det{\Sigma_N}} - 1/2\right),\\
%     \text{such that: }&\; \sp{ (\Sigma_Z^{1/2} (\Sigma_\rho - \Sigma_N) \Sigma_Z^{1/2})^{1/2}} \geq D_\text{max} - D\\
%     & \pm \frac{i}{2} \Delta \leq \Sigma_N \leq \Sigma_\rho .
% \end{align}
We further change the variable of optimization to
$X:= \left(  \Sigma_Z^{1/2} (\Sigma_\rho - \Sigma_N)  \Sigma_Z^{1/2}   \right)^{1/2}$.
This gives the noise covariance $\Sigma_N = \Sigma_\rho - \Sigma_Z^{-1/2} X^2 \Sigma_Z^{-1/2}$. Then the optimization problem for the one-mode system is equivalent to
\begin{align}
    \min_X \; &  \left[\Theta(X) := - g\left(  \sqrt{ \frac{\det{\Sigma - X^2}}{\det{\Sigma_Z}}}    - \frac{1}{2}  \right) \right], \label{eq:Objfunc-onemode}\\
    \text{such that: }& c \leq \sp{X},\\
    &0\leq X  \leq \left(  \Sigma_Z^{1/2} \left(\Sigma_\rho \pm \frac{i}{2} \Delta \right)  \Sigma_Z^{1/2}   \right)^{1/2}, \label{eq:X_constraint} 
\end{align}
where $\Sigma := \Sigma_Z^{1/2}\Sigma_\rho \Sigma_Z^{1/2}$ and  $c = D_\text{max} - D$. 

Note that the inequality $ \pm \frac{i}{2} \Delta \leq \Sigma_N$ implies that $0 \leq \Sigma_N $. This in turn implies that $X^2 \leq \Sigma$.
\ifmycomments
{\color{blue}
To prove this it just suffices to write the eigenvalues for the $2\times 2 $ matrix of $\Sigma_N$ and $\Sigma_{N} \pm \frac{i}{2}\Delta$. Having a $\Sigma_N = \begin{bmatrix}
    a & c\\ c & b
\end{bmatrix}$, its eigenvalues are
\begin{align}
    \lambda_{1,2} = \frac{1}{2}\left((a+b) \pm \sqrt{(a+b)^2 + 4(c^2 - ab)}\right).
\end{align}
Then the eigenvalues of the $\Sigma_N \pm \frac{i}{2}\Delta$ are 
\begin{align}
    \lambda'_{1,2} = \frac{1}{2}\left((a+b) \pm \sqrt{(a+b)^2 + 4(c^2 + \frac{1}{4} - ab)}\right).
\end{align}
It is obvious that when $\lambda'_2 \geq 0$, then it means that $\lambda_2\geq 0$ as well. 
BTW, it could just be said in a direct way that because $\Sigma_N$ is the covariance matrix of the Wigner representation as described before, so it is PSD. One can say that the CCR inequality $\Sigma_N \geq \pm\frac{i}{2}\Delta$ is tighter than the PSD inequality.
}
\fi
Then on this region $\Sigma \geq X^2$, the term $\Sigma - X^2$ is a concave and decreasing function of $X\in S^2_{+}$ in the generalized matrix inequality sense \cite{boyd2004convex}. We also have that $\det{Y}^{1/n}$ is concave for $Y \in S_{++}^n$ \cite[Exercise 3.18]{boyd2004convex}. Applying the composition rule, the term inside the $g(.)$ function is convex. Then by the fact that $-g(.)$ is a convex function, it follows that the optimization problem is convex and can be solved using  KKT conditions. The Lagrangian function is  
    $L(X) = \Theta(X)  - \Tr{\Lambda X^2 } + \Tr{\Pi X^2} -  \pi \sp{X}$. 
\ifmycomments
{\color{blue}
Taking the derivative of the objective function gives
\begin{align}
    \frac{\partial \Theta(X)}{ \partial X} =
    %\log \frac{\phi(X)+ 1/2}{\phi(X) - 1/2} \cdot \frac{1}{\sqrt{\det{\Sigma_Z}}} \cdot \frac{1}{2 \sqrt{\det{\Sigma - X^2}}} \cdot \frac{\partial \det{\Sigma - X^2}}{ \partial X} \\
     \frac{1}{2} \phi(X) \cdot  \log \frac{\phi(X)+ 1/2}{\phi(X) - 1/2}  \cdot \frac{1}{\det{\Sigma - X^2}} \cdot  \left[ \text{adj}(\Sigma - X^2) \cdot X + X \cdot \text{adj}(\Sigma - X^2) \right],
\end{align}
where $\phi(X):= \sqrt{\frac{\det{\Sigma - X^2}}{\det{\Sigma_Z}}}$ and $\text{adj}(X)$ is the adjugate of the matrix.
}
\fi
We develop the KKT conditions,
\begin{align}
    %line 1
    \frac{1}{2}\phi(X) \cdot \log\frac{\phi(X) + 1/2}{\phi(X) - 1/2} \left[(\Sigma - X^2)^{-1} X + X (\Sigma - X^2)^{-1}\right]  - \Lambda X - X\Lambda + \Pi X + X \Pi 
    &= \pi I,\\
    % line 2
    \Lambda X^2 = 0, \quad \Lambda &\geq 0, \\
    % line 3
    \Pi \left(X^2 - \Sigma \pm \frac{i}{2} \Sigma_Z^{1/2} \Delta \Sigma_Z^{1/2} \right) = 0, \quad \Pi &\geq 0, \label{eq:Pi-equality}\\
    % line 4
    \pi (\sp{X} - c ) = 0, \quad \pi &\geq 0,
\end{align}
where $\phi(X):= \sqrt{\frac{\det{\Sigma - X^2}}{\det{\Sigma_Z}}}$. The constraint \eqref{eq:X_constraint} is equivalent to the following two inequalities, $\Tr{X^2} \leq \Tr{\Sigma}$, and $\det{\Sigma - X^2 } \geq \frac{1}{4} \det{\Sigma_Z}$.
\ifmycomments
{\color{blue}
We can prove that by simply expanding the inequality as follows:
\begin{align}
    \Sigma - X^2 \pm \Sigma_Z^{1/2} \Delta \Sigma_Z^{1/2} \geq 0 
\end{align}
Then by defining $M := \Sigma - X^2$, by the element parameters 
\begin{align}
    M = \begin{bmatrix}
        a & c\\
        c& b 
    \end{bmatrix}
\end{align}
the eigenvalues of the $\Sigma - X^2 \pm \Sigma_Z^{1/2} \Delta \Sigma_Z^{1/2}$ are
\begin{align}
    \lambda_{1,2} = \frac{1}{2} \left[  \sp{M} \pm \sqrt{ \sp{M}^2 - 4\det{M} + \det{\Sigma_Z}}  \right].
\end{align}
}
\fi
Note that the second inequality above is the feasibility condition of the objective function $\Theta(X)$ and is always satisfied inside the domain. Also, the first inequality is satisfied if $X^2 \leq \Sigma$. Therefore, under this condition, $\Pi=0$ in the constraint \eqref{eq:Pi-equality}.

Next we restrict ourself to the solutions of the form $X = U \text{diag}(x_i) U^T$, where $X$ has the same diagonalization as $\Sigma = U \text{diag}(\sigma_i^2)_{i=1:2} U^T$, with $0 \leq x_i \leq {\sigma_i}$. 
This reduces the KKT conditions to the following system of equations:
\begin{align}
    \phi(X) \cdot \log\frac{\phi(X) + 1/2}{\phi(X) - 1/2} \frac{x_i}{\sigma_i^2 - x_i^2}&= \pi, \quad \text{for }i = 1, 2 \label{eq:pi_xi_equation}\\
    \sum_i x_i &= c, \quad \pi \geq 0.
\end{align}
The term $\phi(X)$ does not depend on the index $i$, and therefore, it is just a scaling of $\pi$ which can be absorbed in 
$\tilde \pi = \pi \left(\phi(X) \log\frac{\phi(X) + 1/2}{\phi(X) - 1/2}    \right)^{-1}$.
This gives the final solution by the system of equations in \eqref{eq:ratedist_onemode_solution}.
% \begin{align} 
%        \frac{x_i}{\sigma_i^2 - x_i^2} &= \tilde \pi,     \quad \text{for }i = 1, 2  \\
%     \sum_i x_i &= c, \quad \tilde \pi \geq 0.
% \end{align}
It suffices to show that for any given $c_\text{min} < c < c_\text{max}$, there exists a solution in the form of the equations in \eqref{eq:ratedist_onemode_solution}. 
Isolating the equality for $x_i$ gives $x_i = -\frac{1}{2 \tilde \pi} + \sqrt{(\frac{1}{2 \tilde \pi})^2 + \sigma_i^2}$, for $i = 1,2$,
which is always valid inside the domain of $x_i$. Then by defining $y:= 1/(2\tilde \pi)$,  we rewrite the second equation as a function of $y$,
\begin{align} \label{eq:pi_c_formulation}
    f(y) := -n y + \sum_{i=1}^2 \sqrt{y^2 + \sigma_i^2} = c.
\end{align}
Recall that $c = D_\text{max} - D$ and $D\in (D_\text {OT}, D_\text {max})$. Therefore, $c_\text{min} = 0$. The value of  $c_\text{max}$ corresponds to $D_\text {OT}$ which is the Wasserstein distance of the QC system, implied by the feasibility constraint \eqref{eq:X_constraint}.

We next provide a loose upper bound $\hat c_\text{max} >> c_\text{max}$  by relaxing the constraint to $X^2 \leq \Sigma $. Therefore, ${\hat c}_\text{max} = \sp{\Sigma^{1/2}} = \sum_i \sigma_i$. 
As $c_\text{min}=0$ results in $y = \infty$ and $c_\text{max} = \sum_i \sigma_i$ results in $y = 0$, then both have valid solutions. Further, $ f(y)$ function is a monotonically decreasing function of $y$, therefore, there is a one-to-one correspondence between values of $c\in (c_\text{min}, \hat c_\text{max})$ and the values of $y\in (0,\infty)$. Thus, for any feasible value of $c \in (c_\text{min}, {\hat c}_\text{max})$, there exists a solution of the form \eqref{eq:ratedist_onemode_solution}. This approach is similar to the solution in the classical case proved in \cite{chenjun_unpublished_privatecommuncation}.

Next, the following optimization problem obtains the Wasserstein distance $D_{OT}$,
 \begin{align}
     D_{OT} = \min_{X'} D_\text{max} - \sp{X'}, \quad \text{such that:} \; \pm \frac{i}{2} \Delta \leq \Sigma_N \leq \Sigma_\rho .
 \end{align}
 For the case of one-mode, similar to the approach above, by setting $X = U \text{diag}[y_1, y_2]U^T$, this optimization problem reduces to maximizing $\sum_i y_i$ subject to constraint $\prod_{i=1}^2 (\sigma_i^2  - y_i^2) \geq \frac{1}{4} \det{\Sigma_Z}$ which results in \eqref{eq:gaussian_wasserstein_syseq}.

 % Reduces to the following system of equations providing the solution:

 % This again is a convex function that has the following solution:
 % \begin{align}
 %     \begin{cases}
 %         \frac{1}{2\lambda} = y_2 (y_1^2 - \sigma_1^2) = y_1 (y_2^2 - \sigma_2^2),\\
 %         (y_1^2 - \sigma_1^2)(y_2^2 - \sigma_2^2) = \frac{1}{4} \det{\Sigma_Z}.
 %     \end{cases}
 % \end{align}

\section{Proof of Cardinality Bound} \label{appE}
% \subsection{Cardinality Bound} \label{sec:cardinality}
Here, we provide a cardinality bound for the sample space of the intermediate system $\~W$. We start by introducing the QC version of the classical support lemma \cite{elgamal2011network}:
\begin{lemma} \textbf{(Quantum-Classical Support Lemma)} \label{lemma:support-lemma}
    Let $\~P \subseteq \~G(\~H_A)$ be a compact and connected  subset\footnote{The set of density operators is equipped with the weak operator topology and, by \cite[Lemma 11.1]{Holevo_2019_quantum_systems_book}, this topology coincides with the trace norm topology on this set.} of density operators on finite-dimensional Hilbert space $\~H_A$, and let $\~W$ be an arbitrary set. Further, let $\{ \rho_w \}_{w\in\~W} \in \~P$ be a collection of (conditional) density operators. Let $g_j(\rho), j=1,...,J$ be a collection of real-valued continuous functions of quantum states $\~P$. For every random variable $W$ with probability measure $F(w)$ defined on $\~W$, there exist a random variable $W'$ defined on $\~W'$ with PMF $P_{W'}$ with $|\mathcal{W}' | \leq J$ and a collection of conditional density operators $\tilde \rho_{w'} \in \~P$ indexed by $w' \in \~W'$ such that:
    \begin{align}
        \int_\~W g_j(\rho_w) dF(w) = \sum_{w'\in \~W'} g_j(\tilde\rho_{w'}) P_{W'}(w'), \quad \text{for }j=1,...,J.
    \end{align}
\end{lemma}
\begin{proof}
    % \myans{The only thing to prove is that the subspace of density operators is }
    First note that the set of density operators $\~G(\~H_A)$ on finite-dimensional Hilbert space $\~H_A$ is compact under weak topology.  
    % , because it is defined over a complex space that is closed and bounded. 
    Further, it is connected because it is a convex set. Then the lemma follows directly from the Fenchel-Eggleston-Caratheodory Theorem \cite{elgamal2011network}.
\end{proof}

Assume having a composite QC state $\tau^{RWX} \in \~G(\~H_R \otimes \~H_W \otimes \~H_X)$ forming a quantum Markov chain $R - W - X$  with the PMR ensemble $\{\hat \rho_w, P_W(w)\}_{w\in \~W}$  with an arbitrary discrete alphabet $\~W$ as the intermediate outcomes, and the classical channel $P_{X|W}$ to the output space $\~X$. 
Consider the following functions of the composite states $\gamma^{RX} \in \~G(\~H_R \otimes \~H_X)$: 
\begin{align} \label{eq:supportlemmafunctions}
g_{j}(\gamma^{RX})=
    \begin{cases}
        \text{param}_{j}(\gamma^R), &\text{for }\; j = 1, \cdots, (\dim \~H_A)^2 - 1 \\
        \braket{x_{j'}|\gamma^X}{x_{j'}}, &\text{for }\;  j' = 1, \cdots, \abs{\~X}-1, \quad \And j = j' + (\dim \~H_A)^2 - 1\\
        H(R)_{\rho}-H(R)_\gamma,  &\text{for }\;  j =  \abs{\~X}+(\dim \~H_A)^2 - 1\\
        H(X)_{Q_X}-H(X)_\gamma, &\text{for }\;  j =  \abs{\~X}+(\dim \~H_A)^2\\
        Tr[\Delta^{RX} \tau^{RX}], &\text{for }\;  j =  \abs{\~X}+(\dim \~H_A)^2 + 1.
    \end{cases}
\end{align}
where $\text{param}_{j}(\sigma)$ is the function returning the $j$-th parameter of the density operator $\sigma$.
Note that a density operator $\rho \in \~G(\~H_A)$ can be expressed on an arbitrary orthonormal basis $\{\ket{i}\}_{i=1: d}$ as 
    $\rho = \sum_{i,j}^d \alpha_{i,j} \ketbra{i}{j},$
where $d$ is the rank of the density operator. The coefficients  $\alpha_{i,j} \in \mathbb{C}$ for $i\neq j$, and such that $\alpha_{i,j} = \alpha^*_{j,i}$ and, $\alpha_{i,i} \in \mathbb{R}$ such that $0 \leq \alpha_{i,i } \leq 1$, for $i=1,...,d$, and $\sum_i \alpha_{i,i} = 1$. %(also see \cite{devetak_2011_quantumbroadcastchannels}).
% Thus, the density operator $\sigma$ can be represented by at most $ (\dim \~H_A)^2 - 1$ real parameters, denoted by $\text{param}_j(\sigma)$. (Also in \cite{devetak_2011_quantumbroadcastchannels}). 

By directly applying the QC support lemma \ref{lemma:support-lemma} to the set of functions \eqref{eq:supportlemmafunctions}, we obtain that there exists a $\tilde{\~W}$ random variable such that $    \sum_w g_j(\tau^{RX}_w) P_W(w) = \sum_{\tilde w} g_j(\nu^{RX}_{\tilde w}) \tilde P_{\tilde W}(\tilde w)$ for all $j = 1,\cdots, (\dim \~H_A)^2 - 1 $,
where $\tau_w^{RX}$ and $\nu^{RX}_w$ are the conditional composite states given the measurement outcome $w$, respectively. 
This proves the desired cardinality bound.

\ifmycomments
{\color{blue}
This proves the existence of a random variable $W'$ with cardinality at most 
 $|\~W'| \leq (\dim \~H_A)^2+|\~X|+ 1$  ,
with a composite state $\nu^{RW'X}$ which forms the quantum Markov chain $R - W' - X$ that satisfies the equalities:
\begin{align}
    \nu^X = \tau^X \equiv Q_X, \ \ 
    \nu^R = \tau^R \equiv \rho, & \ \ 
    \Tr{\Delta_{RX} \nu^{RX}} = \Tr{\Delta_{RX} \tau^{RX}},\\
    I(X;W')_\nu = I(X;W)_\tau, & \ \ 
    I(R;W')_\nu = I(R;W)_\tau .
\end{align}
}
\fi

% For the output distribution and output conditional entropy, the following conditional representations are straightforward as $W$ and $X$ are classical quantum registers. Thus, there is at most need for $|\~X|-1$ separate functions to represent the distribution $Q_X$ and one for the conditional entropy:
% \begin{align}
%     Q_X(x_i) &= \sum_w Q_{X|W}(x_i|w)P_W(w), \qquad\qquad\forall i\in [|\~X|-1], x_i\in \~X\\
%     H(X|W) &= \sum_w H(X|W=w) P_W(w).
% \end{align}
% Also, the other functions have the following representations,
% \begin{align}
%     % first
%     H(R|W)&= \sum_w H(\hat \rho_w) P_W(w),\\
%     % second
%     \Tr{\Delta_{RX} \tau^{RX}} 
%     % &= \Tr{\Delta_{RX} \left(    \sum_{w,x} \Tr_A \{(\text{id}_R \otimes \Lambda_w)\Psi_\rho^{RA}\}\otimes P_{X|W}(x|w) \ketbra{x}{x}^X    \right)} \nonumber\\
%     % &=\Tr{\Delta_{RX} \left(    \sum_{w,x} \rho^A_w \, P_W(w) \otimes P_{X|W}(x|w) \ketbra{x}{x}^X    \right)} \nonumber\\
%     &=\sum_w P_W(w)\Tr{\Delta_{RX} \left(    \sum_{x} \hat \rho^A_w \,  \otimes P_{X|W}(x|w) \ketbra{x}{x}^X    \right)} \nonumber\\
%     &=\sum_w \Tr{\Delta_{RX} \tau^{RX}|W=w}P_W(w).
% \end{align}
% Therefore, the QC support lemma \ref{lemma:support-lemma}, proves the existence of a random variable $W'$ with cardinality at most 
%  $|\~W'| \leq (\dim \~H_A)^2+|\~X|+ 1$  ,
% with a composite state $\nu^{RW'X}$ which forms the quantum Markov chain $R - W' - X$ with the same marginals and satisfies the entropic  and distortion equalities,

\ifmycomments
\myqu{Is this correct? If yes, why didn't Wilde use this simple argument in his paper?}

\myqu{What to do about the continuous case?}
\fi

\end{appendices}
% --------------------
% %%%%%%%%%%%%%%%%%%%%%%%%%%%%%%%%%%%%%%%%%%%%%%%%%%%%%%%%%%
% %%%%%%%%%%%%%%%%%%%%%%%%%%%%%%%%%%%%%%%%%%%%%%%%%%%%%%%%%%
% REFERENCES SECTION
% %%%%%%%%%%%%%%%%%%%%%%%%%%%%%%%%%%%%%%%%%%%%%%%%%%%%%%%%%%
% %%%%%%%%%%%%%%%%%%%%%%%%%%%%%%%%%%%%%%%%%%%%%%%%%%%%%%%%%%
\medskip
\bibliography{references_first.bib} 

% --------------------
\end{document}